\documentclass[11pt]{article}

\usepackage[T1]{fontenc}
\usepackage{tgpagella}

\usepackage{graphicx,amsmath,amsfonts,amscd,amssymb,bm,url,color,latexsym}
\usepackage{fullpage}
\usepackage[small,bf]{caption}
\usepackage{subcaption}
\usepackage{bbm}
\usepackage{microtype}
\usepackage{algorithm}
\usepackage{algorithmicx}
\usepackage{algpseudocode}
\usepackage{hyperref}
\usepackage{verbatim}
\usepackage{framed}
\usepackage{comment}
\usepackage{tikz}
\usepackage{fge}
\usepackage[margin=0.9in]{geometry}
\usepackage{enumitem}
\usepackage[numbers,sort&compress]{natbib}
\usepackage{cleveref}

\allowdisplaybreaks

\hypersetup{
    colorlinks=true,%
    citecolor=blue,%
    filecolor=blue,%
    linkcolor=blue,%
    urlcolor=blue
}
\usepackage[toc, page]{appendix}

\makeatletter
\def\@IEEEsectpunct{\ \,}
\def\paragraph{\@startsection{paragraph}{4}{\z@}{1.5ex plus 1.5ex minus 0.5ex}%
{0ex}{\normalfont\normalsize\textbf}}
\makeatother

\newtheorem{theorem}{Theorem}[section]
\newtheorem{lemma}[theorem]{Lemma}
\newtheorem{corollary}[theorem]{Corollary}
\newtheorem{proposition}[theorem]{Proposition}
\newtheorem{definition}[theorem]{Definition}

\newcommand{\indep}{\rotatebox[origin=c]{90}{$\models$}}

\renewcommand{\mathbf}{\boldsymbol}

\newcommand{\conv}{\circledast}
\newcommand{\mb}{\mathbf}
\newcommand{\mc}{\mathcal}

\newcommand{\bb}{\mathbb}

\newcommand{\reals}{\bb R}

\newcommand{\eps}{\varepsilon}
\newcommand{\R}{\reals}
\newcommand{\Cp}{\bb C}

\newcommand{\N}{\bb N}

\newcommand{\indicator}[1]{\mathbbm 1_{#1}}

\newcommand{ \Brac }[1]{\left\lbrace #1 \right\rbrace}
\newcommand{ \brac }[1]{\left[ #1 \right]}
\newcommand{ \paren }[1]{ \left( #1 \right) }
\newcommand{\col}{\mathrm{col}}
\DeclareMathOperator{\dist}{dist}
\DeclareMathOperator{\poly}{poly}

\DeclareMathOperator{\supp}{supp}

\DeclareMathOperator{\diag}{diag}
\DeclareMathOperator{\sign}{sign}

\newcommand{\im}{\mathrm{i}}

\newcommand{\wh}{\widehat}
\newcommand{\wt}{\widetilde}
\newcommand{\ol}{\overline}

\newcommand{\norm}[2]{\left\| #1 \right\|_{#2}}
\newcommand{\abs}[1]{\left| #1 \right|}
\newcommand{\innerprod}[2]{\left\langle #1,  #2 \right\rangle}

\newcommand{\edited}[1]{{\color{black}{#1}}}

\numberwithin{equation}{section}

\def \endprf{\hfill {\vrule height6pt width6pt depth0pt}\medskip}
\newenvironment{proof}{\noindent {\bf Proof} }{\endprf\par}

\begin{document}
\title{Convolutional Phase Retrieval via Gradient Descent}
\author{Qing Qu, Yuqian Zhang, Yonina C. Eldar, and John Wright \thanks{\edited{QQ is with Center for Data Science, New York University (Email: \href{mailto:qq213@nyu.edu}{\color{blue}{qq213@nyu.edu}}); YZ is with Department of Computer Science, Cornell University (Email: \href{mailto:yz2557@cornell.edu}{\color{blue}{yz2557@cornell.edu}}); YE is with the Weizmann Faculty of Mathematics and Computer Science, Weizmann Institute of Science, Rehovot, Israel (Email: \href{mailto:yonina.eldar@weizmann.ac.il}{\color{blue}{yonina.eldar@weizmann.ac.il}}); JW is with Department of Electrical Engineering, Department of Applied Physics and Applied Mathematics, and Data Science Institute at Columbia University (Email: \href{mailto:jw2966@columbia.edu}{\color{blue}{jw2966@columbia.edu}}). Most of the work has been conducted when QQ and YZ were with Department of Electrical Engineering and Data Science Institute at Columbia University.} An extended abstract of the current work has appeared in NeurIPS'17~ \cite{qu2017convolutional}.}
}
%\markboth{IEEE Transaction on Information Theory,~Vol.~xx, No.~xx,~2019}%
%{Qu \MakeLowercase{\textit{et al.}}: Convolutional Phase Retrieval via Gradient Descent}
\maketitle

\begin{abstract}
We study the convolutional phase retrieval problem, of recovering an unknown signal $\mathbf x \in \mathbb C^n $ from $m$ measurements consisting of the magnitude of its cyclic convolution with a given kernel $\mathbf a \in \mathbb C^m $. This model is motivated by applications such as channel estimation, optics, and underwater acoustic communication, where the signal of interest is acted on by a given channel/filter, and phase information is difficult or impossible to acquire. We show that when $\mathbf a$ is random and the number of observations $m$ is sufficiently large, with high probability $\mathbf x$ can be efficiently recovered up to a global phase shift using a combination of spectral initialization and generalized gradient descent. The main challenge is coping with dependencies in the measurement operator. We overcome this challenge by using ideas from decoupling theory, suprema of chaos processes and the restricted isometry property of random circulant matrices, and recent analysis of alternating minimization methods.
\end{abstract}

%\begin{IEEEkeywords}
%Phase Retrieval, Nonconvex optimization, Inverse problems, Structured measurements, Nonlinear approximation, 
%\end{IEEEkeywords}

\section{Introduction}
% !TEX root = ../PR Toeplitz.tex

We consider the problem of \emph{convolutional phase retrieval} where our goal is to recover an unknown signal $\mb x \in \bb C^n$ from the magnitude of its cyclic convolution with a given filter $\mb a \in \bb C^m$. Specifically, the measurements take the form 
\begin{align}\label{eqn:pr-conv}
   \mb y \;=\; \abs{ \mb a \conv \mb x },	
\end{align}
where $\conv$ is cyclic convolution modulo $m$ and $\abs{ \cdot }$ denotes entrywise absolute value. This problem can be rewritten in the common matrix-vector form\footnote{The convolution can be written in matrix-vector form as $\mb a \conv \mb x = \mb C_{\mb a} \iota_{n \rightarrow m } \mb x$, where $\mb C_{\mb a}\in \bb R^{m \times m}$ denotes the circulant matrix generated by $\mb a$ and $\iota_{n \rightarrow m }$ is a zero padding operator. In other words, $\mb a \conv \mb x  = \mb A \mb x $ with $\mb A \in \bb C^{m \times n}$ being a matrix formed by the first $n$ columns of $\mb C_{\mb a}$.}
\begin{align}\label{eqn:pr-conv-1}
  \text{find}\quad \mb z, \quad \text{s.t.}\quad  \mb y = \abs{\mb A \mb z}.
\end{align}
This problem is motivated by applications in areas such as channel estimation \cite{Walk:asilomar2015}, noncoherent optical communication \cite{gagliardi1976optical}, and underwater acoustic communication \cite{stojanovic1994phase}. For example, in millimeter-wave (mm-wave) wireless communications for 5G networks \cite{shahmansoori20155g}, one important problem is to estimate the angle of arrival (AoA) of a signal from measurements taken by the convolution of an antenna pattern and AoAs. As the phase measurements are often very noisy, unreliable, and expensive to acquire, it may be preferred to only take measurements of signal magnitude in which case the phase information is lost. 
%\begin{figure*}[!htbp]
%\centering
%\includegraphics[width = \linewidth]{figures/abstract.eps}
%\caption{Applications of convolutional phase retrieval. \jw{the captions are generally not sufficiently informative. A general rule is that the caption should provide enough information to interpret the figure without reference to the description in the text.} \jw{Can you help me to understand the extent to which this figure represents applications of convolutional phase retrieval? I thought that in communications applications, either the filter $\mb a$ or the unknown $\mb x$ represent the impulse response of the channel, and so they live in between the transmitter and the amplitude receiver. But my understanding could be wrong.} \jw{the figure is a little bit sloppy. The font size of $\mb x$ and $\mb a$ are different; the font is different from the font of the text. It is a little bit unprincipled for it to be my joy to notice this and complain -- by the time I complain about it, it is just me being picky. But you may find it useful for yourself to be attuned to quality in this way. Parenthetically, I find tikz useful for making figures with embedded math.}}
%\label{fig:abstract}
%\end{figure*}

%\qq{It would be better if we can insert 1-2 concrete examples of convolutional phase retrieval for motivation purposes of this work here, (better) with figures for illustration purposes. Yonina, as you are more familiar with the applications, would you mind helping writing up here? Many thanks!}

\edited{Most known results on the exact solution of phase retrieval problems  \cite{candes2013phaselift,soltanolkotabi2014algorithms,chen2015solving,wang2016solving,waldspurger2015phase,waldspurger2016phase} pertain to \emph{generic random matrices}, where the entries of $\mb A$ are independent subgaussian random variables. We term problem \eqref{eqn:pr-conv-1} with a generic sensing matrix $\mb A$ as \emph{generalized phase retrieval}}\footnote{\edited{We use this term to make distinctions from Fourier phase retrieval, where the matrix $\mb A$ is an oversampled DFT matrix. Here, generalized phase retrieval refers to the problem with any generic measurement other than Fourier.} }. However, in practice it is difficult to implement purely random measurement matrices. In most applications, the measurement is much more structured -- the convolutional model studied here is one such structured measurement operator. Moreover, structured measurements often admit more efficient numerical methods: by using the \emph{fast Fourier transform} for matrix-vector products, the benign structure of the convolutional model \eqref{eqn:pr-conv} allows to design methods with $\mc O(m)$ memory and $\mc O(m \log m)$ computation cost per iteration. In contrast, for generic measurements, the cost is around $\mc O(mn)$.

In this work, we study the convolutional phase retrieval problem \eqref{eqn:pr-conv} under the assumption that the kernel $\mb a = \brac{a_1,\cdots,a_m}^\top $ is randomly generated from an i.i.d. \emph{standard complex Gaussian distribution} $\mc {CN}\paren{\mb 0,\mb I}$, i.e.,
\begin{align}\label{eqn:complex-Gaussian}
	\edited{\mb a \sim \mc {CN}(\mb 0,\mb I), \quad \text{if}\quad  \mb a \;=\; \mb u + \im \mb v, \quad \mb u, \mb v \sim_{\mathrm{i.i.d.}} \mc N \left(\mb 0, \tfrac{1}{2} \mb I \right).}
\end{align}
Compared to the generalized phase retrieval problem, the random convolution model \eqref{eqn:pr-conv} we study here is far more structured: it is parameterized by only $\mc O(m)$ independent complex normal random variables, whereas the generic model involves $\mc O(mn)$ random variables. This extra structure poses significant challenges for analysis: the rows and columns of the sensing matrix $\mb A$ are probabilistically dependent, so that classical probability tools (based on concentration of functions of independent random vectors) do not apply. 

We propose and analyze a local gradient descent type method, minimizing a weighted, \emph{nonconvex} and \emph{nonsmooth} objective 
\begin{align}\label{eqn:pr-weighted}
   \min_{\mb z \in \bb C^n} f(\mb z) = \frac{1}{2m}  \norm{ \mb b^{1/2} \odot \paren{ \mb y - \abs{\mb A\mb z} }  }{}^2,
\end{align}
where $\odot$ denotes the Hadamard product. Here, $\mb b \in \bb R_{++}^m$ is a weighting vector, which is introduced mainly for analysis purposes. The choice of $\mb b$ is discussed in \Cref{sec:analysis}. Our result can be informally summarized as follows.

\begin{framed}
\centering 
    With $m \geq   \Omega \paren{\frac{ \norm{\mb C_{\mb x}}{}^2}{ \norm{\mb x}{}^2  } n \poly \log n }$ samples, generalized gradient descent starting from a data-driven initialization converges \emph{linearly} to $\mb x$ up to a global phase.
\end{framed}

%\begin{theorem}[Informal]
%	When $m \geq \Omega(n \poly \log n)$, with high probability, spectral initialization \cite{netrapalli2013phase,candes2015phase} produces an initialization $\mb z^{(0)}$ that is $\mc O(1/\poly \log n)$ close to the optimum. Moreover, when $m \geq   \Omega \paren{\frac{ \norm{\mb C_{\mb x}}{}^2}{ \norm{\mb x}{}^2  } n \poly \log n }$, with high probability, a certain gradient descent method optimizing \eqref{eqn:pr-weighted} converges linearly from this initialization to the optimal set $\mc X = \set{ \mb x e^{\im \phi} \mid \phi \in [0,2\pi) }$ of points that differ from the true signal $\mb x$ only by a global phase. 
%\end{theorem}

\vspace{.1in}

Here, $\mb C_{\mb x}\in \bb C^{m\times m}$ denotes the circulant matrix corresponding to cyclic convolution with a length $m$ zero padding of $\mb x$, and $\poly \log n$ denotes a polynomial in $\log n$. Compared to the results of \edited{generalized phase retrieval with i.i.d. Gaussian measurement}, the sample complexity $m$ here has extra dependency on $\norm{\mb C_{\mb x}}{}/\norm{\mb x}{}$. The operator norm $\norm{\mb C_{\mb x}}{}$ is inhomogeneous over $\bb {CS}^{n-1}$: for a typical\footnote{e.g., $\mb x$ is drawn uniformly at random from $\bb {CS}^{n-1}$.} $\mb x \in \bb {CS}^{n-1}$, $\norm{\mb C_{\mb x}}{}$ is of the order $\mc O(\log n)$ and the sample complexity matches that of the generalized phase retrieval up to $\log$ factors; the ``bad'' case is when $\mb x$ is \emph{sparse} in the Fourier domain: $\norm{\mb C_{\mb x}}{} \sim \mc O(\sqrt{n})$ and $m$ can be as large as $\mc O(n^2 \poly \log n)$. 

Our proof is based on ideas from \emph{decoupling theory} \cite{de1999decoupling}, the \emph{suprema of chaos processes} and \emph{restricted isometry property} of random circulant matrices \cite{rauhut2010compressive,krahmer2014suprema}, and is also inspired by a new iterative analysis of alternating minimization methods \cite{waldspurger2016phase}. Our analysis draws connections between the convergence properties of gradient descent and the classical alternating direction method. This allows us to avoid the need to argue uniform concentration of high-degree polynomials in the structured random matrix $\mb A$, as would be required by a straightforward translation of existing analysis to this new setting. Instead, we control the bulk effect of phase errors uniformly in a neighborhood around the ground truth. This requires us to develop new decoupling and concentration tools for controlling nonlinear phase functions of circulant random matrices, which could be potentially useful for analyzing other random circulant convolution problems, such as sparse blind deconvolution \cite{zhang2017global} and convolutional dictionary learning \cite{heide2015fast}.

\subsection{Comparison with literature}

\paragraph{Prior arts on phase retrieval.}~The challenge of developing efficient, guaranteed methods for phase retrieval has attracted substantial interest over the past several decades \cite{shechtman2015phase,jaganathan2015phase}. The problem is motivated by applications such as X-ray crystallography \cite{millane1990phase,robert1993phase}, microscopy \cite{jianwei2002high}, astronomy \cite{fienup1987phase}, diffraction and array imaging \cite{bunk2007diffractive,anwei2011array}, optics \cite{walther1963question}, and more. The most classical method is the \emph{error reduction algorithm} derived by Gerchberg and Saxton \cite{gerchberg1972practical}, also known as the alternating direction method. This approach has been further improved by the \emph{hybrid input-output} (HIO) algorithm \cite{fienup1982phase}. For oversampled Fourier measurements, it often works surprisingly well in practice, while its global convergence properties still largely remains as a mystery \cite{pauwels2017fienup}. 

For the generalized phase retrieval where the sensing matrix $\mb A$ is \edited{i.i.d. Gaussian}, the problem is better-studied: in many cases, when the number of measurements is large enough, the target solution can be exactly recovered by using either convex or nonconvex optimization methods. The first theoretical guarantees for global recovery of generalized phase retrieval with i.i.d. Gaussian measurement are based on convex optimization -- the so-called \emph{Phaselift/Phasemax} methods \cite{candes2013phaselift,candes2013phase-matrix,waldspurger2015phase}. These methods lift the problem to a higher dimension and solve a semi-definite programming (SDP) problem. However, the high computational cost of SDP limits their practicality. Quite recently, \cite{bahmani2016phase,goldstein2016phasemax,hand2016elementary} reveal that the problem can also be solved in the natural parameter space via linear programming.

\edited{Recently, nonconvex approaches have led to new computational guarantees for global optimizations of generalized phase retrieval.} \edited{The first result of this type is due to Netrapalli et al. \cite{netrapalli2013phase}, showing} that the alternating minimization method provably converges to the truth when initialized using a spectral method and provided with fresh samples at each iteration. Later on, Cand\`{e}s et al. \cite{candes2015phase} showed that with the same initialization, gradient descent for the nonconvex least squares objective, 
\begin{align}\label{eqn:pr-smooth}
	\min_{\mb z\in \bb C^n} f_1(\mb z) = \frac{1}{2m} \norm{ \mb y^2  - \abs{\mb A\mb z}^2 }{}^2,
\end{align}
provably recovers the ground truth, with near-optimal sample complexity $m \geq \Omega(n\log n)$. The subsequent work \cite{chen2015solving,zhang2016reshaped,wang2016solving} further reduced the sample complexity to $m \geq \Omega(n)$ by using different nonconvex objectives and truncation techniques. In particular, recent work by \cite{zhang2016reshaped,wang2016solving} studied a nonsmooth objective that is similar to ours \eqref{eqn:pr-weighted} with weighting $\mb b = \mb 1$. Compared to the SDP-based techniques, these methods are more scalable and closer to the approaches used in practice. Moreover, Sun et.\ al.\ \cite{sun2016geometric} reveal that the nonconvex objective \eqref{eqn:pr-smooth} actually has a benign \emph{global geometry}: with high probability, it has no bad critical points with $m \geq \Omega(n \log^3n)$ samples\footnote{\cite{soltanolkotabi2017} further tightened the sample complexity to $m \geq \Omega(n \log n)$ by using more advanced probability tools.}. Such a result enables initialization-free nonconvex recovery\footnote{For convolutional phase retrieval, it would be nicer to characterize the global geometry of the problem as in \cite{ge2015escaping,sun2015nonconvex,sun2016geometric,sun2016complete_a,sun2016complete_b}. However, the inhomogeneity of $\norm{\mb C_{\mb x}}{}$ over $\bb {CS}^{n-1}$ causes tremendous difficulties for concentration with $m \geq \Omega(n\poly \log n)$ samples.} \cite{chen2018gradient,gilboa2018efficient}.

\paragraph{Structured random measurements.}~The study of structured random measurements in signal processing has quite a long history \cite{krahmer2014structured}. For compressed sensing \cite{candes2006robust}, the work \cite{candes2006stable,candes2006near,eldar2012compressed} studied random Fourier measurements, and later \cite{rauhut2010compressive,krahmer2014suprema} proved similar results for partial random convolution measurements. However, the study of structured random measurements for phase retrieval is still quite limited. In particular, \cite{david2013partial} and \cite{candes2015diffraction} studied t-designs and coded diffraction patterns (i.e., random masked Fourier measurements) using semidefinite programming. Recent work studied nonconvex optimization using coded diffraction patterns \cite{candes2015phase} and STFT measurements \cite{bendory2016non}, both of which minimize a nonconvex objective similar to \eqref{eqn:pr-smooth}. These measurement models are motivated by different applications. For instance, coded diffraction is designed for imaging applications such as X-ray diffraction imaging, STFT can be applied to frequency resolved optical gating \cite{tsang1996frequency} and some speech processing tasks \cite{lim1979enhancement}. Both of the results show iterative contraction in a region that is at most $\mc O(1/\sqrt{n})$-close to the optimum. Unfortunately, for both results either the radius of the contraction region is not large enough for initialization to reach, or they require extra artificial technique such as resampling the data. In comparison, the contraction region we show for the random convolutional model is larger $\mc O(1/\mathrm{polylog}(n))$, which is achievable in the initialization stage via the spectral method. For a more detailed review of this subject, we refer the readers to Section 4 of \cite{krahmer2014structured}.

The convolutional measurement can also be reviewed as a single masked coded diffraction patterns \cite{candes2015diffraction,candes2015phase}, since $\mb a \conv \mb x = F^{-1} (  \wh{\mb a} \odot \wh{\mb x} )$, where $\wh{\mb a}$ is the Fourier transform of $\mb a$ and $\wh{\mb x}$ is the oversampled Fourier transform of $\mb x$. The sample complexity for coded diffraction patterns $m \geq \Omega(n \log^4n)$ in \cite{candes2015phase} suggests that the dependence of our sample complexity on $\norm{\mb C_{\mb x}}{}$ for convolutional phase retrieval might not be necessary and can be improved in the future. On the other hand, our results suggest that the contraction region is larger than  $\mc O(1/\sqrt{n})$ for coded diffraction patterns, and resampling for initialization might not be necessary.

\subsection{Notations, Wirtinger Calculus, and Organizations}

\paragraph{Basic notations.}~We use $\mb C_{\mb a} \in \bb C^{m \times m}$ to denote a circulant matrix generated from $\mb a$, i.e.,
\begin{align}\label{eqn:matrix-C-a}
	\mb C_{\mb a} = \begin{bmatrix}
		a_1 & a_m & \cdots & a_3 & a_2 \\
		a_2 & a_1 & a_m & & a_3 \\
		\vdots & a_2 & a_1 & \ddots  & \vdots \\
		a_{m-1} & & \ddots & \ddots & a_m \\
		a_m & a_{m-1} & \cdots & a_2 & a_1
	\end{bmatrix} = \begin{bmatrix}
   s_0[\mb a] & s_1[\mb a]& \cdots & s_{m-1} [\mb a]
 \end{bmatrix},
\end{align}
where $s_\ell [\cdot]\;(0\leq \ell \leq m-1)$ denotes a circulant shift by $\ell$ samples. We use $ \bb {CS}^{n-1} $ to represent the unit complex sphere in $\bb C^n$, and let $\mc {CN}(\mb 0,\mb I)$ be the standard complex Gaussian distribution as introduced in \eqref{eqn:complex-Gaussian}. We use $(\cdot)^\top$ and $(\cdot)^*$ to denote the real and Hermitian transpose of a vector or matrix, respectively, and use $\Re(\cdot)$ and $\Im(\cdot)$ to denote the real and imaginary parts of a complex variable, respectively. We use $g_1\indep g_2$ to denote the independence of two random variables $g_1,\; g_2$. Given a matrix $\mb X\in \bb C^{m\times n}$, $\mathrm{col}(\mb X)$ and $\mathrm{row}(\mb X)$ are its column and row space. For any vector $\mb v \in \bb C^n$, we define
\begin{align*}
	\mb P_{\mb v} = \frac{\mb v \mb v^*}{ \norm{\mb v}{}^2 },\quad \mb P_{\mb v^\perp} = \mb I - \frac{ \mb v \mb v^* }{ \norm{\mb v}{}^2 }
\end{align*}
to be the projection onto the span of $\mb v$ and its orthogonal complement, respectively. We use $\norm{\cdot}{F}$ and $\norm{\cdot}{}$ to denote the Frobenius norm and spectral norm of a matrix, respectively. For a random variable $X$, its $L^p$ norm is defined as $\norm{X}{L^p} = \bb E\brac{ \abs{X}^p  }^{1/p}$
for any positive $p\geq 1$. For a smooth function $f\in \mc C^1$, its $L^\infty$ norm is defined as $ \norm{ f }{ L^\infty } = \sup_{t \in \text{dom}(f)} \abs{f(t)}$. For an arbitrary set $\Omega$, we use $\abs{\Omega}$ to denote the cardinality of $\Omega$, and use $\supp(\Omega)$ to denote the support set of $\Omega$. If $\mb 1_\Omega$ is the indicator function of the set $\Omega$, then
\begin{align*}
	[\mb 1_{\Omega}]_j = \begin{cases}
	1 & \text{if } j \in \Omega,\\
	0 & \text{otherwise,}
\end{cases}
\end{align*}
where $[\cdot]_j$ is the $j$th coordinate of a given vector. If $ \abs{\Omega}= \ell$, we use $\mb R_{\Omega}: \bb R^m \mapsto \bb R^\ell$ to denote a mapping that maps a vector into its coordinates restricted to the set $\Omega$. Let $\mb F_n \in \bb C^{n \times n}$ denote a unnormalized $n\times n$ Fourier matrix with $\norm{\mb F_n}{} = \sqrt{n}$, and let $\mb F_n^m \in \bb C^{m \times n} \;(m\geq n)$ be an oversampled Fourier matrix. \edited{For all theorems and proofs, we use $c_i$ and $C_i$ $(i=1,2,\cdots)$ to denote positive numerical constants.}

\paragraph{Wirtinger calculus.}~Consider a real-valued function $g(\mb z): \bb C^n \mapsto \bb R$. The function is not holomorphic, so that it is not complex differentiable unless it is constant \cite{ken2009complex}. However, if one identifies $\Cp^n$ with $\R^{2n}$ and treats $g$ as a function in the real domain, $g$ can be differentiable in the real sense. Doing calculus for $g$ directly in the real domain tends to produce cumbersome expressions. A more elegant way is adopting the Wirtinger calculus \cite{wirtinger1927formalen}, which can be considered as a neat way of organizing the real partial derivatives \edited{(see also \cite{soltanolkotabi2014algorithms} and Section 1 of \cite{sun2016geometric})}. The Wirtinger derivatives can be defined {\em formally} as 
\begin{align*}
	\frac{\partial g}{\partial \mb z} 
	\;& \doteq\; \left. \frac{\partial g(\mb z, \ol{\mb z}) }{\partial \mb z}\right |_{\ol{\mb z} \text{ constant} } = 
	\left.
	\brac{
	\frac{\partial g(\mb z, \ol{\mb z}) }{\partial z_1}, \dots, \frac{\partial g(\mb z, \ol{\mb z}) }{\partial z_n}
	}
	\right|_{\ol{\mb z} \text{ constant} }\\
	\frac{\partial g}{\partial \ol{\mb z}} 
	\;& \doteq\; \left. \frac{\partial g(\mb z, \ol{\mb z}) }{\partial \ol{\mb z}}\right |_{\mb z \text{ constant} } = 
	\left.
	\brac{
	\frac{\partial g(\mb z, \ol{\mb z}) }{\partial \ol{z_1}}, \dots, \frac{\partial g(\mb z, \ol{\mb z}) }{\partial \ol{z_n}}
	}
	\right|_{\mb z \text{ constant} }. 
\end{align*}
\edited{Basically it says that when evaluating $\partial g/\partial \mb z$, one just writes $\partial g/\partial \mb z$ in the pair of $(\mb z,\ol{\mb z})$, and conducts the calculus by treating $\ol{\mb z}$ as if it was a constant. We compute $\partial g/\partial \ol{\mb z}$ in a similar fashion.} To evaluate the individual partial derivatives, such as $\frac{\partial g(\mb z, \ol{\mb z}) }{\partial z_i}$, all the usual rules of calculus apply. \edited{For more details on Wirtinger calculus, we refer interested readers to \cite{ken2009complex}.} 

\paragraph{Organization.}~The rest of the paper is organized as follows. In \Cref{sec:alg}, we introduce the basic formulation of the problem and the proposed algorithm. In \Cref{sec:analysis}, we present the main results and \edited{a sketch of the proof}; detailed analysis is postponed to \Cref{sec:proofs}. In \Cref{sec:exp}, we corroborate our analysis with numerical experiments. We discuss the potential impacts of our work in \Cref{sec:discuss}. Finally, all the basic probability tools that are used in this paper are described in the appendices.

\section{Nonconvex Optimization via Gradient Descent}\label{sec:alg}
% !TEX root = ../PR Toeplitz.tex

In this work, we develop an approach to convolutional phase retrieval based on local nonconvex optimization. Our proposed algorithm has two components: (1) a careful data-driven initialization using a spectral method; (2) local refinement by gradient descent. We introduce the two steps below. 
\subsection{Minimization of a nonconvex and nonsmooth objective}
\edited{We consider minimizing a weighted \emph{nonconvex} and \emph{nonsmooth} objective introduced in \eqref{eqn:pr-weighted}.} The adoption of the positive weights $\mb b$ facilitates our analysis, by enabling us to compare certain functions of the dependent random matrix $\mb A$ to functions involving more independent random variables. We will substantiate this claim in the next section. \edited{As aforementioned, we consider the \emph{generalized} Wirtinger gradient of \eqref{eqn:pr-weighted},
\begin{align*}
  	\frac{\partial}{ \partial \mb z} f(\mb z) \;=\; \frac{1}{m} \mb A^*\diag\paren{  \mb b } \brac{ \mb A \mb z - \mb y \odot \exp\paren{ \im \phi(\mb A\mb z) }  }.
\end{align*}
Here, because of the nonsmoothness of \eqref{eqn:pr-weighted}, $f(\cdot)$ is not differentiable everywhere even in the real sense. To deal with this issue, we specify 
\begin{align*}
\exp\paren{ \im \phi(u) } \; \doteq \; \begin{cases}
 u/\abs{u} & \text{if} \abs{u}\not = 0, \\
 1 & \text{otherwise},	
 \end{cases}
\end{align*}
for any complex number $u \in \bb C$ and $\phi(u)\in [0,2\pi)$.} Starting from some initialization $\mb z^{(0)}$, we minimize the objective \eqref{eqn:pr-weighted} by generalized gradient descent
\begin{align}\label{eqn:grad-step-0}
   \mb z^{(r+1)} \;=\; \mb z^{(r)} - \tau \frac{\partial}{ \partial \mb z} f(\mb z^{(r)} ),
\end{align}
where $\tau>0$ is the stepsize. Indeed, $\frac{\partial}{ \partial \mb z} f(\mb z)$ can be interpreted as the subgradient of $f(\mb z)$ in the real case; this method can be seen as a variant of \emph{amplitude flow} \cite{wang2016solving}. 

\subsection{Initialization via spectral method}

\begin{algorithm}
\caption{Spectral Initialization}\label{alg:init}
\begin{algorithmic}[1]
\renewcommand{\algorithmicrequire}{\textbf{Input:}}
\renewcommand{\algorithmicensure}{\textbf{Output:}}
\Require~\
Observations $\Brac{y_k}_{k=1}^m$.
\Ensure~~\
The initial guess $\mb z^{(0)}$.

\State Estimate the norm of $\mb x$ by
\begin{align*}
	\lambda = \sqrt{\frac{1}{m} \sum_{k=1}^m y_k^2}
\end{align*}
\State Compute the leading eigenvector $\wt{\mb z}^{(0)} \in \bb {CS}^{n-1}$ of the matrix,
\begin{align*}
	\mb Y = \frac{1}{m}\sum_{k=1}^m y_k^2 \mb a_k \mb a_k^* = \frac{1}{m} \mb A^* \diag\paren{ \mb y^2} \mb A,
\end{align*}
\State Set $\mb z^{(0)} = \lambda \wt{\mb z}^{(0)}$.
\end{algorithmic}
\end{algorithm}

Similar to \cite{netrapalli2013phase,soltanolkotabi2014algorithms}, we compute the initialization $\mb z^{(0)}$ via a spectral method, detailed in Algorithm \ref{alg:init}. More specifically, $\mb z^{(0)}$ is a scaled version of the leading eigenvector of the following matrix
\begin{align}\label{eqn:Y-matrix}
	\mb Y \; =\; \frac{1}{m}\sum_{k=1}^m y_k^2 \mb a_k \mb a_k^* \;=\; \frac{1}{m} \mb A^* \diag\paren{ \mb y^2} \mb A,
\end{align}
which is constructed from the knowledge of the sensing vectors and observations. The leading eigenvector of $\mb Y$ can be efficiently computed via the power method. Note that $\bb E\brac{\mb Y} = \norm{\mb x}{}^2 \mb I + \mb x \mb x^*$, so the leading eigenvector of $\bb E\brac{\mb Y}$ is proportional to the target solution $\mb x$. Under the random convolutional model of $\mb A$, by using probability tools from \cite{krahmer2014structured}, we show that $\mb v^* \mb Y \mb v$ concentrates to its expectation $\mb v^* \bb E\brac{\mb Y}\mb v$ for all $\mb v \in \bb {CS}^{n-1}$ whenever $m\geq \Omega( n \poly\log n)$, ensuring that the initialization $\mb z^{(0)}$ is close to the optimal set $\mc X $. It should be noted that several variants of this initialization approach in \Cref{alg:init} have been introduced in the literature. They improve upon the $\log$ factors of sample complexity for generalized phase retrieval with i.i.d.\ measurements. Those methods include the truncated spectral method \cite{chen2015solving}, null initialization \cite{chen2015phase} and orthogonality-promoting initialization \cite{wang2016solving}. For the simplicity of analysis, here we only consider Algorithm \ref{alg:init} for the convolutional model.

\section{Main Result and Sketch of Analysis}\label{sec:analysis}
% !TEX root = ../PR Toeplitz.tex
In this section, we introduce our main theoretical result, and sketch the basic ideas behind the analysis. \edited{Without loss of generality, we assume the ground truth signal to be $\mb x\in \bb {CS}^{n-1}$. Because the problem can only be solved up to a global phase shift, we define the optimal solution set as $\mc X = \Brac{ \mb x e^{\im \phi } \mid \phi \in [0,2\pi ) }$, and correspondingly define 
\begin{align*}
	\dist(\mb z,\mc X) \doteq \inf_{\phi \in [0,2\pi)} \norm{\mb z - \mb x e^{\im \phi} }{},
\end{align*}
which measures the distance from a point $\mb z \in \bb C^{n}$ to the optimal set $\mc X$.
}
\subsection{Main Result}
\edited{Suppose the weighting vector $\mb b = \zeta_{\sigma^2}(\mb y)$ in \eqref{eqn:pr-weighted}, where
    \begin{align}\label{eqn:weighting-b}
       	\zeta_{\sigma^2}(t) = 1 - 2\pi \sigma^2 \xi_{\sigma^2}(t),\qquad \xi_{\sigma^2}(t) = \frac{1}{2\pi \sigma^2} \exp\paren{ - \frac{ \abs{t}^2}{2\sigma^2} },
    \end{align}
    with $\sigma^2>1/2$. Our main theoretical result shows that with high probability, the generalized gradient descent \eqref{eqn:grad-step-0} with spectral initialization converges \emph{linearly} to the optimal set $\mc X$.}
\begin{theorem}[Main Result]\label{thm:main}
	If $m \geq C_0 n \log^{31} n $, then \Cref{alg:init} produces an initialization $\mb z^{(0)}$ that
	\begin{align*}
	\dist\paren{\mb z^{(0)},\mc X}  \;\leq\; c_0\log^{-6} n  \norm{\mb x}{},
	\end{align*}
    with probability at least $ 1 - c_1 m^{-c_2} $. Starting from $\mb z^{(0)}$, with $\sigma^2 = 0.51$ and stepsize $\tau = 2.02$, whenever $m \geq  C_1\frac{\norm{\mb C_{\mb x}}{}^2}{\norm{\mb x}{}^2} \max \Brac{  \log^{17} n, n \log^4 n } $, for all iterates $\mb z^{(r)}$ ($r\geq 1$) in \eqref{eqn:grad-step-0}, we have
    \begin{align}\label{eqn:contraction}
       \dist\paren{\mb z^{(r)},\mc X} \;\leq \; (1- \varrho )^r \dist\paren{\mb z^{(0)}	,\mc X},
    \end{align}
    with probability at least $1 - c_3 m^{-c_4}$ for some numerical constant $\varrho \in (0,1)$.
\end{theorem}

\paragraph*{Remark.}~Our result shows that by initializing the problem $\mc O(1/\mathrm{polylog}(n) )$-close to the optimum via the spectral method, the gradient descent \eqref{eqn:grad-step-0} converges linearly to the optimal solution. As we can see, the sample complexity here also depends on $\norm{\mb C_{\mb x}}{}$, which is quite different from the i.i.d.\ case. For a typical $\mb x \in \bb {CS}^{n-1}$ (e.g., $\mb x$ is drawn uniformly random from $\bb {CS}^{n-1}$), $\norm{\mb C_{\mb x}}{}$ is on the order of $\mc O(\log n)$, and the sample complexity $m \geq \Omega\paren{n \poly \log n }$ matches the i.i.d.\ case up to log factors. However, $\norm{\mb C_{\mb x}}{}$ is nonhomogeneous over $\mb x \in \bb {CS}^{n-1}$: if $\mb x$ is sparse in the Fourier domain (e.g., $\mb x = \frac{1}{\sqrt{n}} \mb 1$), the sample complexity can be as large as $m \geq \Omega\paren{n^2 \poly \log n }$. Such a behavior is also demonstrated in the experiments of Section \ref{sec:exp}. We believe the (very large!) number of logarithms in our result is an artifact of our analysis, rather than a limitation of the method. We expect to reduce the sample complexity to $m \geq \Omega\paren{ \frac{\norm{\mb C_{\mb x}}{}^2}{\norm{\mb x}{}^2} n \log^{6} n }$ by a tighter analysis, which is left for future work. The choices of the weighting $\mb b \in \bb R^m$ in \eqref{eqn:weighting-b}, $\sigma^2 = 0.51$, and the stepsize $\tau=2.02$ are purely for the purpose of analysis. In practice, the algorithm converges with $\mb b = \mb 1$ and a choice of small stepsize $\tau$, or by using backtracking linesearch for the stepsize $\tau$.

\subsection{A Sketch of the Analysis}

In this subsection, we briefly highlight some major challenges and new ideas behind the analysis. All the detailed proofs are postponed to \Cref{sec:proofs}. The core idea behind the analysis is to show that the iterate contracts once we initialize close enough to the optimum. In the following, we first describe the basic ideas of proving iterative contraction, which critically depends on bounding a certain nonlinear function of a random circulant matrix. We sketch the core ideas of how to bound such a complicated term via the decoupling technique.

\subsubsection{Proof sketch of iterative contraction}\label{subsec:APM-connect}

\begin{comment}
The recent work \cite{soltanolkotabi2014algorithms,chen2015solving,zhang2016reshaped,wang2016solving} analyzed similar nonconvex objectives as \eqref{eqn:pr-weighted}, their work focused on characterizing the local geometry of the function landscape. When $\mb A$ is generic random, they show that the function satisfies certain \emph{regularity condition} in a region close to the optimum. However, for the convolutional phase retrieval problem, characterizing the local geometry of the nonconvex objective \eqref{eqn:pr-weighted} seems to require  $m \geq  \Omega(n^2)$ samples for measure concentration. This is due to the nonhomogeneity of $\norm{\mb C_{\mb z}}{}$ over the complex sphere $\bb {CS}^{n-1}$, which makes it hard for measure concentration in the ``bad'' regions where $\norm{\mb C_{\mb z}}{}$ is large (i.e., $\norm{\mb C_{\mb z}}{} = \mc O(\sqrt{n})$). The sample complexity $m \geq  \Omega(n^2)$ is much larger than the typical requirement $m \geq \Omega(n)$ for phase retrieval \cite{shechtman2015phase}. 

Nonetheless, as we only care about the measure concentration over the ``solution path'', here analyzing the function landscape throughout the local region seems to be an ``overkill''. Therefore, for convolutional phase retrieval, we provide a new iterative analysis which shows that the distance between every iterate and the optimum set $\mc X$ contracts, for which only the concentration over the ``solution path'' is required.
\end{comment}

Our iterative analysis is inspired by the recent analysis of \emph{alternating direction method} (ADM) \cite{waldspurger2016phase}. In the following, we draw connections between the gradient descent method \eqref{eqn:grad-step-0} and ADM, and sketch the basic ideas of convergence analysis.

\paragraph{ADM iteration.}~ADM is a classical method for solving phase retrieval problems \cite{gerchberg1972practical,netrapalli2013phase,waldspurger2016phase}, which can be considered as a heuristic method for solving the following nonconvex problem
\begin{align*}
	\min_{ \mb z \in \bb C^n, \abs{\mb u}= \mb 1} \tfrac{1}{2} \norm{ \mb A \mb z - \mb y \odot \mb u }{}^2.
\end{align*}
At every iterate $\wh{\mb z}^{(r)}$, ADM proceeds in two steps: 
\begin{align*}
  \mb c^{(r+1)} \;&=\; \mb y \odot \exp\paren{ \mb A\wh{\mb z}^{(r)} }, \\
  \wh{\mb z}^{(r+1)} \;&=\;  \arg\min_{\mb z} \frac{1}{2} \norm{\mb A\mb z - \mb c^{(r+1)} }{}^2, 
\end{align*}
which leads to the following update
\begin{align*}
  \wh{\mb z}^{(r+1)} \;=\; \mb A^\dagger \paren{ \mb y \odot \exp\paren{ \mb A \wh{\mb z}^{(r)} } },	
\end{align*}
where $\mb A^\dagger = \paren{\mb A^*\mb A}^{-1}\mb A^*$ is the pseudo-inverse of $\mb A$. Let $\wh{\theta}_r = \arg \min_{ \ol{\theta} \in [0,2\pi) } \norm{ \wh{\mb z}^{(r)} - \mb x e^{ \im \ol{\theta} } }{}  $. The distance between $\wh{\mb z}^{(r+1)}$ and $\mc X$ is bounded by

\begin{align}\label{eqn:ADM-contract}
   \dist\paren{ \wh{\mb z}^{(r+1)} ,\mc X } \;=\;	\norm{\wh{\mb z}^{(r+1)} - \mb x e^{\im \wh{\theta}_{r+1}} }{} \;\leq\; \norm{\mb A^\dagger }{} \norm{  \mb A \mb x e^{\im \wh{\theta}_{r}} -  \paren{ \mb y \odot \exp\paren{ \mb A \wh{\mb z}^{(r)} } }  }{}.
\end{align}

\paragraph{Gradient descent with $\mb b= \mb 1$.}~For simplicity and illustration purposes, let us first consider the gradient descent update \eqref{eqn:grad-step-0} with $\mb b = \mb 1$. Let $\theta_r = \arg \min_{ \ol{\theta} \in [0,2\pi) } \norm{ \mb z^{(r)} - \mb x e^{ \im \ol{\theta} } }{}  $, with stepsize $\tau =1$. The distance between the iterate $\mb z^{(r+1)}$ and the optimal set $\mc X$ is bounded by
\begin{align}
  \dist\paren{ \mb z^{(r+1)} ,\mc X } =  \norm{\mb z^{(r+1)} - \mb x e^{\im \theta_{r+1} }  }{} \;\leq\;& \norm{ \mb I - \frac{1}{m} \mb A^*\mb A }{} \norm{ \mb z^{(r)} - \mb x e^{ \im \theta_r }  }{} \nonumber \\
  &+ 
    \frac{1}{m} \norm{\mb A}{}  \norm{ \mb A \mb x e^{\im \theta_{r} } - \mb y \odot \exp\paren{ \im \phi(\mb A\mb z^{(r)} ) }    }{}. \label{eqn:grad-contract}
\end{align}

\paragraph{Towards iterative contraction.}~By measure concentration, it can be shown that
\begin{align}\label{eqn:concentration-quantities}
  	\norm{ \mb I - \frac{1}{m} \mb A^*\mb A }{} = o(1),\quad\norm{\mb A}{} \approx \sqrt{m}, \quad \norm{\mb A^\dagger }{} \approx 1/\sqrt{m},
\end{align}
holds with high probability whenever $m \geq \Omega \left( n  \poly \log n \right)$. \edited{ Therefore, based on \eqref{eqn:ADM-contract} and \eqref{eqn:grad-contract}, to show iterative contraction, it is sufficient to prove}
\begin{align}\label{eqn:contract-iterate}
   	\norm{ \mb A \mb x e^{ \im \theta}   - \mb y \odot \exp\paren{ \im \phi(\mb A\mb z ) }    }{} \leq (1-\eta) \sqrt{m} \norm{  \mb z - \mb x e^{\im \theta} }{},
\end{align}
for some constant $\eta \in (0,1)$ sufficiently small, where $\theta = \arg \min_{ \ol{\theta} \in [0,2\pi) } \norm{ \mb z - \mb x e^{ \im \ol{\theta} } }{}  $ such that $ e^{\im \theta} = \mb x^*\mb z / \abs{\mb x^*\mb z} $. \edited{By borrowing ideas from controlling \eqref{eqn:contract-iterate} in} the ADM method \cite{waldspurger2016phase}, this observation provides a new way of analyzing the gradient descent method. As an attempt to show \eqref{eqn:contract-iterate} for the random circulant matrix $\mb A$, we invoke Lemma \ref{lem:phase-diff} in the appendix, which controls the error in a first order approximation to $\exp( \im \phi( \cdot ) )$. Let us decompose 
\begin{align*}
	\mb z \;=\;  \alpha \mb x \;+\; \beta \mb w,
\end{align*}
where $\mb w \in \bb {CS}^{n-1}$ with $\mb w \perp \mb x$, and $\alpha,\beta \in \bb C$. Notice that $\phi(\alpha) = \theta$, so that by Lemma \ref{lem:phase-diff}, for any $\rho \in (0,1)$ we have
\vspace{-.05in}
\begin{eqnarray*}
  \lefteqn{  \norm{ \mb A \mb x e^{ \im \theta }   - \mb y \odot \exp\paren{ \im \phi(\mb A\mb z ) }    }{} \quad =\quad \norm{ \abs{\mb A\mb x} \odot \brac{ \exp\paren{ \im \phi\paren{ \mb A\mb x } }  - \exp\paren{ \im \phi \paren{ \mb A \mb x+ \frac{\beta}{ \alpha }  \mb A \mb w }  } }   }{} } \\
    &\leq & 2\cdot \underbrace{ \norm{ \abs{\mb A\mb x} \odot \indicator{ \abs{ \frac{\beta}{\alpha} }\abs{\mb A\mb w} \geq \rho \abs{\mb A\mb x} } }{} }_{\mc T_1}  \quad + \quad \frac{1}{1 - \rho} \abs{\frac{\beta}{\alpha } } \underbrace{\norm{ \Im\paren{ \paren{  \mb A\mb w} \odot \exp\paren{ - \im \phi(\mb A\mb x)  } } }{} }_{\mc T_2} . \qquad \qquad \qquad 
\end{eqnarray*}
%\begin{align*}
 %   &\norm{ \mb A \mb x e^{ \im \theta }   - \mb y \odot \exp\paren{ \im \phi(\mb A\mb z ) }    }{}\\
  %   =\;& \norm{ \abs{\mb A\mb x} \odot \brac{ \exp\paren{ \im \phi\paren{ \mb A\mb x } }  - \exp\paren{ \im \phi \paren{ \mb A \mb x+ \frac{\beta}{ \alpha }  \mb A \mb w }  } }   }{} \\
  %  \leq\;& \underbrace{ \norm{ \abs{\mb A\mb x} \odot \indicator{ \abs{ \frac{\beta}{\alpha} }\abs{\mb A\mb w} \geq \rho \abs{\mb A\mb x} } }{} }_{\mc T_1} + \frac{1}{1 - \rho} \abs{\frac{\beta}{\alpha } } \underbrace{\norm{ \Im\paren{ \paren{  \mb A\mb w} \odot \exp\paren{ - \im \phi(\mb A\mb x)  } } }{} }_{\mc T_2} .
%\end{align*}
\edited{The first term $\mc T_1$ is relatively much smaller than $\mc T_2$}, which can be bounded by a small numerical constant using the \emph{restricted isometry property} of a random circulant matrix \cite{krahmer2014suprema}, together with some auxiliary analysis. The detailed analysis is provided in \Cref{app:lem-bound-1}. The second term $\mc T_2$ involves a nonlinear function $\exp\paren{ - \im \phi(\mb A\mb x)  } $ of the random circulant matrix $\mb A$. Controlling this nonlinear, highly dependent random process for all $\mb w$ is a nontrivial task. In the next subsection, we explain why bounding $\mc T_2$ is technically challenging, and describe the key ideas on how to control a smoothed variant of $\mc T_2$, by using the weighting $\mb b$ introduced in \eqref{eqn:weighting-b}. We also provide intuitions for why the weighting $\mb b$ is helpful.

\subsubsection{Controlling a smoothed variant of the phase term $\mc T_2$}

As elaborated above, the major challenge of showing iterative contraction is 
bounding the suprema of the nonlinear, dependent random process $\mc T_2(\mb w)$ over the set
\begin{align*}
	\mc S\;\doteq\; \Brac{ \mb w \in \bb {CS}^{n-1} \mid \mb w \perp \mb x }.
\end{align*}
By using the fact that $\Im(u) = \frac{1}{2 \im } \paren{ u - \ol{u} }$ for any $u \in \bb C$, we have
\begin{align*}
   \sup_{ \mb w 
   \in \mc S }  \mc T_2^2(\mb w) \; \leq \; \tfrac{1}{2} \norm{\mb A}{}^2 + \tfrac{1}{2}  \sup_{ \mb w \in \mc S } \abs{ \underbrace{ \mb w^\top \mb A^\top \diag \paren{ \psi(\mb A\mb x)  } \mb A \mb w  }_{\mc L(\mb a,\mb w)  } },
\end{align*}
where we define $\psi(t) \doteq \exp\paren{ -2\im \phi(t) }$. As from \eqref{eqn:concentration-quantities}, we know that $\norm{\mb A}{} \approx \sqrt{m}$. Thus, to show \eqref{eqn:contract-iterate}, the major task left is to prove that 
\begin{align}\label{eqn:bound-desired}
	\sup_{ \mb w \in \mc S } \abs{ \mc L(\mb a,\mb w) }< (1-\eta') m
\end{align}
for some constant $\eta' \in(0,1)$.

\paragraph{Why decoupling?}~Let $\mb a_k$ ($1\leq k \leq m$) be a row vector of $\mb A$, then the term
\begin{align*}
   \mc L(\mb a,\mb w) \;=\; \mb w^\top \mb A^\top \diag \paren{ \psi(\mb A\mb x)  } \mb A \mb w \;=\; \sum_{k=1}^m  \underbrace{\psi(\mb a_k^* \mb x)   \mb w^\top \ol{\mb a}_k \ol{\mb a}_k^\top \mb w}_{\text{dependence across } k}
\end{align*}
is a summation of dependent random variables. To address this
problem, we deploy ideas from \emph{decoupling} \cite{de1999decoupling}. Informally, decoupling allows us to compare moments of random functions to functions of more independent random variables, which are usually easier to analyze. The book \cite{de1999decoupling} provides a beautiful introduction to this area. In our problem, notice that the random vector $\mb a$ occurs twice in the definition of $\mc L(\mb a, \mb w )$ -- one in the phase term $\psi(\mb A\mb x) = \exp( -2 \im \phi( \mb A \mb x ) )$, and another in the quadratic term. The general spirit of decoupling is to seek to replace one of these copies of $\mb a$ with an {\em independent} copy $\mb a'$ of the same random vector, yielding a random process with fewer dependencies. Here, we seek to replace $\mc L( \mb a, \mb w )$ with 
\begin{align}\label{eqn:L-Qec}
   \mc Q_{dec}^{\mc L}(\mb a,\mb a',\mb w)  = \mb w^\top \mb A^\top \diag\paren{ \psi( \mb A' \mb x ) }  \mb A \mb w.	
\end{align}
The utility of this new, decoupled form $\mc Q_{dec}^{\mc L}(\mb a,\mb a',\mb w)$ of $\mc L(\mb a,\mb w)$ is that it introduces extra randomness --- $\mc Q_{dec}^{\mc L}(\mb a,\mb a',\mb w)$ is now a \emph{chaos} process of $\mb a$ conditioned on $\mb a'$. This makes analyzing $\sup_{\mb w \in \mc S}\mc Q_{dec}^{\mc L}(\mb a,\mb a',\mb w)$ amenable to existing analysis of \emph{suprema of chaos processes} for random circulant matrices \cite{krahmer2014structured}. However, achieving the decoupling requires additional work; the most general existing results on decoupling pertain to {\em tetrahedral polynomials}, which are polynomials with no monomials involving any power larger than one of any random variable. By appropriately tracking cross terms, these results can also be applied to more general (non-tetrahedral) polynomials in Gaussian random variables \cite{kwapien1987decoupling}. However, our random process $\mc L(\mb a,\mb w)$ involves a nonlinear phase term $\psi(\mb A\mb w)$ which is not a polynomial, and hence is not amenable to a direct appeal to existing results. 

\paragraph{Decoupling is ``recoupling''.}~Existing results \cite{kwapien1987decoupling} for decoupling polynomials of Gaussian random variables are derived from two simple facts:
\renewcommand\labelenumi{(\theenumi)}
\begin{enumerate}
\item orthogonal projections of Gaussian variables are independent\footnote{\edited{If two random variables are jointly Gaussian, they are statistically independent if and only if they are uncorrelated.}};\label{num:1}
\item Jensen's inequality.\label{num:2}
\end{enumerate}
For the random vector $\mb a \sim \mc C \mc N( \mb 0, \mb I )$, let us introduce an independent copy $\mb \delta \sim \mc C \mc N( \mb 0, \mb I )$. Write 
\begin{align*}
   \mb g^1 = \mb a + \mb \delta,\qquad \mb g^2 = \mb a - \mb \delta.
\end{align*}
Because of \Cref{num:1}, $\mb g^1$ and $\mb g^2$ are two {\em independent} $\mc C \mc N(\mb 0,2 \mb I)$ vectors. Now, by taking conditional expectation with respect to $\mb \delta$, we have
\begin{align}\label{eqn:recouple-1}
\bb E_{\mb \delta} \left[ \mc Q_{dec}^{\mc L}( \mb g^1 , \mb g^2, \mb w ) \right] \quad =\quad \bb E_{\mb \delta} \left[ \mc Q_{dec}^{\mc L}( \mb a + \mb \delta, \mb a - \mb \delta, \mb w ) \right]  \quad  \doteq \quad \widehat{\mc L}( \mb a, \mb w ).  
\end{align} 
Thus, we can see that the key idea of decoupling $\mc L(\mb a,\mb w)$ into $\mc Q_{dec}^{\mc L}(\mb a,\mb a',\mb w)$, is essentially ``recoupling'' $\mc Q_{dec}^{\mc L}(\mb g^1,\mb g^2,\mb w)$ via conditional expectation -- the ``recoupled'' term $\wh{\mc L}$ can be reviewed as an approximation of $\mc L(\mb a,\mb w)$. Notice that by Fact 2, Jensen's inequality, for any convex function $\varphi$,
\begin{align*}
\bb E_{\mb a}\brac{  \sup_{\mb w \in \mc S} \varphi \paren{ \wh{\mc L}(\mb a,\mb w) } } \quad & = \quad \bb E_{\mb a} \brac{ \sup_{\mb w \in \mc S} \varphi\paren{  \bb E_{\mb \delta} \brac{ \mc Q_{dec}^{\mc L}(\mb a + \mb \delta,\mb a - \mb \delta, \mb w)   }   }   } \\
    & \leq \quad \bb E_{\mb a, \mb \delta}\brac{ \sup_{\mb w \in \mc S} \varphi \paren{  \mc Q_{dec}^{\mc L} (\mb a + \mb \delta,\mb a - \mb \delta, \mb w)  }   } \\
    & = \quad \bb E_{\mb g^1,\mb g^2} \brac{   \sup_{\mb w \in \mc S} \varphi \paren{  \mc Q_{dec}^{\mc L} (\mb g^1,\mb g^2, \mb w)  }  }.	
\end{align*}
Thus, by choosing $\varphi$ appropriately, i.e., as $\varphi(t) = \abs{t}^p$, we can control all the moments of $\sup_{\mb w \in \mc S}\widehat{\mc L}(\mb a,\mb w)$ via
\begin{align}\label{eqn:jensen-moments}
  \norm{ \sup_{\mb w \in \mc S} \abs{ \wh{\mc L}(\mb a,\mb w) } }{L^p}	 \leq \norm{  \sup_{\mb w \in \mc S} \abs{  \mc Q_{dec}^{\mc L} (\mb g^1,\mb g^2, \mb w)  } }{L^p}.
\end{align}
This type of inequality is very useful because it relates the moments of $\sup_{\mb w \in \mc S} \abs{ \wh{\mc L}(\mb a,\mb w) }$ to \edited{those of} $\sup_{\mb w \in \mc S} \abs{ \mc Q_{dec}^{\mc L}  (\mb a,\mb a',\mb w) }$. As discussed previously, $\mc Q_{dec}^{\mc L}$ is a chaos process of $\mb g^1$ \edited{conditioned on $\mb g^2$}. Its moments can be bounded using existing results \cite{krahmer2014suprema}.

If $\mc L$ was a tetrahedral polynomial, then we have $\widehat{\mc L} = \mc L$, i.e., the approximation is exact. As the tail bound of $\sup_{\mb w \in \mc S} \abs{\mc L(\mb a,\mb w) }$ can be controlled via its moments bounds \cite[Chapter 7.2]{foucart2013mathematical}, this allows us to directly control the object of interest $\mc L$. The reason of achieving this bound is because the conditional expectation operator $\bb E_{\mb \delta}\brac{ \cdot \mid \mb a }$ ``recouples'' $\mc Q_{dec}^{\mc L}(\mb a,\mb a',\mb w)$ back to the target $\mc L(\mb a,\mb w)$. In other words, {\em (Gaussian)  decoupling is recoupling}.

\paragraph{``Recoupling'' is Gaussian smoothing.}~\edited{In convolutional phase retrieval,
a distinctive feature of the term $\mc L(\mb a,\mb w)$ is that $\psi(\cdot)$ is a phase function and therefore $\mc L$ is {\em not} a polynomial}. Hence, it may be challenging to posit a $\mc Q_{dec}^{\mc L}$ which ``recouples'' back to $\mc L$. In other words, as $\widehat{\mc L} \ne \mc L$ in the existing form, we need to tolerate an {\em approximation error}. Although $\wh{\mc L}$ is not exactly $\mc L$, we can still control $\sup_{\mb w \in \mc S} \abs{\mc L(\mb a,\mb w)}$ through its approximation $\wh{\mc L}$, 
\begin{align}\label{eqn:decouple-triangle}
   	\sup_{\mb w \in \mc S } \abs{\mc L(\mb a, \mb w) } \quad \leq \quad \sup_{\mb w \in \mc S} \abs{ \wh{\mc L}(\mb a,\mb w) } \;+\;  \sup_{\mb w \in \mc S} \abs{ \wh{\mc L}(\mb a,\mb w) - \mc L(\mb a,\mb w) }.
\end{align}
As we discussed above, the term $\sup_{\mb w \in \mc S} \abs{ \wh{\mc L}(\mb a,\mb w) }$ can be controlled by using decoupling and the moments bound in \eqref{eqn:jensen-moments}. Therefore, the inequality \eqref{eqn:decouple-triangle} is useful to derive a sufficiently tight bound for $\mc L( \mb a, \mb w )$ if $\wh{\mc L}( \mb a, \mb w )$ is very close to $\mc L( \mb a, \mb w)$ uniformly, i.e., the approximation error is small. Now the question is: {\em for what $\mc L$ is it possible to find a ``well-behaved'' $\mc Q_{dec}^{\mc L}$ \edited{such that} the approximation error is small?} To understand this question, recall that the mechanism that links $\mc Q_{dec}$ to $\wh{\mc L}$ is the conditional expectation operator $\bb E_{\mb \delta}\brac{ \cdot \mid \mb a }$. For our case, from \eqref{eqn:recouple-1} orthogonality leads to
\begin{align}
  \wh{\mc L}(\mb a,\mb w) \;=\; \mb w^\top \mb A^\top  \diag \paren{ h(\mb A \mb x)  } \mb A\mb w, \qquad h(t) \;\doteq\; \bb E_{s\sim \mc {CN}(0,\norm{\mb x}{}^2 ) }\brac{ \psi(t+s) }	.\label{eqn:recouple-2}
\end{align}

Thus, by using the results in \eqref{eqn:decouple-triangle} and \eqref{eqn:recouple-2}, we can bound $\sup_{\mb w \in \mc S }  \abs{\mc L(\mb a, \mb w) }$ as
\begin{align}\label{eqn:bound-L}
  	\sup_{\mb w \in \mc S } \abs{\mc L(\mb a, \mb w) }  \quad \leq \quad \sup_{\mb w \in \mc S} \abs{ \wh{\mc L}(\mb a,\mb w) } \; + \; \underbrace{ \norm{ h - \psi }{ L^\infty } }_{\text{approximation error} } \norm{ \mb A }{}^2.
\end{align}
Note that the function $h$ is not exactly $\psi$, but generated by convolving $\psi$ with a multivariate Gaussian \emph{pdf}: indeed, {\em recoupling is Gaussian smoothing}. The Fourier transform of a multivariate Gaussian is again a Gaussian; it decays quickly with frequency. So, in order to admit a small approximation error, the target $\psi$ must be {\em smooth}. However, in our case, the function $\psi(t) = \exp( - 2 \im \phi( t ) )$ is discontinuous at $t = 0$; it changes extremely rapidly in the vicinity of $t = 0$, and hence its Fourier transform (appropriately defined) does not decay quickly at all. Therefore, the term $\mc L(\mb a,\mb w)$ is a poor target for approximation by using a smooth function $\widehat{\mc L}(\mb a,\mb w) = \bb E_{\mb \delta} [ \mc Q_{dec}^{\mc L}(\mb g^1,\mb g^2,\mb w) ]$. From \Cref{fig:func_psi_h_zeta}, the difference between $h$ and $\psi$ increases as $ \abs{t} \searrow 0 $. The poor approximation error $\norm{\psi - f }{L^\infty}=1$ results in a trivial bound for $\sup_{ \mb w \in \mc S } \abs{ \mc L(\mb a,\mb w) }$ instead of \edited{the desired bound \eqref{eqn:bound-desired}.}

\begin{figure*}[!htbp]
\centering
\includegraphics[width = 0.5\linewidth]{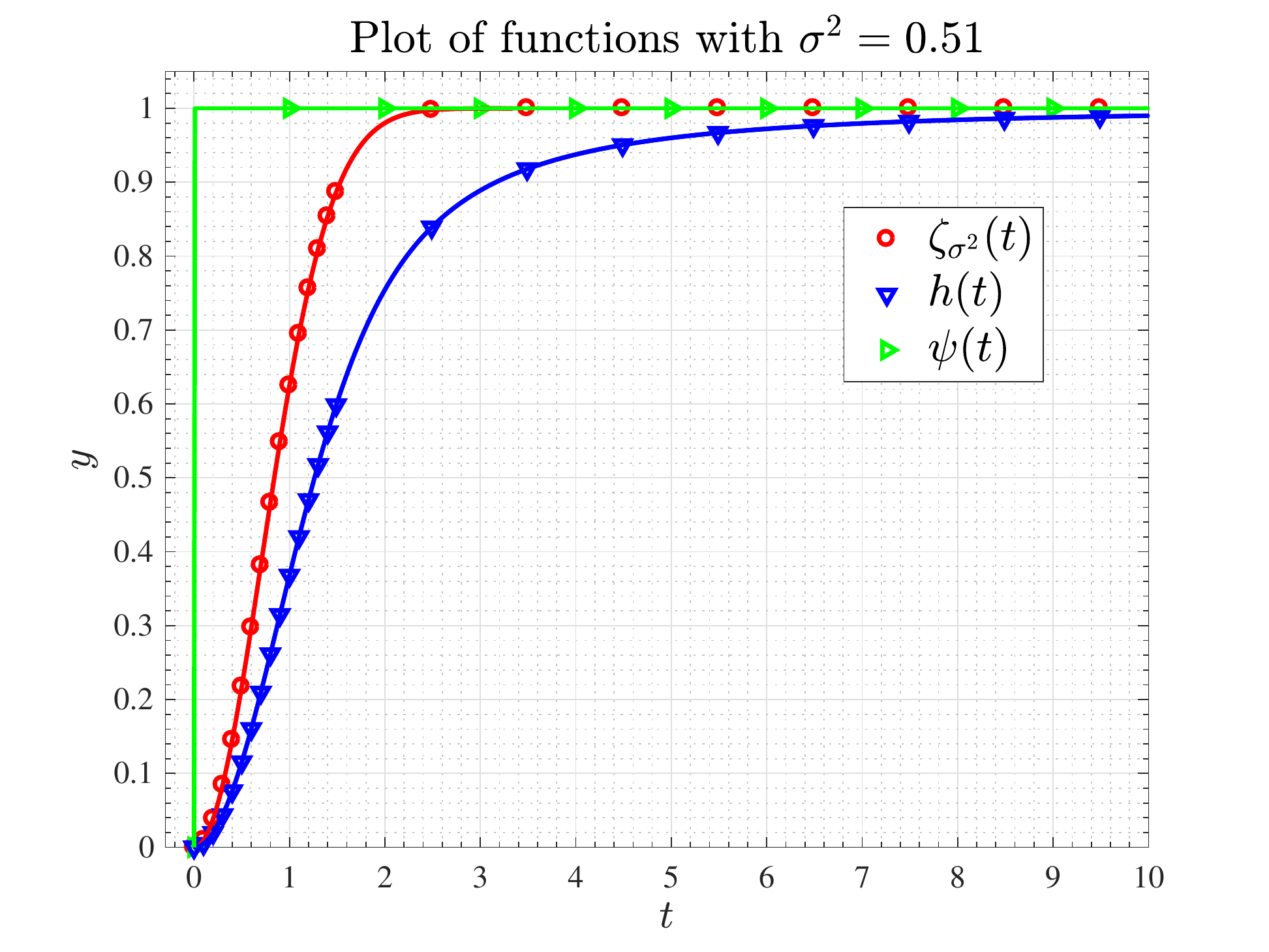}
\caption{\textbf{Plots of functions $h(t)$, $\psi(t)$ and $\zeta_{\sigma^2}(t)$ over the real line with $\sigma^2 = 0.51$.} The $\psi(t)$ function is discontinuous at $0$, and cannot be uniformly approximated by $h(t)$. On the other hand, the function $h(t)$ serves as a good approximation of the weighting $\psi(t)$.}
\label{fig:func_psi_h_zeta} 
\end{figure*}

\paragraph{Decoupling and convolutional phase retrieval.}~To reduce the approximation error caused by the nonsmoothness of $\psi$ at $t=0$, we smooth $\psi$. More specifically, we introduce a new weighted objective \eqref{eqn:pr-weighted} with Gaussian weighting $\mb b = \zeta_{\sigma^2}(\mb y)$ in \eqref{eqn:pr-conv-1} , replacing the analyzing target $\mc T_2$ with
\begin{align*}
	 \wh{\mc T}_2 \;=\; \norm{ \diag\paren{\mb b^{1/2}} \Im\paren{ \paren{  \mb A\mb w} \odot \exp\paren{ - \im \phi(\mb A\mb x)  } } }{}.
\end{align*}
Consequently, we obtain a smoothed variant $\mc L_s(\mb a,\mb w)$ of $\mc L(\mb a,\mb w)$,
\begin{align*}
   \mc L_s(\mb a,\mb w) \;=\; \mb w^\top \mb A^\top \diag \paren{ \zeta_{\sigma^2}(\mb y) \odot  \psi(\mb A\mb x)  } \mb A \mb w.
\end{align*}
Similar to \eqref{eqn:bound-L}, we obtain
\begin{align*}
  	\sup_{\mb w \in \mc S } \abs{\mc L_s(\mb a, \mb w) }  \quad \leq \quad \sup_{\mb w \in \mc S} \abs{ \wh{\mc L}(\mb a,\mb w) } \; + \; \norm{ h(t) -  \zeta_{\sigma^2}(t) \psi(t) }{ L^\infty }  \norm{ \mb A }{}^2.
\end{align*}
Now the approximation error $\norm{ h - \psi   }{L^\infty}$ in \eqref{eqn:bound-L} is replaced by $\norm{ h(t) -  \zeta_{\sigma^2}(t)\psi(t)   }{L^\infty}$. As observed from \Cref{fig:func_psi_h_zeta}, the function $\zeta_{\sigma^2}(t)$ smoothes $\psi(t)$ especially near the vicinity of $t=0$, such that the new approximation error $\norm{ f(t) -  \zeta_{\sigma^2}(t)\psi(t)   }{L^\infty}$ is significantly reduced. Thus, by using similar ideas above, we can provide a \emph{nontrivial} bound
\begin{align*}
	\sup_{ \mb w \in \mc S } \abs{ \mc L_s(\mb a,\mb w) }\;<\; (1 - \eta_s)\;m,
\end{align*}
for some $\eta_s \in (0,1)$, which is sufficient for showing iterative contraction. Finally, because of the weighting $\mb b = \zeta_{\sigma^2}(\mb y)$, it should be noticed that the overall analysis needs to be slightly modified accordingly. For a more detailed analysis, we refer the readers to \Cref{sec:proofs}.

\section{Experiments}\label{sec:exp}

In this section, we conduct experiments on both synthetic and \edited{real datasets} to demonstrate the effectiveness of the proposed method.

\begin{figure}[t]
  \centering
  \begin{minipage}[b]{0.49\textwidth}
  \centering
    \includegraphics[width=\textwidth]{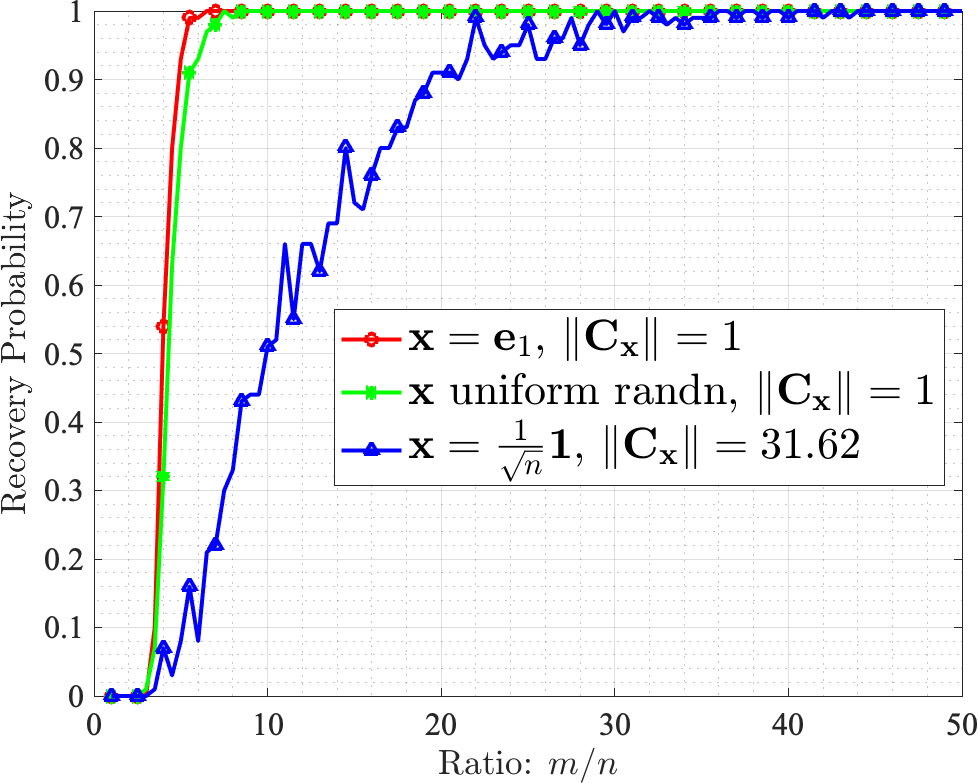}
\caption{\textbf{Phase transition for signals $\mb x \in \bb {CS}^{n-1}$ with different signal patterns.} We fix $n=1000$ and vary the ratio $m/n$. }
\label{fig:phase_transition-C_x}
  \end{minipage}
  \hfill
  \begin{minipage}[b]{0.49\textwidth}
  \centering
    \includegraphics[width=\textwidth]{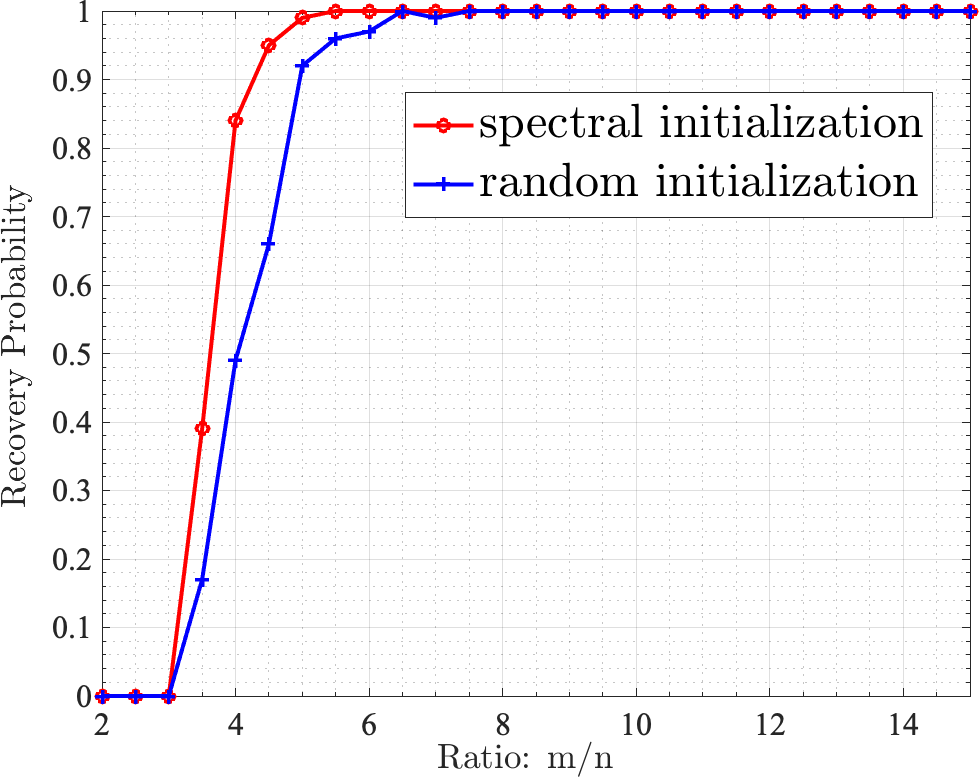}
\caption{\textbf{Phase transition with different initializations schemes.} We fix $n=1000$ and $\mb x$ is generated uniformly random from $\bb {CS}^{n-1}$. We vary the ratio $m/n$.}
\label{fig:phase_transition-initialization}
  \end{minipage}
\end{figure}

\subsection{Experiments on synthetic dataset}

\paragraph{Dependence of sample complexity on $\norm{\mb C_{\mb x}}{}$.}~First, we investigate the dependence of the sample complexity $m$ on $\norm{\mb C_{\mb x}}{}$.  We assume the ground truth $\mb x\in \bb {CS}^{n-1}$, and consider three cases:
\begin{itemize}
\item $\mb x = \mb e_1$ with $\mb e_1$ to be the standard basis vector, such that $\norm{\mb C_{\mb x}}{} =1 $;
\item $\mb x$ is uniformly random generated on the complex sphere $\bb {CS}^{n-1}$;
\item $\mb x = \frac{1}{\sqrt{n}} \mb 1$, such that $\norm{\mb C_{\mb x}}{} = \sqrt{n} $.
\end{itemize}
For each case, we fix the signal length $n =1000$ and vary the ratio $m/n$. For each ratio $m/n$, we randomly generate the kernel $\mb a \sim \mc {CN}(\mb 0,\mb I)$ in \eqref{eqn:pr-conv} and repeat the experiment $100$ times. We initialize the algorithm by the spectral method in \Cref{alg:init} and run the gradient descent \eqref{eqn:grad-step-0}. Given the algorithm output $\wh{\mb x}$, we judge the success of recovery by  
\begin{align}\label{eqn:tolerance-error}
   	\inf_{\phi \in [0,2\pi)} \norm{\wh{\mb x} - \mb x e^{\im \phi} }{} \;\leq\; \epsilon,
\end{align}
where $\epsilon=10^{-5}$. From \Cref{fig:phase_transition-C_x}, for the case when $\norm{\mb C_{\mb x}}{} =\mc O(1)$, the number of measurements needed is far less than \Cref{thm:main} suggests. Bridging the gap between the practice and theory is left for the future work.

Another observation is that the larger $\norm{\mb C_{\mb x}}{}$ is, the more samples we needed for the success of recovery. One possibility is that the sample complexity depends on $\norm{\mb C_{\mb x}}{}$, another possibility is that the extra logarithmic factors in our analysis are truly necessary for worst case (here, spectral sparse) inputs.

\paragraph{Necessity of initializations.}~As has been shown in \cite{sun2016geometric,soltanolkotabi2017}, for phase retrieval with generic measurement, when the sample complexity satisfies $m \geq \Omega(n \log n)$, with high probability the landscape of the nonconvex objective \eqref{eqn:pr-smooth} is nice enough that it enables initialization free global optimization. This raises an interesting question \edited{of whether spectral initialization is necessary} for the random convolutional model. We consider a similar setting as the previous experiment, where the ground truth $\mb x \in \bb C^n$ is drawn uniformly at random from $\bb {CS}^{n-1}$. We fix the dimension $n =1000$ and change the ratio $m/n$. For each ratio, we randomly generate the kernel $\mb a \sim \mc {CN}(\mb 0,\mb I)$ in \eqref{eqn:pr-conv} and repeat the experiment $100$ times. For each instance, we start the algorithm from random and spectral initializations, respectively. We choose the stepsize via backtracking linesearch and terminate the experiment either when the number of iterations is larger than $2 \times 10^{4}$ or the distance of the iterate to the solution is smaller than $1\times 10^{-5}$. As we can see from \Cref{fig:phase_transition-initialization}, the number of samples required for successful recovery with random initializations is only slightly more than that with the spectral initialization. \edited{This implies that the requirement of spectral initialization is an artifact of our analysis. For convolutional phase retrieval, the result in \cite{chen2018gradient} shows some promises for analyzing global convergence of gradient methods with random initializations.}

\begin{figure*}[t]
\centering
%\captionsetup{font=normalsize,labelfont={bf,sf}}
\captionsetup[sub]{font=small,labelfont={bf,sf}}
\centering
\begin{minipage}[c]{0.3\textwidth}
\subcaption{$\mb x = \mb e_1$, $\norm{\mb C_{\mb x}}{}=1$.}
\centering
	\includegraphics[width = \linewidth,height=1.8in]{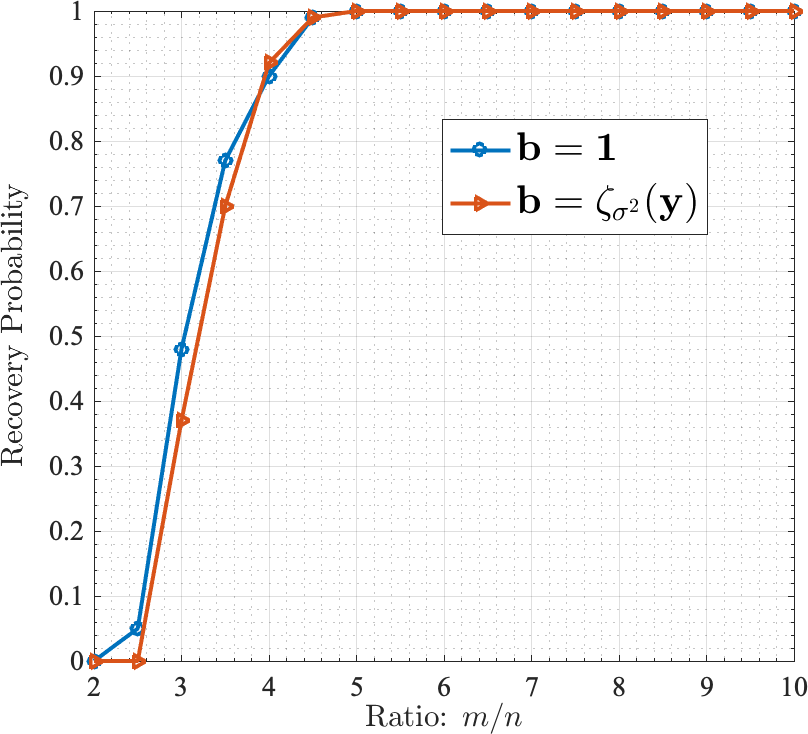}
\end{minipage}
%\hspace{0.04\textwidth}
\begin{minipage}[c]{0.3\textwidth}
\subcaption{$\mb x \sim \mc U(\bb {CS}^{n-1})$, $\norm{\mb C_{\mb x}}{} = 4.04 $.}
\centering
	\includegraphics[width = \linewidth,height=1.8in]{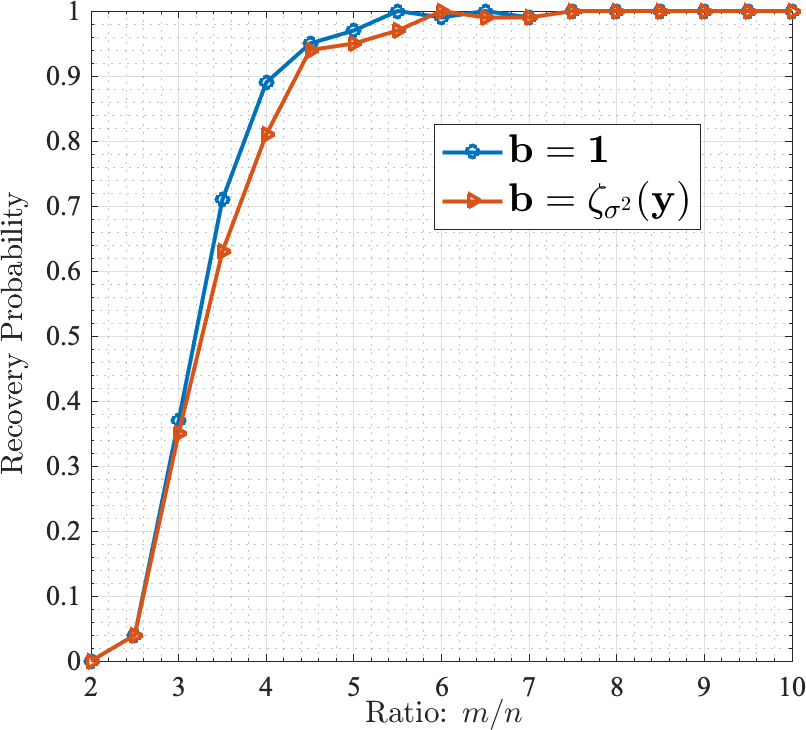}
\end{minipage}
%\hspace{0.04\textwidth}
\begin{minipage}[c]{0.3\textwidth}
\subcaption{$\mb x = \frac{1}{\sqrt{n}} \mb 1$, $\norm{\mb C_{\mb x}}{} = 10 $.}
\centering
	\includegraphics[width = \linewidth,height=1.8in]{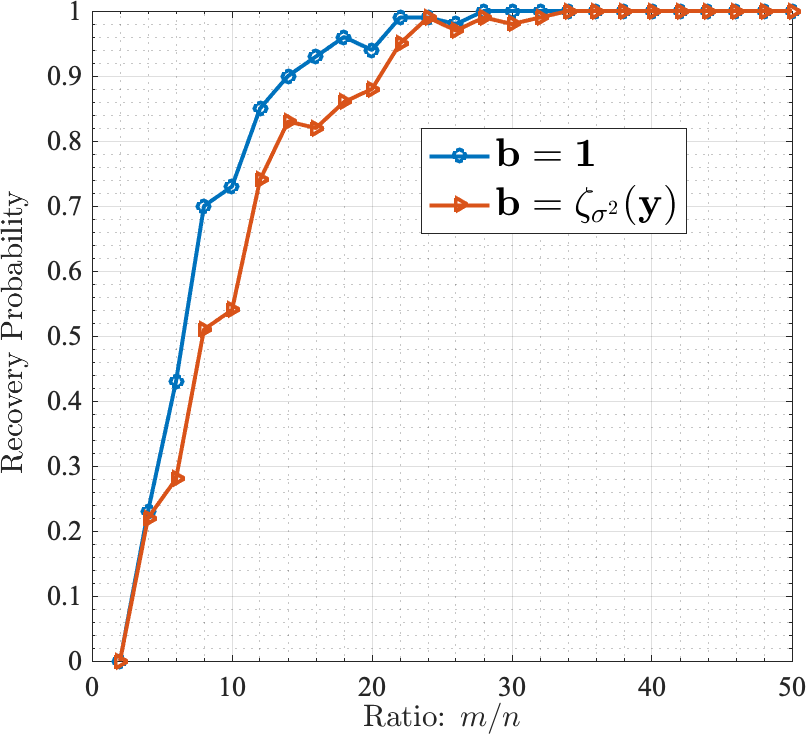}
\end{minipage}
\caption{\textbf{Phase transition for different signal patterns with weightings $\mb b$.} We fix $n=1000$ and vary the ratio $m/n$. }
\label{fig:weighting}
\end{figure*}

\paragraph{Effects of weighting $\mb b$.}~Although the weighting $\mb b$ in \eqref{eqn:weighting-b} that we introduced in \Cref{thm:main} is mainly for analysis, here we investigate its effectiveness in practice. We consider the same three cases for $\mb x$ as we did before. For each case, we fix the signal length $n =100$ and vary the ratio $m/n$. For each ratio $m/n$, we randomly generate the kernel $\mb a \sim \mc {CN}(\mb 0,\mb I)$ in \eqref{eqn:pr-conv} and repeat the experiment $100$ times. We initialize the algorithm by the spectral method in \Cref{alg:init} and run the gradient descent \eqref{eqn:grad-step-0} with weighting $\mb b = \mb 1$ and $\mb b $ in \eqref{eqn:weighting-b}, respectively. We judge success of recovery once the error \eqref{eqn:tolerance-error} is smaller than $10^{-5}$. From \Cref{fig:weighting}, we can see that the sample complexity is slightly larger for $\mb b = \zeta_{\sigma^2}(\mb y) $, \edited{the benefit of weighting here is} more for the ease of analysis.

\begin{figure*}[t]
\centering
%\captionsetup{font=normalsize,labelfont={bf,sf}}
\captionsetup[sub]{font=small,labelfont={bf,sf}}
\centering
\begin{minipage}[c]{0.3\textwidth}
\subcaption{$\mb x = \mb e_1$, $\norm{\mb C_{\mb x}}{}=1$.}
\centering
	\includegraphics[width = \linewidth,height=1.8in]{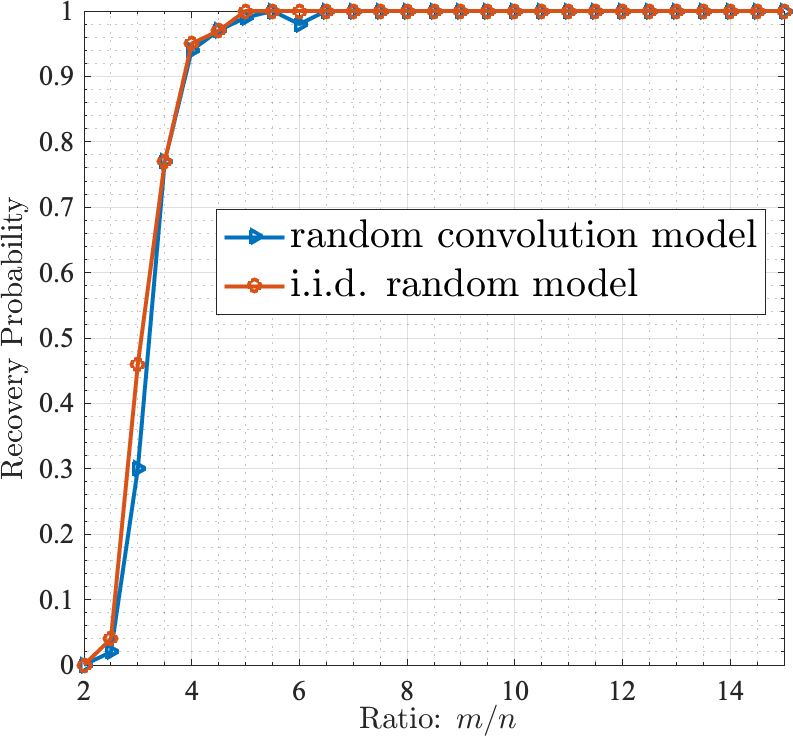}
\end{minipage}
%\hspace{0.04\textwidth}
\begin{minipage}[c]{0.3\textwidth}
\subcaption{$\mb x \sim \mc U(\bb {CS}^{n-1})$, $\norm{\mb C_{\mb x}}{} = 4.04 $.}
\centering
	\includegraphics[width = \linewidth,height=1.8in]{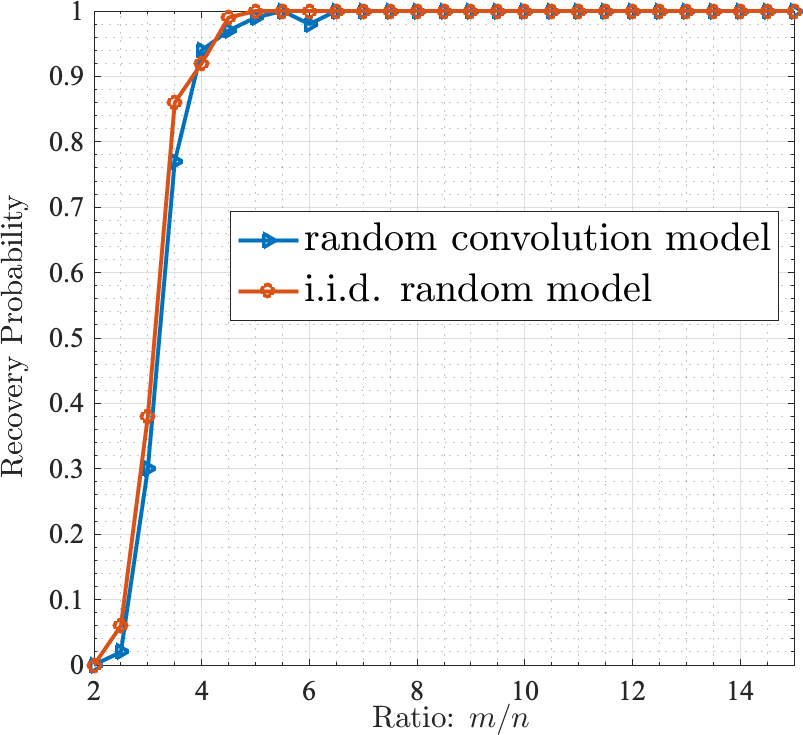}
\end{minipage}
%\hspace{0.04\textwidth}
\begin{minipage}[c]{0.3\textwidth}
\subcaption{$\mb x = \frac{1}{\sqrt{n}} \mb 1$, $\norm{\mb C_{\mb x}}{} = 10 $.}
\centering
	\includegraphics[width = \linewidth,height=1.8in]{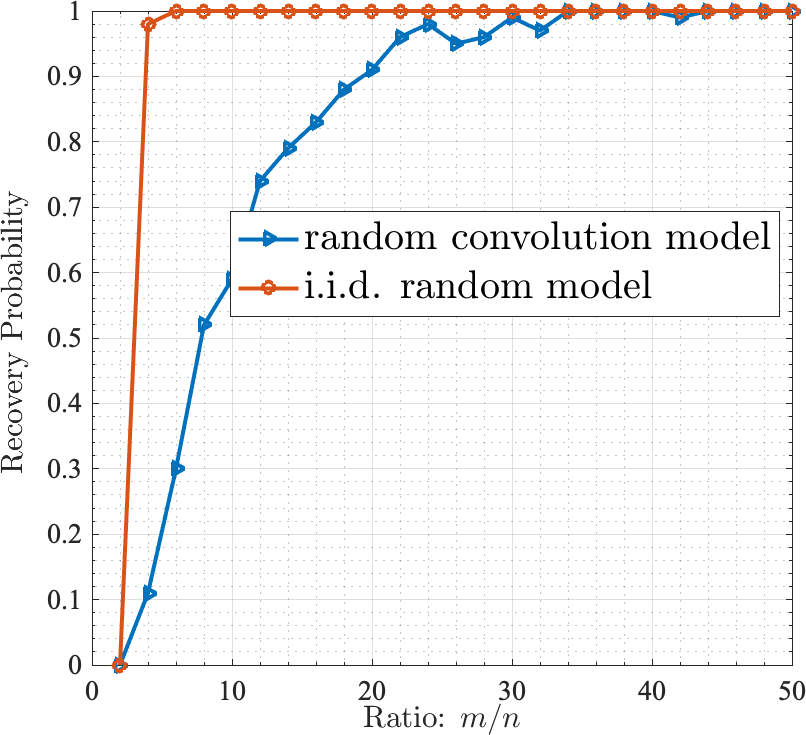}
\end{minipage}
\caption{\textbf{Phase transition of random convolution model vs. i.i.d.\ random model.} We fix $n=1000$ and vary the ratio $m/n$. }
\label{fig:convolution-iid}
\end{figure*}

\paragraph{Comparison with generic random measurements.}~Another interesting question is that, in comparison with a pure random model, \edited{how many} more samples are needed for the random convolutional model in practice? We investigate this question numerically. We consider the same three cases for $\mb x$ as we did before, and consider two random measurement models
\begin{align*}
  \mb y_1 = \abs{ \mb a \conv \mb x },\qquad \mb y_2 = \abs{ \mb A\mb x },	
\end{align*}
where $\mb a \sim \mc {CN}(\mb 0,\mb I)$, and $\mb a_k \sim_{i.i.d.} \mc {CN}(\mb 0,\mb I)$ is a row vector of $\mb A$. For each case, we fix the signal length $n =100$ and vary the ratio $m/n$. We repeat the experiment $100$ times. We initialize the algorithm by the spectral method in Algorithm \ref{alg:init} for both models, and run gradient descent \eqref{eqn:grad-step-0}. We judge success of recovery once the error \eqref{eqn:tolerance-error} is smaller than $10^{-5}$. From \Cref{fig:convolution-iid}, we can see that when $\mb x$ is typical (e.g., $\mb x = \mb e_1$ or $\mb x$ is uniformly random generated from $\bb {CS}^{n-1}$), under the same settings, the samples needed for the two random models are almost the same. However, when $\mb x$ is Fourier sparse (e.g., $\mb x = \frac{1}{\sqrt{n}}\mb 1$), more samples are required for the random convolution model.

\begin{figure*}[!htbp]
\centering
\captionsetup[sub]{font=small,labelfont={bf,sf}}
\begin{minipage}[c]{0.33\textwidth}
\subcaption{Magnitude of antenna pattern}
\label{fig:calcium_2D_img}
	\includegraphics[width = \linewidth]{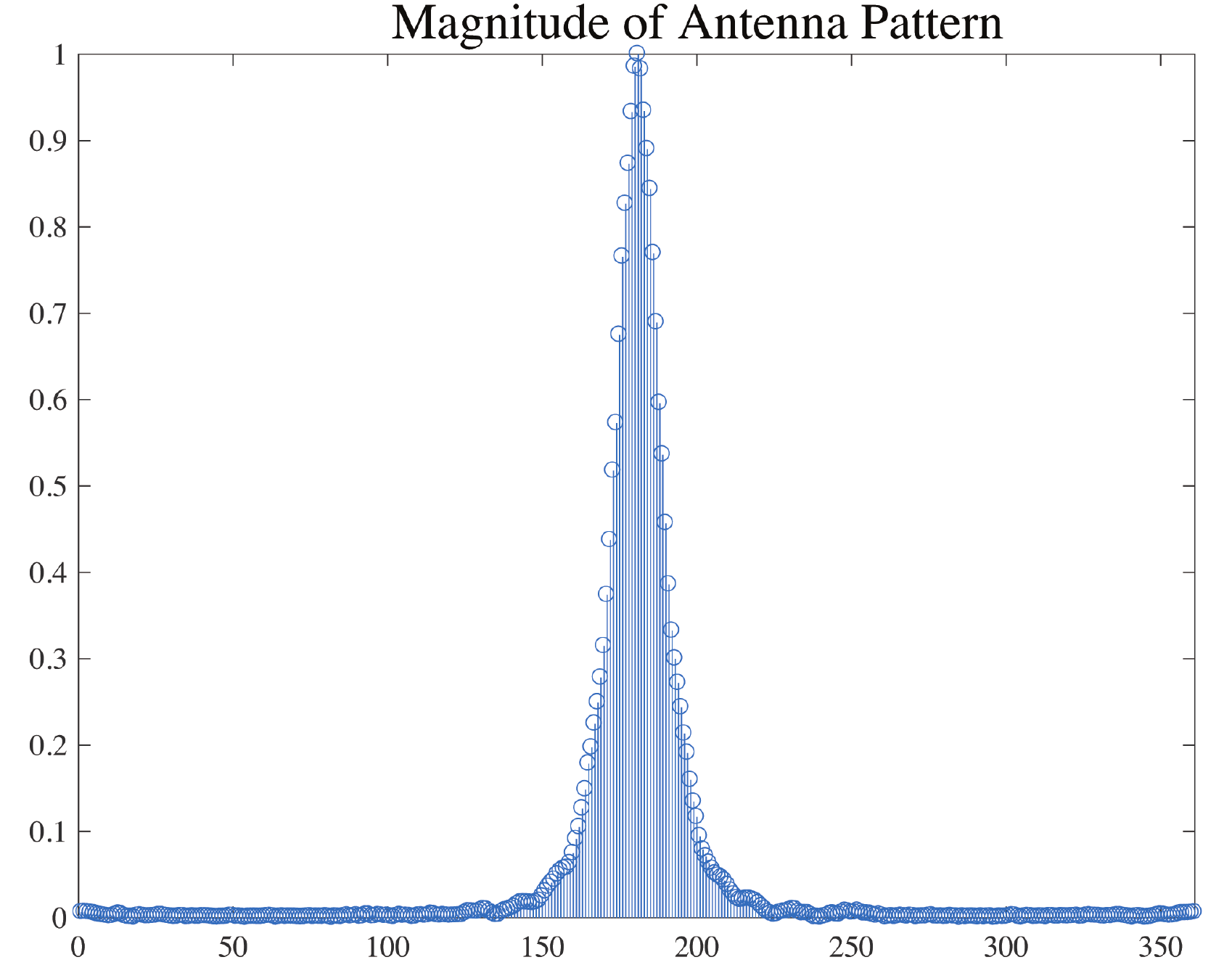}
\end{minipage}
\begin{minipage}[c]{0.65\textwidth}
\subcaption{Truth, $\mb x = \mb e_1$ (left), $\mb x \sim \mc U(\bb {CS}^{n-1}) $ (middle), $\mb x = \frac{1}{\sqrt{n}} \mb 1$ (right)}
	\includegraphics[width = \linewidth]{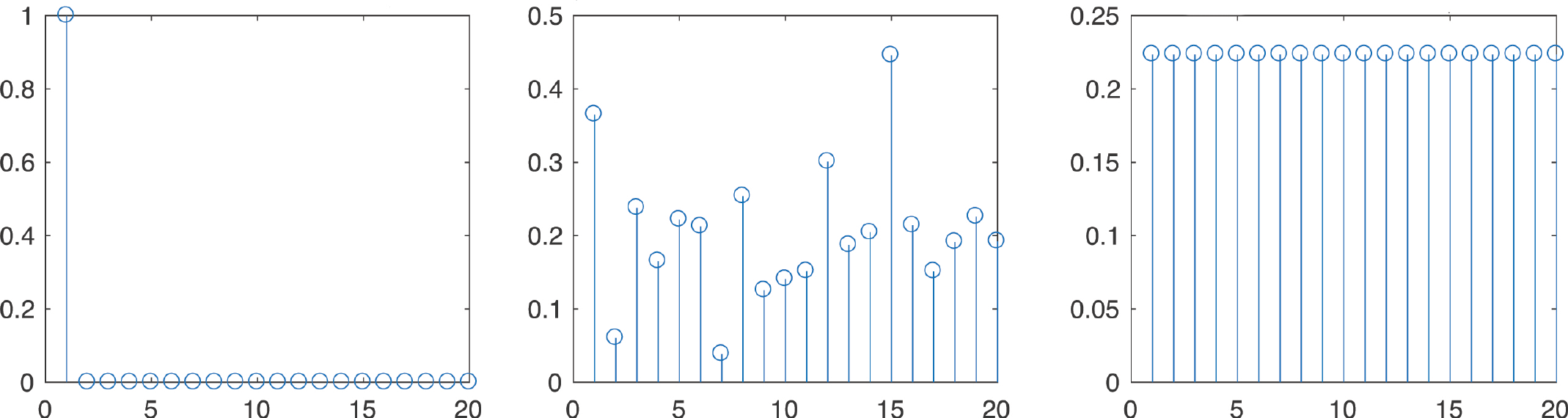}
	\subcaption{Recovered, $\mb x = \mb e_1$ (left), $\mb x \sim \mc U(\bb {CS}^{n-1}) $ (middle), $\mb x = \frac{1}{\sqrt{n}} \mb 1$ (right)}
	\includegraphics[width = \linewidth]{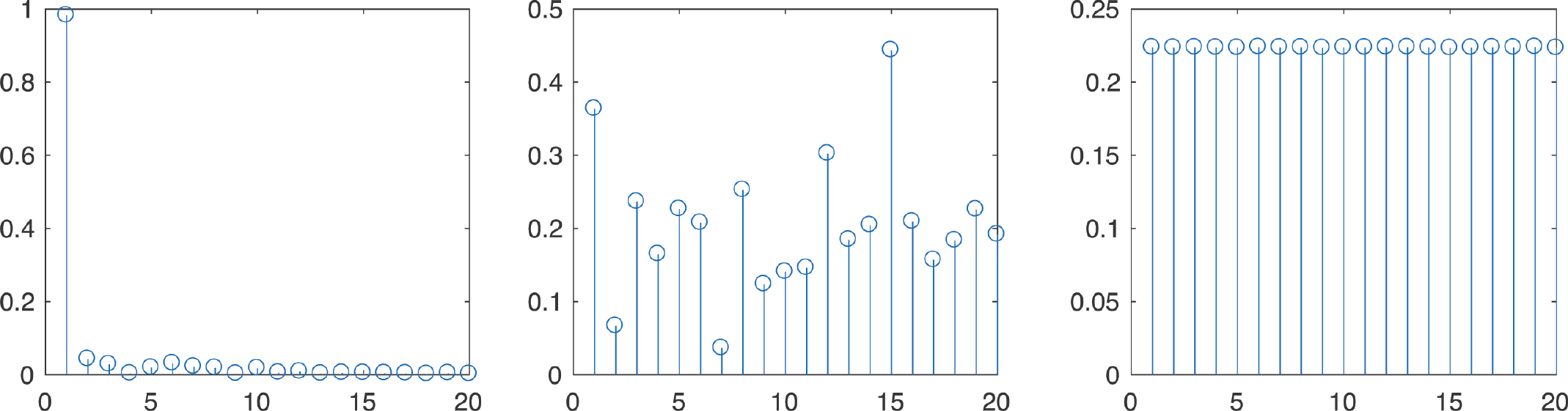}
\end{minipage}
\caption{\textbf{Experiments on real antenna filter for 5G communication.}}
\label{fig:real_data_antenna}
\end{figure*}

\subsection{Experiments on real problems}

\paragraph{Experiments on real antenna data for 5G communication.}~We demonstrate the effectiveness of the proposed method on a problem arising in 5G communication, as we mentioned in the introduction. \Cref{fig:real_data_antenna} (left) shows an antenna pattern $\mb a\in \bb C^{361}$ obtained from Bell labs. We observe the modulus of the convolution of this pattern with the signal of interest. For three different types of signals with length $n=20$, (1) $\mb x = \mb e_1$ , (2) $\mb x$ is uniformly random generated from $\bb {CS}^{n-1}$, (3) $\mb x = \frac{1}{\sqrt{n}} \mb 1$, our result in \Cref{fig:real_data_antenna} (right) shows that we can achieve almost perfect recovery.

\begin{figure*}[!htbp]
%\captionsetup{font=normalsize,labelfont={bf,sf}}
\captionsetup[sub]{font=small,labelfont={bf,sf}}

\begin{minipage}[c]{0.49\textwidth}
\subcaption{Original colored image.}
\centering
\includegraphics[width = \linewidth]{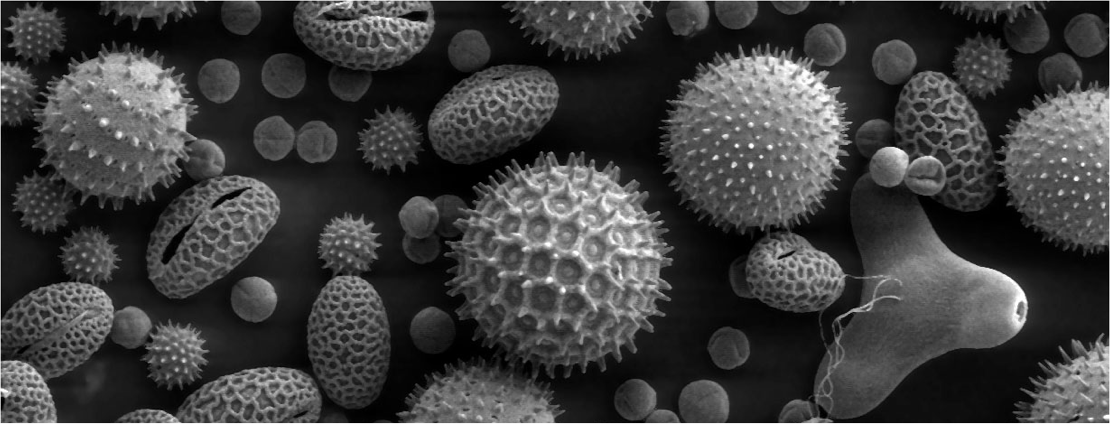}
\subcaption{Initialization.}
\centering
\includegraphics[width = \linewidth]{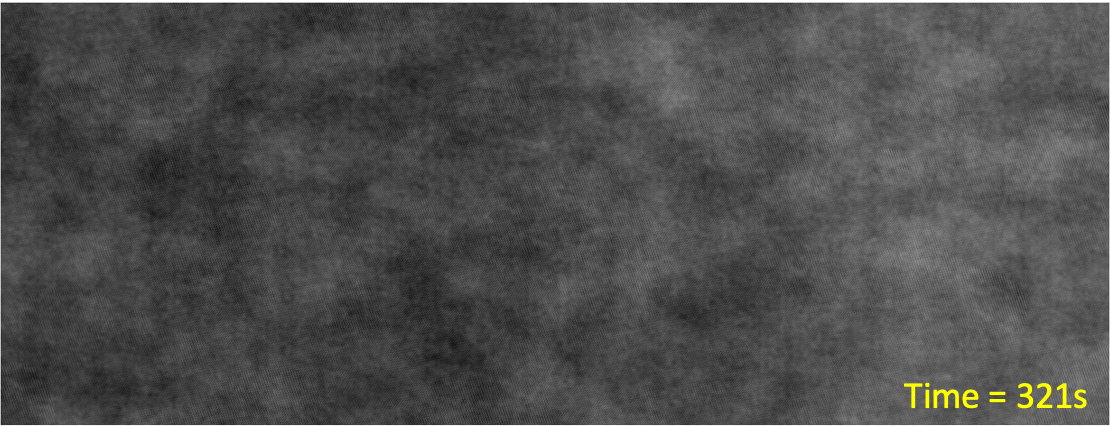}
\subcaption{$8$-th iteration.}
\centering
\includegraphics[width = \linewidth]{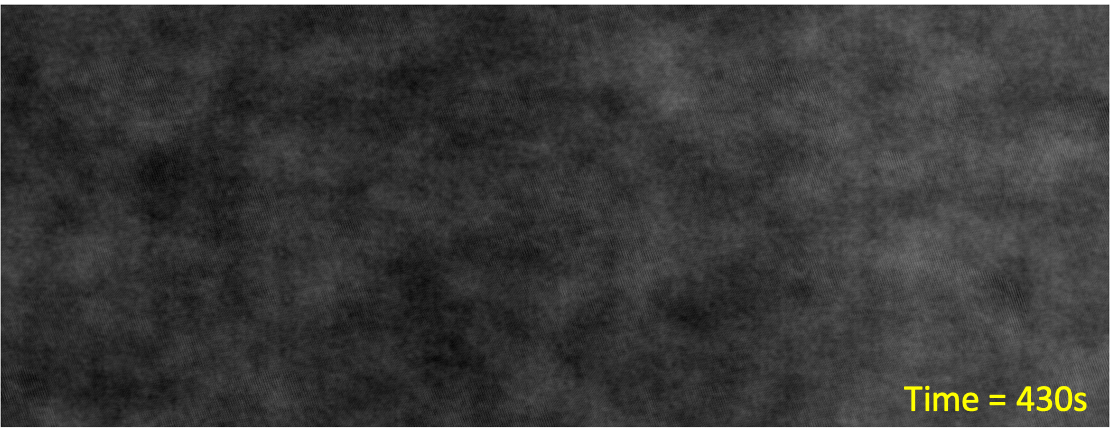}
\end{minipage}
\hfill
\begin{minipage}[c]{0.49\textwidth}
\subcaption{$16$-th iteration.}
\centering
\includegraphics[width = \linewidth]{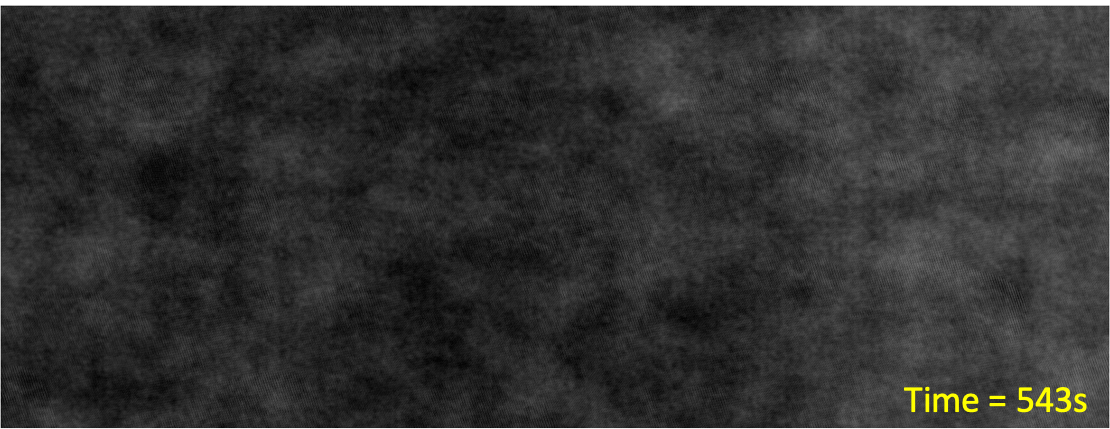}
\subcaption{$32$-th iteration.}
\centering
\includegraphics[width = \linewidth]{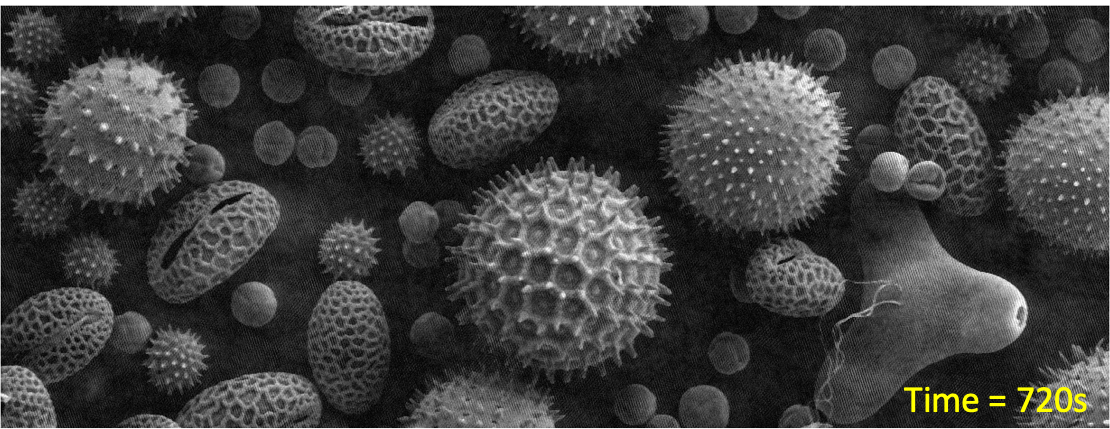}
\subcaption{$64$-th iteration.}
\centering
\includegraphics[width = \linewidth]{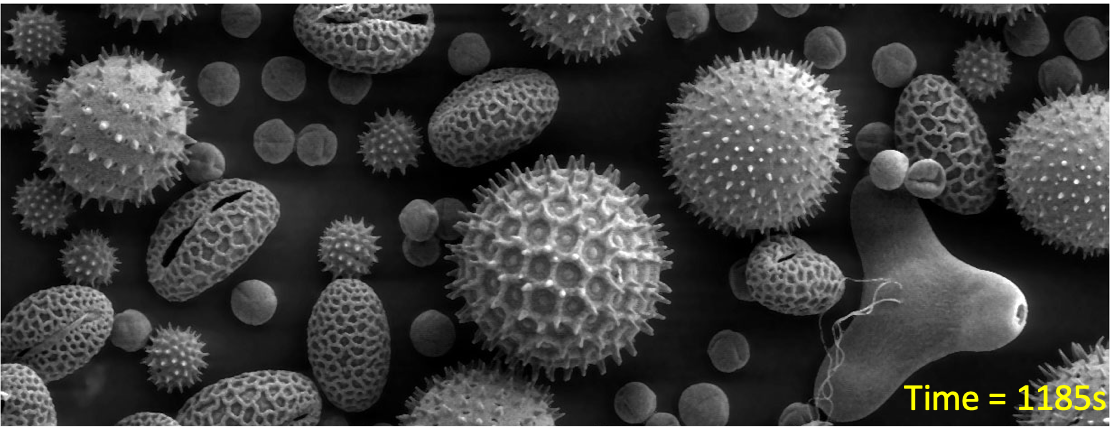}
\end{minipage}

\caption{\textbf{Experiment on a gray $468\times 1228$ electron microscopy image.}}
\label{fig:real_data_1}
\end{figure*}

\paragraph{Experiments on real images.}~\edited{Finally, we run the experiment on some real images to demonstrate the effectiveness and the efficiency of the proposed method. We use $m = 5 n\log n$ samples for reconstruction. The kernel $\mb a \in \bb C^m$ is randomly generated as complex Gaussian $\mc {CN}(\mb 0,\mb I)$. We run the power method for $100$ iterations for initialization, and stop the algorithm once the error is smaller than $1\times 10^{-4}$. We first test the proposed method on a gray $468\times 1228$ electron microscopy image. As shown in \Cref{fig:real_data_1}, the gradient descent method with spectral initialization converges to the target solution in around $64$ iterations. Second, we test our method on a color image of size $200 \times 300$ as shown in \Cref{fig:real_data}, it takes $197.08s$ to reconstruct all the RGB channels. In contrast, methods using general Gaussian measurements $\mb A\in \bb C^{m \times n}$ could easily run out of memory on a personal computer for problems of this size.}

\begin{figure*}[!htbp]
\centering
%\captionsetup{font=normalsize,labelfont={bf,sf}}
\captionsetup[sub]{font=small,labelfont={bf,sf}}
\centering
\begin{minipage}[c]{0.45\textwidth}
\subcaption{Original colored image.}
\centering
	\includegraphics[width = \linewidth]{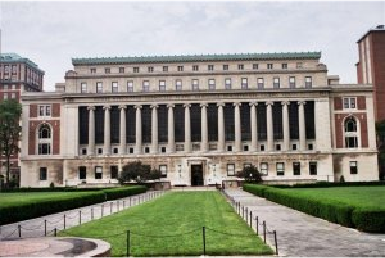}
\end{minipage}
\hfill
\begin{minipage}[c]{0.45\textwidth}
\subcaption{Recovered colored image.}
\centering
	\includegraphics[width = \linewidth]{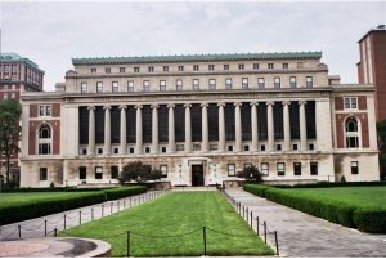}
\end{minipage}
\caption{\textbf{Experiment on colored images.}}
\label{fig:real_data}
\end{figure*}

\section{Discussion and Future Work}\label{sec:discuss}

In this work, we showed that via nonconvex optimization, the phase retrieval problem with random convolutional measurement can be solved to global optimum with $m \geq \Omega\paren{ \frac{ \norm{\mb C_{\mb x}}{}^2  }{ \norm{\mb x}{}^2 } n \poly \log n  }$ samples. Our result raises several interesting questions that we discuss below.

\paragraph{Tightening sample complexity.}~Our estimate of the sample complexity is only tight up to logarithm factors: there is a substantial gap between our theory and practice for the dependence of the logarithm factors. We believe the high order dependence of the logarithm factors is an artifact of our analysis. In particular, our analysis in Appendix \ref{app:lem-bound-1} is based on the result of RIP conditions for partial circulant random matrices, which is in no way tight. We believe that by using advanced tools in probability, the sample complexity can be tightened to at least $m \geq \Omega\paren{ n  \log^6 n  }$.

\paragraph{Geometric analysis and global result.}~Our convergence analysis is based on showing iterative contraction of gradient descent methods. However, it would be interesting if we could characterize the function landscape of nonconvex objectives as in \cite{sun2016geometric}. Such a result would provide a better explanation of why the gradient descent method works, and help us design more efficient algorithms. The major difficulty we encountered is the lack of probability tools for analyzing the random convolutional model: because of the nonhomogeneity of $\norm{\mb C_{\mb z}}{}$ over the sphere, it is hard to tightly uniformize quantities of random convolutional matrices over the complex sphere $\bb {CS}^{n-1}$. Our preliminary analysis results in suboptimal bounds for sample complexity. 

\paragraph{Tools for analyzing other structured nonconvex problems.}~This work is part of a recent surge of research efforts on deriving provable and practical nonconvex algorithms to central problems in modern signal processing and machine learning
\cite{jain2013low,hardt2014understanding,hardt2014fast,netrapalli2014non,jain2014fast,sun2014guaranteed, jain2014provable, wei2015guarantees,sa2015global,zheng2015convergent,tu2015low,chen2015fast, anandkumar2014analyzing,anandkumar2014tensor,ge2015escaping,qu2014finding,hopkins2015speeding,arora2015simple,sun2016complete_a,sun2016complete_b,lee2015blind,ge2016matrix,ge2017no,boumal2016nonconvex}. On the other hand, we believe the probability tools of decoupling and measure concentration we developed here can form a solid foundation for studying other nonconvex problems under the random convolutional model. Those problems include blind calibration \cite{ling2015self,cambareri2016through,ling2016self}, sparse blind deconvolution \cite{levin2009understanding,NIPS2011_4416,choudhary2014sparse,ahmed2014blind,lee2015blind,li2016rapid,lee2016fast,ling2017regularized,zhang2018structured,kuo2019geometry,li2018global,qu2019nonconvex}, and convolutional dictionary learning \cite{bristow2013fast,heide2015fast,chun2017convolutional,garcia2018convolutional,lau2019short}.

%\paragraph{Beyond convolutional models for phase retrieval.} \qq{It would be great if we can discuss a little bit about the more practical models we talked about in the skype meetings.}

%\paragraph{Application ideas.} Finally, beyond the cases we mentioned in the introduction, the
%application of convolutional phase retrieval seems ubiquitous in many signal processing problems, \edited{e.g. \cite{ben2017}}. We hope that the algorithm and theoretical guarantees we developed here could invite and inspire more application ideas. 

\section{Proofs of Technical Results}\label{sec:proofs}

In this section, we provide the detailed proof of \Cref{thm:main}.
The section is organized as follows. In \Cref{app:init}, we show that the the initialization produced by \Cref{alg:init} is close to the optimal solution. In \Cref{app:main}, we sketch the proof of our main result, i.e., \Cref{thm:main}, where some key details are provided in \Cref{app:contraction}. All the other supporting results are provided subsequently. We provide detailed proofs of two key supporting lemmas in \Cref{app:lem-bound-1} and \Cref{app:lem-bound-2}, respectively. Finally, other supporting lemmas are postponed to the appendices: {\em (i)} in Appendix \ref{app:tools}, we introduce the elementary tools and results that are useful throughout analysis; {\em (ii)} in Appendix \ref{app:moments-circulant-matrix}, we provide results of bounding the suprema of chaos processes for random circulant matrices; {\em (iii)} in Appendix \ref{app:decoupling}, we provide concentration results for suprema of some dependent random processes via decoupling.

\subsection{Spectral Initialization}\label{app:init}

\begin{proposition}\label{prop:init}
	Suppose $\mb z_0$ is produced by \Cref{alg:init}. Given a fixed scalar $\delta>0$, whenever $m\geq C \delta^{-2} n  \log^7 n$, we have
	\begin{align*}
		\dist^2\paren{\mb z_0, \mc X} \;\leq\; \delta \norm{\mb x}{}^2
	\end{align*}
	with probability at least $1- c_1 m^{-c_2}$. 
\end{proposition}

\edited{The proof is similar to that of \cite{soltanolkotabi2014algorithms}. However, the proof in  \cite{soltanolkotabi2014algorithms} only holds for generic random measurements. Our proof here is tailored for random circulant matrices.} We sketch the main ideas of the proof below: more detailed analysis for \edited{concentration of random circulant matrices} is retained to Appendix \ref{app:moments-circulant-matrix} and Appendix \ref{app:decoupling}.

\begin{proof}
Without loss of generality, we assume that $\norm{\mb x}{}=1$. Let $\wt{\mb z}_0$ be the leading eigenvector of
\begin{align*}
	\mb Y \;=\; \frac{1}{m}\sum_{k=1}^m \abs{ \mb a_k^*\mb x }^2 \mb a_k \mb a_k^*
\end{align*}
with $\norm{\wt{\mb z}_0}{} =1$, and let $\sigma_1$ be the corresponding eigenvalue. We have
\begin{align*}
  \dist(\mb z_0,\mc X) \leq \norm{\mb z_0 -  \wt{\mb z}_0}{}+ \dist\paren{ \wt{\mb z}_0, \mc X}.
\end{align*}
First, since $\mb z_0 = \lambda \wt{\mb z}_0$, we have
\begin{align*}
	\norm{\mb z_0 -  \wt{\mb z}_0}{} = \abs{\lambda -1}.
\end{align*}
By \Cref{thm:moments-norm-circ} in Appendix \ref{app:moments-circulant-matrix}, for any $\eps > 0$, whenever $m \geq C\eps^{-2} n \log^4 n $, we know that
\begin{align}\label{eqn:init-1}
	\abs{\lambda - 1} \leq \abs{\lambda^2-1 } = \abs{ \frac{1}{m} \sum_{k=1}^m \abs{\mb a_k^*\mb x}^2 -1  } \leq \eps/2
\end{align}
with probability at least $ 1- 2 m^{-c\log^3 n}$, where $c,C>0$ are some numerical constants. On the other hand, we have
\begin{align*}
	\dist^2(\wt{\mb z}_0,\mc X) = \arg \min_{\theta}  \norm{\wt{\mb z}_0 - \mb x e^{\im \theta} }{}^2 = 2 - 2 \abs{\mb x^*\wt{\mb z}_0 }.
\end{align*}
\Cref{thm:concentration-Y} in Appendix \ref{app:decoupling} implies that for any $\delta>0$, whenever $m \geq C' \delta^{-2} n \log^7 n$
\begin{align*}
	\norm{ \mb Y - \paren{\mb x \mb x^* + \norm{\mb x}{}^2 \mb I } }{} \;\leq\; \delta,
\end{align*}
with probability at least $1-2m^{-c_1}$. Here $c_1>0$ is some numerical constant. It further implies that
\begin{align*}
	\abs{ \wt{\mb z}_0^* \mb Y \wt{\mb z}_0 - \abs{\wt{\mb z}_0^*\mb x}^2 -1  } \;\leq\; \delta,
\end{align*}
so that
\begin{align*}
	\abs{\wt{\mb z}_0^*\mb x }^2 \;\geq\; \sigma_1 - 1 - \delta,
\end{align*}
where $\sigma_1$ is the top singular value of $\mb Y$. Since $\sigma_1$ is the top singular value, we have
\begin{align*}
	\sigma_1 \geq \mb x^*\mb Y \mb x \;=\; \mb x^*(\mb Y - \mb x\mb x^* - \norm{\mb x}{}^2 \mb I ) \mb x +2 \geq 2 - \delta.
\end{align*}
Thus, for $\delta>0$ sufficiently small, we obtain
\begin{align}\label{eqn:init-2}
	\dist^2(\wt{\mb z}_0,\mc X) \leq 2 - 2 \sqrt{1- 2\delta} 
\;\leq\; 2\delta.
\end{align}
Choose $\delta = \eps^2/8$. Combining the results in \eqref{eqn:init-1} and \eqref{eqn:init-2}, we obtain that
\begin{align*}
	\dist(\mb z_0,\mc X) \leq \norm{\mb z_0 -  \wt{\mb z}_0}{}+ \dist\paren{ \wt{\mb z}_0, \mc X} \;\leq\; \eps,
\end{align*}
holds with high probability.
\end{proof}

\subsection{Proof of Main Result}\label{app:main}

\edited{In this section, we \edited{prove} \Cref{thm:main}. Without loss of generality, we assume $\norm{\mb x}{}=1$ for the rest of the section. Given the function
\begin{align*}
	f(\mb z) = \frac{1}{2m}  \norm{ \mb b^{1/2} \odot \paren{ \mb y - \abs{\mb A\mb z} }  }{}^2,
\end{align*}
we show that simple generalized gradient descent 
\begin{align}
   \wh{\mb z} &= \mb z - \tau \frac{\partial}{ \partial \mb z} f(\mb z ),\label{eqn:grad-step} \\
  	\frac{\partial}{ \partial \mb z} f(\mb z) &= \frac{1}{m} \mb A^*\diag\paren{  \mb b } \brac{ \mb A \mb z - \mb y \odot \exp\paren{ \im \phi(\mb A\mb z) }  }, \label{eqn:grad-app}
\end{align}
with spectral initialization converges linearly to the target solution.} \edited{We restate our main result below}.
\begin{theorem}[Main Result]\label{thm:main-app}
Whenever $m \geq C_0 n \log^{31} n $, \Cref{alg:init} produces an initialization $\mb z^{(0)}$ that satisfies
\begin{align*}
\dist\paren{\mb z^{(0)},\mc X}  \;\leq\; c_0\log^{-6} n  \norm{\mb x}{}	
\end{align*}
with probability at least $ 1 - c_1 m^{-c_2} $. Suppose $\mb b = \zeta_{\sigma^2}(\mb y)$, where
\begin{align}\label{eqn:xi-zeta}
    \zeta_{\sigma^2}(t) \;=\; 1 - 2\pi \sigma^2 \xi_{\sigma^2}(t),\qquad \xi_{\sigma^2}(t) \;=\; \frac{1}{2\pi \sigma^2} \exp\paren{ - \frac{ \abs{t}^2}{2\sigma^2} },  
\end{align}
with $\sigma^2>1/2$. Starting from $\mb z^{(0)}$, with $\sigma^2 = 0.51$ and stepsize $\tau = 2.02$, whenever $m \geq  C_1\frac{\norm{\mb C_{\mb x}}{}^2}{\norm{\mb x}{}^2} \max \Brac{  \log^{17} n, n \log^4 n } $, with probability at least $1 - c_3 m^{-c_4}$ for all iterate $\mb z^{(r)} \;(r\geq 1)$ defined in \eqref{eqn:grad-step}, we have
 \begin{align*}
      \dist\paren{\mb z^{(r)},\mc X} \;\leq \; (1- \varrho )^r \dist\paren{\mb z^{(0)}	,\mc X},
  \end{align*}
holds for some small numerical constant $\varrho \in (0,1)$. 
\end{theorem}

Our proof critically depends on the following result, where we show that with high probability  for every $\mb z\in \bb C^n$ close enough to the optimal set $\mc X$, \edited{the iterate produced by \eqref{eqn:grad-step} is a contraction}.

\begin{proposition}[Iterative Contraction]\label{prop:iterate-contract}
Let $\sigma^2=0.51$ and $\tau = 2.02$. There exists some positive constants $c_1, c_2, c_3$ and $C$, such that whenever $m \geq  C  \frac{ \norm{\mb C_{\mb x}}{}^2 }{ \norm{\mb x }{}^2 } \max \Brac{  \log^{17} n, n \log^4 n } $, with probability at least $1- c_1m^{-c_2}$ for every $\mb z \in \bb C^n$ satisfying $ \dist\paren{ \mb z, \mc X } \leq c_3\log^{-6} n \norm{\mb x}{}$, we have
    \begin{align*}
       \dist\paren{ \mb z - \tau \frac{\partial}{ \partial \mb z } f(\mb z),\mc X } \;\leq\; (1-\varrho  ) \dist\paren{\mb z,\mc X}	
    \end{align*}
    holds for some small constant $\varrho \in (0,1)$. Here, $\frac{\partial}{ \partial \mb z } f(\mb z)$ is defined in \eqref{eqn:grad-app}.
\end{proposition}

\edited{To prove this proposition, let us first define
\begin{align}
   \mb M \;&=\; \mb M(\mb a) \; \doteq \; \frac{2\sigma^2+1}{m} \mb A^*\diag\paren{  \zeta_{\sigma^2}(\mb y) } \mb A, \label{eqn:M-def-main} \\
   \mb H \;&=\; \mb H(\mb a) \; \doteq \; \mb P_{\mb x^\perp} \mb M(\mb a) \mb P_{\mb x^\perp}, \label{eqn:H-def-main}
\end{align}
and introduce
\begin{align}\label{eqn:h-def-main}
	h(t) \;\doteq\; \bb E_{s\sim \mc N(0,1)}\brac{ \psi(t+s)}.
\end{align}
Given some scalar $\eps>0$ and $\sigma^2>1/2$, let us introduce a quantity
\begin{align}\label{eqn:Delta-infty}
   \Delta_{\infty}(\eps) \doteq \paren{1+2\sigma^2} \norm{ (1+\eps) \bb E_{s\sim \mc {CN}(0,1)}\brac{ \psi(t+s) } - \zeta_{\sigma^2}(t) \psi(t) }{L^\infty},
\end{align}
where $\psi(t) = \exp\paren{ - 2\im \phi(t) } $ and $\zeta_{\sigma^2}$ is defined in \eqref{eqn:xi-zeta}.} We sketch the main idea of the proof below. More detailed analysis is postponed to Appendix \ref{app:contraction}, Appendix \ref{app:lem-bound-1} and Appendix \ref{app:lem-bound-2}.

\begin{proof}[Proof of Proposition \ref{prop:iterate-contract}]
By \eqref{eqn:grad-step} and \eqref{eqn:grad-app}, and with the choice of stepsize $\tau = 2\sigma^2+1$, we have
\begin{align*}
	\wh{\mb z} \;&=\; \mb z - \frac{2\sigma^2+1}{m} \mb A^* \diag\paren{ \zeta_{\sigma^2}(\mb y) } \brac{  \mb A \mb z - \mb y \odot \exp\paren{ \im \phi\paren{\mb A\mb z  } } } \\
	\;&=\; \mb z - \mb M \mb z + \frac{2\sigma^2+1}{m} \mb A^* \diag\paren{ \zeta_{\sigma^2}(\mb y) } \brac{  \mb y \odot \exp\paren{ \im \phi\paren{ \mb A \mb z } } }.
\end{align*}
For any $\mb z \in \bb C^n$, let us decompose $\mb z$ as
\begin{align}\label{eqn:z-decomposition}
  \mb z \;=\; \alpha \mb x \;+\; \beta \mb w,	
\end{align}
where $\alpha,\beta \in \bb C$, and $\mb w \in \bb {CS}^{n-1}$ with $\mb w \perp \mb x$, and $\alpha = \abs{\alpha} e^{\im \phi(\alpha) }$ with the phase $\phi(\alpha)$ of $\alpha$ satisfies $ e^{\im \phi(\alpha) } = \mb x^*\mb z / \abs{\mb x^*\mb z} $. Therefore, if we let
\begin{align}\label{eqn:def-theta}
    \theta = {\arg \min}_{ \ol{\theta} \in [0,2\pi) } \norm{\mb z - \mb x e^{\im \ol{\theta} } }{},
\end{align}
then we also have $ \phi\paren{ \alpha } = \theta $. Thus, by using the results above, we observe
\begin{align*}
\dist^2\paren{ \wh{\mb z}, \mc X} =  \min_{ \ol{\theta} \in [0,2\pi) } \norm{ \wh{\mb z} - \mb x e^{ \im \ol{\theta} } }{}^2 &\leq \norm{ \wh{\mb z} - e^{\im \theta }\mb x }{}^2 \leq \norm{ \mb P_{\mb x^\perp} \mb d }{}^2 + \norm{ \mb P_{\mb x} \mb d  }{}^2,
\end{align*}
where we define
\begin{align}\label{eqn:L-def}
	\mb d(\mb z) \doteq \paren{\mb I - \mb M} \paren{\mb z - e^{\im \theta }\mb x } - e^{\im \theta } \mb M \mb x + \frac{2\sigma^2+1}{m} \mb A^* \diag\paren{ \zeta_{\sigma^2}(\mb y) } \brac{ \mb y \odot \exp\paren{ \im \phi\paren{\mb A\mb z } } }.
\end{align}
Let $\delta>0$, by Lemma \ref{lem:L-x-perp} and Lemma \ref{lem:L-x-parallel}, whenever $m\geq C \norm{\mb C_{\mb x}}{}^2 \max \Brac{  \log^{17} n, \delta^{-2}  n \log^4 n } $, with probability at least $1- c_1m^{-c_2}$ for all $\mb z\in \bb C^n$ such that $\norm{\mb z - \mb x e^{\im \theta } }{} \leq c_3 \delta^3 \log^{-6} n$, we have
\begin{align*}
   \norm{ \mb P_{\mb x^\perp} \mb d}{} &\leq \brac{\delta 	+ (1+\delta)\paren{ \frac{4\delta }{\rho}  \sqrt{2\sigma^2+1 } + \frac{1}{1-\rho} \frac{1}{1- \delta }   \sqrt{ \frac{1+ (2+\eps)\delta+ (1+\delta)\Delta_{\infty}(\eps) }{2}  }    }  } \norm{ \mb z -  \mb x e^{ \im \theta }  }{} \\
   \norm{ \mb P_{\mb x} \mb d}{} &\leq \paren{ \frac{2\sigma^2}{1+2\sigma^2 } + c_{\sigma^2} \delta  } \norm{\mb z - \mb x e^{\im \theta }}{}
\end{align*}
holds for any $\rho\in (0,1)$, where $\Delta_\infty(\eps)$ be defined in \eqref{eqn:Delta-infty} with $\eps\in( 0,1)$. Here, $c_{\sigma^2} $ is a numerical constant only depending on $\sigma^2$. With $\eps = 0.2$ and $\sigma^2 =0.51$, Lemma \ref{lem:L-x-perp} implies that $\Delta(\eps)\leq 0.404$. Thus, we have
\begin{align*}
   \norm{ \mb P_{\mb x^\perp} \mb d}{} \;&\leq\; \brac{\delta 	+ (1+\delta)\paren{ \frac{5.686\delta }{\rho}  + \frac{1}{1-\rho} \frac{1}{1- \delta }   \sqrt{ \frac{1+ 2.2\delta+ 0.404(1+\delta) }{2}  }    }  } \norm{ \mb z -  \mb x e^{ \im \theta }  }{} \\
   \norm{ \mb P_{\mb x} \mb d}{}\;&\leq\; (0.505 + c_{\sigma^2} \delta)\norm{\mb z - \mb x e^{\im \theta }}{}.
\end{align*}
By choosing the constants $\delta$ and $\rho$ sufficiently small, direct calculation reveals that
\begin{align*}
	\dist^2\paren{ \wh{\mb z}, \mc X} \leq \norm{ \mb P_{\mb x^\perp} \mb d }{}^2 + \norm{ \mb P_{\mb x} \mb d  }{}^2 \leq 0.96 \norm{ \mb z -\mb x e^{ \im \theta } }{}^2 = 0.96 \dist ^2\paren{ \mb z, \mc X },
\end{align*}
as desired.
\end{proof}

Now with Proposition \ref{prop:iterate-contract} in hand, we are ready to prove \edited{\Cref{thm:main-app} (in other words, \Cref{thm:main})}.

\begin{proof}[Proof of \Cref{thm:main-app}]
We prove the theorem by recursion. Let us assume that the properties in Proposition \ref{prop:iterate-contract} holds, which happens on an event $\mc E$ with probability at least $1 - c_1m^{-c_2}$ for some numerical constants $c_1,c_2>0$. By Proposition \ref{prop:init} in Appendix \ref{app:init}, for any numerical constant $\delta>0$, whenever $m \geq C \delta^{-12} n\log^{31}n$, the initialization $\mb z^{(0)}$ produced by \Cref{alg:init} satisfies 
\begin{align*}
	\dist\paren{\mb z^{(0)},\mc X}  \leq c_3 \delta^3 \log^{-6} n  \norm{\mb x}{},
\end{align*}
with probability at least $ 1 - c_4 m^{-c_5} $. Therefore, conditioned on the event $\mc E$, we know that 
\begin{align*}
   \dist\paren{ \mb z^{(1)},\mc X } = \dist\paren{  \mb z^{(0)} -\tau \frac{ \partial }{ \partial \mb z } f(\mb z), \mc X}	\leq (1-\varrho ) \dist \paren{ \mb z^{(0)},\mc X}
\end{align*}
holds for some small constant $\varrho \in (0,1)$. This proves \eqref{eqn:contraction} for the first iteration $\mb z^{(1)}$. Notice that the inequality above also implies that $\dist\paren{\mb z^{(1)},\mc X } \leq c_3 \delta^3  \log^{-6}n \norm{\mb x}{}$. Therefore, by reapplying the same reasoning, we can prove \eqref{eqn:contraction} for the iterations $r = 2,3,\cdots$.
\end{proof}

\subsection{Bounding \boldmath{$\norm{ P_{x^\perp} d(z) }{}$} and \boldmath{ $\norm{ P_{x} d(z) }{}$}  }\label{app:contraction}
Let $\mb d(\mb z)$ be defined as in \eqref{eqn:L-def} and assume that $\norm{\mb x}{} =1$. In this section, we provide bounds for $\norm{\mb P_{\mb x} \mb d }{}$ and $\norm{\mb P_{\mb x^\perp} \mb d }{}$ under the condition that $\mb z$ and $\mb x$ are close. Before presenting the main results, let us first introduce some useful preliminary lemmas. First, based on the decomposition of $\mb z$ in \eqref{eqn:z-decomposition} and the definition of $\theta$ in \eqref{eqn:def-theta}, we can show the following result.
\begin{lemma}\label{lem:beta-alpha-bound}
Let $\theta = {\arg \min}_{ \ol{\theta} \in [0,2\pi) } \norm{\mb z - \mb x e^{\im \ol{\theta} } }{} $ and suppose $\dist\paren{\mb z,\mb x} = \norm{ \mb z - \mb x e^{\im \theta } }{} \leq \epsilon$  for some $\epsilon\in (0,1)$, then we have
\begin{align*}
   \abs{ \frac{\beta }{ \alpha } 	} \;\leq\;  \frac{1}{1-\epsilon} \norm{\mb z - \mb x e^{\im \theta} }{}.
\end{align*}
\end{lemma}
\begin{proof}
Given the facts in \eqref{eqn:z-decomposition} and \eqref{eqn:def-theta} that $\mb z = \alpha \mb x + \beta \mb w $ with $\mb w\in \bb {CS}^{n-1}$ and $\mb w \perp \mb x$, and $\phi(\alpha) = \theta$, we have
\begin{align*}
   \norm{\mb z - \mb x e^{\im \theta } }{}^2 = \paren{ \abs{\alpha } - 1}^2 + \abs{ \beta }^2.
\end{align*}
This implies that 
\begin{align*}
   \abs{\beta} \leq \norm{\mb z - \mb x e^{\im \theta } }{}, \;
   \abs{\alpha}  \geq 1 - \norm{\mb z - \mb x e^{\im \theta } }{} \; \Longrightarrow \;\abs{ \frac{ \beta }{\alpha }	} \leq \frac{  \norm{\mb z - \mb x e^{\im \theta } }{} }{ 1 - \norm{\mb z - \mb x e^{\im \theta } }{} } \leq  \frac{1}{1-\epsilon} \norm{\mb z - \mb x e^{\im \theta } }{},
\end{align*}
as desired.
\end{proof}

On the other hand, our proof is also critically depends on the concentration of $\mb M(\mb a)$ in \edited{\Cref{thm:M-H-concentration} of Appendix \ref{app:decoupling}}, and the following \edited{lemmas}. Detailed proofs are given in Appendix \ref{app:lem-bound-1} and Appendix \ref{app:lem-bound-2}.

\begin{lemma}\label{lem:bound-1}
	For any given scalar $\delta\in (0,1)$, let $\gamma = c_0\delta^3\log^{-6} n $, whenever \\
	$m\geq C\max \Brac{  \norm{\mb C_{\mb x}}{}^2 \log^{17} n, \delta^{-2} n \log^4 n } $,  with probability at least $1- c_1m^{-c_2}$ for all $\mb w$ with $\norm{\mb w}{} \leq \gamma \norm{\mb x}{} $, we have the inequality
	\begin{align*}
		\norm{ \mb A\mb x\odot \mb 1_{\abs{\mb A \mb w} \geq \abs{\mb A\mb x} } }{} \leq \delta \sqrt{m} \norm{\mb w}{}.
	\end{align*}
\end{lemma}

\begin{lemma}\label{lem:bound-2}
For any scalar $\delta\in (0,1)$, whenever $m \geq C \norm{\mb C_{\mb x}}{}^2 \delta^{-2} n \log^4 n$, with probability at least $1 - cm^{-c'\log^3n}$ for all $\mb w \in \bb C^n$ with $\mb w \perp \mb x$, we have
\begin{align*}
  \norm{  \sqrt{ \frac{2\sigma^2+1 }{m } } \diag\paren{ \zeta_{\sigma^2}^{1/2}(\mb y) } \Im\paren{ \mb A \mb w \odot \exp\paren{ - \im ·\phi(\mb A\mb x) } } }{}^2 \leq \frac{ 1+(2+\eps) \delta +(1+\delta) \Delta_\infty(\eps) }{2}\norm{\mb w}{}^2.
\end{align*}
\edited{Here, $\Delta_\infty(\eps)$ is defined in \eqref{eqn:Delta-infty} for any scalar $\eps \in (0,1)$.} In particular, when $\sigma^2 = 0.51$ and $\eps = 0.2$, we have $\Delta(\eps)\leq 0.404$. With the same probability for all $\mb w \in \bb C^n$ with $\mb w \perp \mb x$, we have
\begin{align*}
  \norm{  \sqrt{ \frac{2\sigma^2+1 }{m } } \diag\paren{ \zeta_{\sigma^2}^{1/2}(\mb y) } \Im\paren{ \mb A \mb w \odot \exp\paren{ - \im ·\phi(\mb A\mb x) } } }{}^2 \leq \frac{ 1+2.2 \delta + 0.404 (1+\delta)  }{2}\norm{\mb w}{}^2.
\end{align*}
\end{lemma}

\subsubsection{Bounding the ``\boldmath{$x$}-perpendicular'' term \boldmath{$\norm{P_{ x^\perp } d}{} $} }

\begin{lemma}\label{lem:L-x-perp}
Let $\mb d$ be defined in \eqref{eqn:L-def}, and suppose $\sigma^2>1/2$ be a constant. For any $\delta>0$, whenever $m\geq C  \norm{\mb C_{\mb x}}{}^2 \max \Brac{ \log^{17} n, \delta^{-2} n \log^4 n } $, with probability at least $1- c_1m^{-c_2}$ for all $\mb z\in \bb C^n $ such that $\norm{\mb z - \mb x e^{\im \theta} }{} \leq c_3 \delta^3 \log^{-6} n$, we have
\begin{align*}
   \norm{ \mb P_{\mb x^\perp} \mb d}{} \;\leq\; \brac{\delta 	+ (1+\delta)\paren{ \frac{4\delta }{\rho}  \sqrt{2\sigma^2+1 } + \frac{1}{1-\rho} \frac{1}{1- \delta }   \sqrt{ \frac{1+ 2\delta+ (1+\delta)\Delta_{\infty}(\eps) }{2}  }    }  } \norm{ \mb z -  \mb x e^{ \im \theta }  }{}. 
\end{align*}
Here, \edited{$\Delta_\infty(\eps)$ is} defined in \eqref{eqn:Delta-infty} for any scalar $\eps \in (0,1)$. In particular, when $\eps = 0.2$ and $\sigma^2 = 0.51$, we have $\Delta_\infty(\varepsilon) \leq 0.404 $.
\end{lemma}
The analysis of bounding $\norm{ \mb P_{\mb x^\perp} \mb d }{}$ is similar to that of \cite{waldspurger2016phase}.

\begin{proof}
By the definition \eqref{eqn:L-def} of $\mb d(\mb z)$, notice that
\begin{align*}
   \norm{ \mb P_{\mb x^\perp} \mb d }{} \;\leq\;& \norm{ \mb P_{\mb x^\perp} \Brac{ \frac{2\sigma^2+1}{m} \mb A^* \diag\paren{ \zeta_{\sigma^2}(\mb y) } \brac{ \mb y \odot \exp\paren{ \im \phi\paren{\mb A\mb z } } } - e^{\im \theta } \mb M \mb x  } }{} \\
   &+\norm{ \mb P_{\mb x^\perp} (\mb I - \mb M)}{}\norm{  \mb z - e^{ \im \theta}\mb x  }{} 
\end{align*}
For the second term, \edited{by \eqref{eqn:H-bound} in} \Cref{thm:M-H-concentration}, for any $\delta>0$, whenever $m \geq C_1 \delta^{-2} \norm{\mb C_{\mb x}}{}^2 n \log^4 n$, we have 
\begin{align}
   \norm{\mb P_{\mb x^\perp} \paren{ \mb I - \mb M }  }{} \leq \delta,	
\end{align}
with probability at least $1 - c_1m^{-c_2\log^3 n}$. For the first term, we observe
\begin{align*}
   	&\norm{ \mb P_{\mb x^\perp} \Brac{ \frac{2\sigma^2+1}{m} \mb A^* \diag\paren{ \zeta_{\sigma^2}(\mb y) } \brac{ \mb y \odot \exp\paren{ \im \phi\paren{\mb A\mb z } } } - e^{\im \theta } \mb M \mb x  } }{} \\
 = \;& \norm{ \mb P_{\mb x^\perp} \Brac{ \frac{2\sigma^2+1}{m} \mb A^* \diag\paren{ \zeta_{\sigma^2}(\mb y) } \paren{ \abs{\mb A\mb x} \odot \brac{ \exp\paren{ \im \phi\paren{ \mb A\mb z }  } - \exp\paren{ \im \theta + \im \phi\paren{\mb A\mb x} }   }  }  }  }{} \\
 \leq \;& \norm{  \sqrt{\frac{ 2\sigma^2+1}{m}  } \mb P_{\mb x^\perp}\mb A^* \diag\paren{ \zeta_{\sigma^2}^{1/2} (\mb y) }  }{} \times\\
 & \norm{ \sqrt{\frac{2\sigma^2+1}{m} } \diag\paren{ \zeta_{\sigma^2}^{1/2}(\mb y)  } \paren{\abs{\mb A\mb x} \odot  \brac{ \exp\paren{ \im \phi\paren{ \mb A\mb z } } - \exp\paren{ \im \theta + \im \phi\paren{\mb A \mb x} }  } } }{}.
\end{align*}
\edited{By \eqref{eqn:H-bound} in} \Cref{thm:M-H-concentration} and Lemma \ref{lem:M-H-expectation} in Appendix \ref{app:decoupling}, for any $\delta>0$, whenever $m \geq C_1 \delta^{-2} \norm{\mb C_{\mb x}}{}^2 n \log^4 n$, we have 
\begin{align*}
 \norm{  \sqrt{\frac{ 2\sigma^2+1}{m}  } \mb P_{\mb x^\perp}\mb A^* \diag\paren{ \zeta_{\sigma^2}^{1/2} (\mb y) }  }{} \;\leq\; \norm{ \mb H }{}^{1/2} \;&\leq\; \paren{ \norm{\bb E\brac{\mb H}}{} + \norm{ \mb H - \bb E\brac{\mb H} }{} }^{1/2} \\
 & \;\leq\; \paren{1 + \delta }^{1/2} \;\leq\; 1+\delta,
\end{align*}
with probability at least $1- c_1m^{-c_2\log^3 n}$. And by Lemma \ref{lem:phase-diff} and decomposition of $\mb z$ in \eqref{eqn:z-decomposition} with $\phi(\alpha) = \theta$, we obtain
\begin{align*}
   &\norm{ \sqrt{\frac{2\sigma^2+1}{m} } \diag\paren{ \zeta_{\sigma^2}^{1/2}(\mb y)  } \paren{ \abs{\mb A\mb x} \odot  \brac{ \exp\paren{ \im \phi\paren{ \mb A\mb z } } - \exp\paren{ \im \theta + \im \phi\paren{\mb A \mb x} }  }  } }{} \\
   = \;& \norm{ \sqrt{\frac{2\sigma^2+1}{m} } \diag\paren{ \zeta_{\sigma^2}^{1/2}(\mb y)  } \paren{ \abs{\mb A\mb x} \odot  \brac{ \exp\paren{ \im \phi\paren{\mb A\mb x} } - \exp\paren{ \im \phi \paren{ \mb A\mb x + \frac{\beta}{\alpha} \mb A \mb w }  }  }  } }{} \\
   \leq \;&  \frac{1}{1-\rho} \abs{ \frac{\beta}{\alpha } } \norm{ \sqrt{ \frac{2\sigma^2+1 }{m} } \diag\paren{ \zeta_{\sigma^2}^{1/2}(\mb y)  } \Im\paren{  \mb A\mb w \odot \exp\paren{ - \im \phi\paren{\mb A\mb x} } }   }{} \\
   &+2\sqrt{\frac{2\sigma^2+1}{m} } \norm{  \abs{\mb A\mb x} \odot \mb 1_{ \abs{\frac{\beta}{\alpha}} \abs{\mb A \mb w } \geq \rho \abs{\mb A\mb x}  }  }{}  ,
\end{align*}
for any $\rho \in (0,1)$. By Lemma \ref{lem:beta-alpha-bound}, we know that $ \rho^{-1} \abs{\frac{\beta}{\alpha }} \leq \frac{2}{\rho} \norm{\mb z - \mb x e^{\im \theta} }{} < c_\rho \delta^3 \log^{-6}n$ holds under our assumption, where $c_\rho$ is a constant depending on $\rho$. Thus, whenever $ m \geq C_2\max \Brac{  \norm{\mb C_{\mb x}}{}^2 \log^{17} n, \delta^{-2} n \log^4 n }$ for any $\delta \in (0,1)$, with probability at least $1- c_1m^{-c_2}$ for all $\mb w \in \bb {CS}^{n-1}$, Lemma \ref{lem:bound-1} implies that
\begin{align*}
 \norm{  \abs{\mb A\mb x} \odot \mb 1_{ \abs{\frac{\beta}{\alpha}} \abs{\mb A \mb w } \;\geq\; \rho \abs{\mb A\mb x}  }  }{} \;\leq \; \frac{\delta }{\rho} \abs{ \frac{\beta}{ \alpha }  }  \sqrt{m} \;\leq\; \frac{2\delta }{\rho} \sqrt{m} \norm{\mb z - \mb x e^{\im \theta}}{}.
\end{align*}
Moreover, for any $\delta\in (0,1)$, whenever $m \geq C_3 \norm{\mb C_{\mb x}}{}^2 n \log^4 n$, with probability at least $ 1- c_3 m^{-c_4\log^3 n }$ for all $\mb w \in \bb {CS}^{n-1}$ with $\mb w \perp \mb x$, Lemma \ref{lem:bound-2} implies that
\begin{align*}
\norm{ \sqrt{ \frac{2\sigma^2+1 }{m} } \diag\paren{ \zeta_{\sigma^2}^{1/2}(\mb y)  } \Im\paren{  \mb A\mb w \odot \exp\paren{ - \im \phi\paren{\mb A\mb x} } }   }{} \leq \sqrt{ \frac{1+ 2\delta+ (1+\delta)\Delta_{\infty}(\eps) }{2}  },
\end{align*}
where $\Delta_\infty(\eps)$ is defined in \eqref{eqn:Delta-infty} for some $\eps \in (0,1) $. In addition, whenever $ \norm{\mb z - \mb x e^{ \im \theta}}{}\leq c_5 \delta^3 \log^{-6}n \norm{\mb z - \mb x e^{ \im \theta}}{}$ for some constant $c_5>0$, Lemma \ref{lem:beta-alpha-bound} implies that
\begin{align*}
\abs{\frac{\beta}{\alpha }	} \leq \frac{1}{ 1-c_5 \delta^3 \log^{-6}n } \norm{\mb z - \mb x e^{ \im \theta}}{} \leq \frac{1}{1-\delta} \norm{\mb z - \mb x e^{ \im \theta}}{},
\end{align*}
for $\delta>0$ sufficiently small. Thus, combining the results above, we have the bound
\begin{align*}
   \norm{ \mb P_{\mb x^\perp} \mb d}{} \leq \brac{\delta 	+ (1+\delta)\paren{ \frac{4\delta }{\rho}  \sqrt{2\sigma^2+1 } + \frac{1}{1-\rho} \frac{1}{1- \delta }   \sqrt{ \frac{1+ 2\delta+ (1+\delta)\Delta_{\infty}(\eps) }{2}  }    }  } \norm{ \mb z -  \mb x e^{ \im \theta }  }{}
\end{align*}
holds as desired. Finally, when $\sigma^2 = 0.51$ and $\eps = 0.2$, the bound for $\Delta_\infty(\eps)$ can be found in Lemma \ref{lem:Delta-infty-bound} in Appendix \ref{app:lem-bound-2}.
\end{proof}

\subsubsection{Bounding the ``\boldmath{$x$}-parallel'' term \boldmath{$\norm{P_{ x} d}{} $} }
\begin{lemma}\label{lem:L-x-parallel}
Let $\mb d(\mb z)$ be defined in \eqref{eqn:L-def}, and let $\sigma^2>1/2$ be a constant. For any $\delta>0$, whenever $m\geq C  \norm{\mb C_{\mb x}}{}^2 \max \Brac{  \log^{17} n, \delta^{-2}  n \log^4 n } $,  with probability at least $1- c_1m^{-c_2} $ for all $\mb z$ such that $\norm{\mb z - \mb x e^{\im \theta} }{} \leq c_3 \delta^3 \log^{-6} n$, we have
\begin{align*}
   \norm{\mb P_{\mb x} \mb d}{} \;\leq\; \paren{ \frac{2\sigma^2}{1+2\sigma^2 } + c_{\sigma^2} \delta  } \norm{\mb z - \mb x e^{\im \theta }}{}.
\end{align*}
Here, $c_{\sigma^2}>0$ is some numerical constant depending only on $\sigma^2$.
\end{lemma}

\begin{proof}
Given the decomposition of $\mb z$ in \eqref{eqn:z-decomposition} with $\mb w \perp \mb x$ and $\phi(\alpha) = \theta$, 
and by the definition of $\mb d(\mb z)$ in \eqref{eqn:L-def}, we observe
\begin{align*}
   \norm{ \mb P_{\mb x} \mb d }{} 
    \;=&\; \abs{ \mb x^* \Brac{  \paren{\mb I - \mb M} \paren{\mb z - e^{\im \theta }\mb x } - e^{\im \theta } \mb M \mb x + \frac{2\sigma^2+1}{m} \mb A^* \diag\paren{ \zeta_{\sigma^2}(\mb y) } \brac{ \mb y \odot \exp\paren{ \im \phi\paren{\mb A\mb z } } } }  } \\
   \;\leq &\; \abs{ \paren{ 1 - \mb x^*\bb E\brac{\mb M} \mb x }\paren{ \abs{\alpha}  - 1}  e^{\im \theta} - e^{\im \theta} \mb x^* \mb M \mb x + \frac{2\sigma^2+1}{m} \mb x^* \mb A^* \diag\paren{ \zeta_{\sigma^2}(\mb y) } \brac{ \mb y \odot \exp\paren{ \im \phi\paren{\mb A\mb z } } }  }  \\
   &+ \norm{ \mb M - \bb E\brac{\mb M} }{} \norm{ \mb z - \mb x e^{\im \theta} }{}  \\
   \leq& \underbrace{\abs{ \paren{ 1 - \mb x^*\bb E\brac{\mb M} \mb x }  } \abs{ \abs{\alpha}  - 1 } }_{\mc T_1} +  \norm{ \mb M - \bb E\brac{\mb M} }{} \norm{ \mb z - \mb x e^{\im \theta} }{}  \\
   & +\underbrace{   \abs{     \frac{2\sigma^2+1}{m} \mb x^* \mb A^* \diag\paren{\zeta_{\sigma^2}(\mb y) } \brac{ (\mb A\mb x) \odot \paren{ \exp\paren{ \im \phi\paren{\mb A \mb z - \im \phi\paren{\mb A\mb x}  }  - e^{\im \theta} \mb 1  }    }     } }    }_{\mc T_2},
\end{align*}
where for the second inequality, we used Lemma \ref{lem:M-H-expectation} such that $\mb x^*\paren{\mb I - \bb E\brac{ \mb M } }\mb w = 0$. For the first term $\mc T_1$, notice that
\begin{align*}
   \norm{ \mb z - \mb x e^{\im \theta} }{} \;=\; \sqrt{ \abs{ \abs{\alpha } - 1 }^2 + \norm{\beta \mb w } {}^2  }	 \;\geq\; \abs{ \abs{\alpha } - 1 },
\end{align*}
and by using the fact that $\bb E\brac{\mb M} = \mb I + \frac{2\sigma^2}{1+2\sigma^2 } \mb x\mb x^* $ in Lemma \ref{lem:M-H-expectation}, we have
\begin{align*}
  \mc T_1 \;=\; \frac{2\sigma^2}{1+2\sigma^2} \abs{ \abs{\alpha} - 1 } \; \leq \; \frac{2\sigma^2}{1+2\sigma^2} \norm{ \mb z - \mb x e^{\im \theta} }{}.
\end{align*}
For the term $\mc T_2$, using the fact that $\mb z = \alpha \mb x + \beta \mb w $ and $\theta = \phi(\alpha) $, and by Lemma \ref{lem:phase-approx}, notice that
\begin{align*}
   &\abs{ \exp\paren{ \im \phi\paren{\mb A \mb z} - \im \phi\paren{\mb A\mb x}  }  - e^{\im \theta} \mb 1 + \im e^{\im \theta} \Im\paren{ \frac{ \beta \mb A \mb w }{ \alpha \mb A\mb x } }  }	\\
   =\;& \abs{ \exp\paren{ \im\phi\paren{ 1+ \frac{\beta \mb A\mb w }{\alpha \mb A\mb x }  }  }  -\mb 1 + \im \Im\paren{ \frac{ \beta \mb A \mb w }{ \alpha \mb A\mb x } } } \;\leq\; 6 \abs{\frac{ \beta }{ \alpha } }^2 \abs{\frac{ \mb A \mb w }{ \mb A \mb x} }^2,
\end{align*}
whenever $  \abs{\frac{ \beta \mb A \mb w }{ \alpha \mb A \mb x} } \leq 1/2 $. Thus, by using the result above, we observe
\begin{align*}
   \mc T_2 \;\leq\;&  2 \abs{ \frac{2\sigma^2+1}{m} \mb x^* \mb A^* \diag\paren{\zeta_{\sigma^2}(\mb y) } \brac{ (\mb A\mb x) \odot \mb 1_{ \abs{\frac{\beta}{ \alpha }} \abs{\mb A\mb w  } \geq \frac{1}{2} \abs{\mb A\mb x} }  } } \\
   & + \abs{ \frac{2\sigma^2+1}{m} \mb x^* \mb A^* \diag\paren{\zeta_{\sigma^2}(\mb y) \odot \paren{ \exp\paren{ \im \phi\paren{\mb A \mb z }- \im \phi\paren{\mb A\mb x}  }  - e^{\im \theta} \mb 1  } \odot \mb 1_{ \abs{\frac{\beta}{ \alpha }} \abs{\mb A\mb w  } \leq \frac{1}{2} \abs{\mb A\mb x} }  }  \mb A\mb x    } \\
   \leq\;& 2 \abs{ \frac{2\sigma^2+1}{m} \mb x^* \mb A^* \diag\paren{\zeta_{\sigma^2}(\mb y) } \brac{ (\mb A\mb x) \odot \mb 1_{ \abs{\frac{\beta}{ \alpha }} \abs{\mb A\mb w  } \geq \frac{1}{2} \abs{\mb A\mb x} }  } }  \\
   & + \abs{ \frac{2\sigma^2+1}{m} \mb x^* \mb A^* \diag\paren{\zeta_{\sigma^2}(\mb y) } \brac{ (\mb A\mb x) \odot  \Im\paren{   \frac{  \beta  \mb A\mb w }{  \alpha \mb A\mb x } } } } \\
   &+ 6 \abs{\frac{\beta}{\alpha }}^2 \abs{ \frac{2\sigma^2+1}{m} \mb x^*\mb A^* \diag\paren{\zeta_{\sigma^2}(\mb y) \odot \frac{ \abs{\mb A\mb w }^2 }{ \abs{\mb A\mb x}^2 } } \mb A\mb x  } \\
   \leq \;& 2 \frac{ 2\sigma^2+1 }{m} \norm{\mb A}{} \norm{ \mb A\mb x\odot \mb 1_{ \abs{\frac{\beta}{ \alpha }} \abs{\mb A\mb w  } \leq \frac{1}{2} \abs{\mb A\mb x} }   }{} + 6 \abs{\frac{\beta}{\alpha }}^2 \abs{ \frac{2\sigma^2+1}{m}\mb w^* \mb A^* \diag\paren{ \zeta_{\sigma^2}(\mb y) }\mb A \mb w } \\
   & + \frac{1}{2} \abs{\frac{\beta}{\alpha} } \abs{ \frac{2\sigma^2+1}{m} \mb x^*\mb A^* \diag\paren{ \zeta_{\sigma^2}(\mb y) } \mb A\mb w } \\
   &+ \frac{1}{2}\abs{\frac{\beta}{\alpha} } \abs{ \frac{2\sigma^2+1}{m} \mb x^*\mb A^* \diag\paren{ \zeta_{\sigma^2}(\mb y) } \brac{ \exp\paren{ 2\im \phi\paren{\mb A\mb x} }\odot \ol{\mb A\mb w } } } 
\end{align*}
Given the fact that $\mb x \perp \mb w$, by Lemma \ref{lem:M-H-expectation} again we have $\mb x^*\bb E\brac{\mb M} \mb w  = 0$. Thus 
\begin{align*}
\abs{ \frac{2\sigma^2+1}{m} \mb x^*\mb A^* \diag\paren{ \zeta_{\sigma^2}(\mb y) } \mb A\mb w } \;= \; \abs{  \mb x^* \mb M \mb w - \mb x^* \bb E\brac{\mb M}\mb w } \;\leq\; \norm{ \mb M - \bb E\brac{\mb M} }{},
\end{align*}
and similarly we have
\begin{align*}
  &\abs{ \frac{2\sigma^2+1}{m} \mb x^*\mb A^* \diag\paren{ \zeta_{\sigma^2}(\mb y) } \brac{ \exp\paren{ 2\im \phi\paren{\mb A\mb x} }\odot \ol{\mb A\mb w } } } \\
  =\;& \abs{ \frac{2\sigma^2+1}{m} \mb w^\top \mb A^\top \diag\paren{ \zeta_{\sigma^2}(\mb y) }\brac{ \exp \paren{ -2\im \phi\paren{\mb A\mb
   x} }\odot \mb A\mb x } } \\
   =\;& \abs{ \frac{2\sigma^2+1}{m} \mb w^\top \mb A^\top \diag\paren{ \zeta_{\sigma^2}(\mb y) } \ol{ \mb A\mb x} } \\
   =\;& \abs{ \frac{2\sigma^2+1}{m} \mb x^* \mb A^* \diag\paren{ \zeta_{\sigma^2}(\mb y) } \mb A \mb w } \leq \norm{\mb M - \bb E\brac{\mb M}}{}.
\end{align*}
Thus, suppose $\norm{\mb z - \mb x e^{\im \theta } }{} \leq \frac{1}{2}$, by using Lemma \ref{lem:beta-alpha-bound} we know that $\abs{ \frac{\beta }{ \alpha }  } \leq 2\norm{\mb z - \mb x e^{\im \theta } }{} $. \edited{Combining the estimates above, we obtain}
\begin{align*}
   \mc T_3 \;\leq \;&  2 \frac{ 2\sigma^2+1 }{m} \norm{\mb A}{} \norm{ \mb A\mb x\odot \mb 1_{ \abs{\frac{\beta}{ \alpha }} \abs{\mb A\mb w  } \leq \frac{1}{2} \abs{\mb A\mb x} }   }{} + 24 \norm{ \mb z - \mb x e^{\im \theta} }{}^2  \norm{\mb M}{} \\
   &+ 2\norm{ \mb M - \bb E\brac{\mb M} }{} \norm{ \mb z - \mb x e^{\im \theta} }{} .
\end{align*}
Combining the estimates for $\mc T_1$ and $\mc T_2$, we have
\begin{align*}
   \norm{ \mb P_{\mb x}\mb d }{} \;\leq \;& \frac{2\sigma^2}{1+2\sigma^2} \norm{\mb z - \mb x^{\im\theta} }{} + 3\norm{ \mb M - \bb E\brac{\mb M} }{} \norm{\mb z - \mb x^{ \im \theta }}{} \\
   &+ 2 \frac{ 2\sigma^2+1 }{m} \norm{\mb A}{} \norm{ \mb A\mb x\odot \mb 1_{ \abs{\frac{\beta}{ \alpha }} \abs{\mb A\mb w  } \leq \frac{1}{2} \abs{\mb A\mb x} }   }{} + 24 \norm{\mb M}{} \norm{ \mb z - \mb x e^{\im \theta} }{}^2.
\end{align*}
By \Cref{thm:M-H-concentration}, for any $\delta>0$, whenever $m \geq C_1\delta^{-2} \norm{\mb C_{\mb x}}{}^2 n \log^4 n$, we have
\begin{align*}
    \norm{ \mb M - \bb E\brac{\mb M } }{} \leq  \delta, \quad \norm{\mb M}{} \leq \norm{ \bb E\brac{\mb M} }{} +  \delta =\frac{1+4\sigma^2}{1+2\sigma^2}+ \delta
\end{align*}
holds with probability at least $1 - c_1m^{-c_2\log^3n}$. By Corollary \ref{cor:circ-spectral-norm}, for any $\delta \in (0,1)$, whenever $m \geq C_2\delta^{-2} n \log^4 n$, we have
\begin{align*}
   \norm{\mb A}{} \;\leq\; (1+\delta) \sqrt{m}
\end{align*}
holds with probability at least $1- 2m^{-c_3\log^3 n}$ for some constant $c_3>0$. If $  \frac{1}{2}\abs{\frac{\beta}{\alpha }} \leq  \norm{\mb z - \mb x e^{\im \theta} }{} \leq c_4\delta^3 \log^{-6}n$, whenever $ m \geq C_3\max \Brac{  \norm{\mb C_{\mb x}}{}^2 \log^{17} n, \delta^{-2} n \log^4 n }$, Lemma \ref{lem:bound-1} implies that
\begin{align*}
 \norm{  \abs{\mb A\mb x} \odot \mb 1_{ \abs{\frac{\beta}{\alpha}} \abs{\mb A \mb w } \;\geq \;\frac{1}{2} \abs{\mb A\mb x}  }  }{} \;\leq\;  2 \delta  \abs{ \frac{\beta}{ \alpha }  }  \sqrt{m} \leq 4\delta  \sqrt{m} \norm{\mb z - \mb x e^{\im \theta}}{}
\end{align*}
holds for all $\mb w \in \bb {CS}^{n-1}$ with probability at least $1- c_5m^{-c_6} $. Given $\norm{\mb z - \mb x e^{\im \theta} }{} \leq \frac{c_4}{4}\delta^3 \log^{-6}n$, combining the estimates above, we have
\begin{align*}
   \norm{ \mb P_{\mb x} \mb d }{} \;&\leq\; \brac{ \frac{2\sigma^2}{1+2\sigma^2} + 3\delta +8(1+\delta)\delta \paren{2\sigma^2+1 } + 24c_4 \paren{  \frac{1+4\sigma^2}{1+2\sigma^2} + \delta }  \delta^3 \log^{-6}n  } \norm{ \mb z - \mb x e^{ \im \theta } }{}\\
   \;&\leq\; \paren{  \frac{2\sigma^2}{1+2\sigma^2} +  c_{\sigma^2} \delta   } \norm{ \mb z - \mb x e^{ \im \theta } }{}
\end{align*}
for $\delta$ sufficiently small. Here, $c_{\sigma^2}$ is some positive numerical constant depending only on $\sigma^2$.
\end{proof}

\subsection{Proof of Lemma \ref{lem:bound-1} }\label{app:lem-bound-1}
 
 In this section, we \edited{prove Lemma \ref{lem:bound-1}} in \Cref{app:contraction}, which can be restated as follows.
 \begin{lemma}
	For any given scalar $\delta\in (0,1)$, let $\gamma = c_0\delta^3\log^{-6} n $, whenever \\
	$m\geq C\max \Brac{  \norm{\mb C_{\mb x}}{}^2 \log^{17} n, \delta^{-2} n \log^4 n } $, with probability at least $1- c_1m^{-c_2}$ for all $\mb w$ with $\norm{\mb w}{}\leq \gamma \norm{\mb x}{}$, we have the inequality
	\begin{align}\label{eqn:bound-1}
		\norm{ \mb A\mb x\odot \mb 1_{\abs{\mb A \mb w} \geq \abs{\mb A\mb x} } }{} \;\leq\; \delta \sqrt{m} \norm{\mb w}{}.
	\end{align}
\end{lemma}

We prove this lemma using the results in Lemma \ref{lem:lem-1-1} and Lemma \ref{lem:lem-1-2}. 

\begin{proof}
  By Corollary \ref{cor:circ-spectral-norm}, for some small scalar $\eps \in (0,1)$, whenever $m \geq C n \log^4 n$, with probability at least $1 - m^{- c\log^3 n}$ for every $\mb w$ with $\norm{\mb w}{} \leq \gamma \norm{\mb x}{}$, we have
  \begin{align*}
  	     \norm{\mb A\mb w}{} \leq (1+\eps) \sqrt{m} \norm{\mb w}{} \leq (1+\eps) \gamma \sqrt{m} \norm{\mb x}{} \leq \paren{ \frac{1+\eps}{1-\eps}}^{1/2} \gamma \norm{\mb A\mb x}{} \leq 2 \gamma \norm{\mb A \mb x}{}.
  \end{align*}
  Let us define a set 
  \begin{align*}
   \mc S \doteq \Brac{ k \mid \abs{\mb a_k^* \mb w} \geq \abs{\mb a_k^* \mb x} }.	
  \end{align*}
  By Lemma \ref{lem:lem-1-1}, \edited{for every set $\mc S$ with $\abs{\mc S}> \rho m$ (with some $\rho\in (0,1)$ to be chosen later),  with probability at least $1- \exp\paren{-\frac{ \rho^4m }{2\norm{\mb C_{\mb x}}{}^2 } }$,} we have
  \begin{align*}
     	\norm{\paren{\mb A\mb x} \odot\mb 1_{\mc S} }{} > \frac{ \rho^{3/2} }{32} \norm{\mb A \mb x}{}.
  \end{align*}
Choose $\rho$ such that $\gamma = \frac{ \rho^{3/2}}{64}$, we have
   \begin{align*}
      \norm{\mb A\mb w}{} \geq \norm{ \paren{ \mb A\mb w } \odot \mb 1_{\mc S} }{} \geq \norm{ \paren{\mb A\mb x} \odot \mb 1_{\mc S} }{} > 2\gamma \norm{\mb A \mb x}{}.	
   \end{align*}
   This contradicts with the fact that $\norm{\mb A\mb w }{} \leq 2\gamma \norm{\mb A\mb x}{}$. Therefore, whenever $\norm{\mb w}{} \leq \gamma \norm{\mb x}{}$, with high probability we have $\abs{\mc S}\leq \rho m$ holds. Given any $\delta>0$, choose $\gamma = c\delta^3\log^{-6} n $ for some constant $c>0$. Because $\gamma = \frac{ \rho^{3/2}}{64}$, we know that $\rho = c'\delta^2/\log^4 n $. By Lemma \ref{lem:lem-1-2}, whenever $m \geq C\delta^{-2} n\log^4 n$, with probability at least $1 - 2 m^{-c\log^2 n}$ for all $\mb w\in \bb {CS}^{n-1}$, we have 
   \begin{align*}
      \norm{ \abs{\mb A \mb x} \odot \mb 1_{ \abs{\mb A\mb w} \geq \abs{\mb A\mb x} } }{} \leq \norm{ \abs{\mb A \mb w} \odot \mb 1_{\mc S} }{} \leq \delta \sqrt{m} \norm{\mb w}{}.
   \end{align*}
Combining the results above, we complete the proof.
\end{proof}

\begin{lemma}\label{lem:lem-1-1}
	Let $\rho \in (0,1)$ be a positive scalar, with probability at least $1- \exp\paren{- \frac{ \rho^4 m }{2\norm{\mb C_{\mb x}}{}^2 } }$, for every set $\mc S \in [m]$ with $\abs{\mc S} \geq \rho m$, we have
	\begin{align*}
		\norm{\abs{\mb A\mb x} \odot\mb 1_{\mc S} }{} \;>\; \frac{1}{32} \rho^{3/2} \norm{\mb A \mb x}{}.
	\end{align*}
\end{lemma}
\edited{To prove this, let us define
\begin{align}\label{eqn:g-v-u}
g_v(u) = \begin{cases}
 1 & \text{if } \abs{u}  \leq v, \\
 \frac{1}{ v}\paren{ 2v - \abs{u}  } & v < \abs{u} \leq 2v \\
 0 & \text{otherwise},	
 \end{cases}	
\end{align}
for a variable $u\in \bb C$ and a fixed positive scalar $v\in \bb R$.}

\begin{proof}
Let $\rho \in (0,1)$ be a positive scalar, from Lemma \ref{lem:lem-1-aux}, we know that 
\begin{align*}
	\norm{g_{\rho } (\mb A \mb x)}{1} \geq \norm{\mb 1_{ \abs{\mb A\mb x} \leq \rho  } }{1}
\end{align*}
holds uniformly. Thus, for an independent copy $\mb a'$ of $\mb a$, we have
\begin{align*}
\abs{\norm{g_{\rho  } (\mb C_{\mb x} \mb a)}{1} - \norm{ g_{\rho } (\mb C_{\mb x} \mb a')  }{1}}	\leq \norm{ g_{\rho  } (\mb C_{\mb x} \mb a) -  g_{\rho } (\mb C_{\mb x} \mb a')  }{1} &\leq \frac{\sqrt{m} }{ \rho } \norm{\mb C_{\mb x} \mb a - \mb C_{\mb x} \mb a'}{} \\
&\leq  \frac{\sqrt{m} }{ \rho } \norm{\mb C_{\mb x}}{} \norm{\mb a - \mb a'}{}.
\end{align*}
Therefore, we can see that $\norm{g_{ \rho } (\mb C_{\mb x} \mb a)}{1}$ is $L$-Lipschitz with respect to $\mb a$, with $L=\frac{\sqrt{m} }{ \rho } \norm{\mb C_{\mb x}}{}$. By Gaussian concentration inequality in Lemma \ref{lem:gauss-concentration}, we have
\begin{align}
\bb P\paren{ \abs{ \norm{g_{\rho  } (\mb C_{\mb x} \mb a)}{1} - \bb E \brac{\norm{g_{\rho } (\mb C_{\mb x} \mb a)}{1} } } \geq t } \leq 2\exp\paren{ -\frac{t^2}{2L^2} }	.
\end{align}
By using the fact that $\sqrt{2} \abs{\mb a_k^*\mb x}$ follows the $\chi$ distribution, we have
\begin{align*}
	\bb E \brac{\norm{g_{\rho  } (\mb C_{\mb x} \mb a)}{1} } \leq  \sum_{k=1}^m \bb E\brac{ \indicator{ \abs{\mb a_k^*\mb x }\leq 2\rho } } = \sum_{k=1}^m \bb P\paren{ \abs{\mb a_k^*\mb x} \leq 2\rho }  \leq  \rho m.
\end{align*}
Thus, with probability at least $1 - 2\exp\paren{ - \frac{\rho^4 m}{ 2\norm{\mb C_{\mb x}}{}^2 } } $, we have
\begin{align*}
	 \norm{\mb 1_{ \abs{\mb A\mb x} \leq \rho } }{1} \leq \norm{g_{\rho } (\mb A \mb x)}{1} \leq  2 \rho m
\end{align*}
holds.
Thus, for any set $\mc S$ such that $\abs{\mc S} \geq 4\rho m $, we have
\begin{align*}
\norm{ (\mb A\mb x) \odot \mb 1_{\mc S} }{}^2 \geq  \norm{ \paren{\mb A\mb x} \odot \mb 1_{ \abs{\mb A\mb x} \leq \rho } }{}^2  \geq  2 \rho^3 m .
\end{align*}
Thus, by replacing $4\rho$ with $\rho$, we complete the proof.
\end{proof}

\begin{lemma}\label{lem:lem-1-2}
	Given any scalar $\delta>0$, let $\rho \in (0,c_\delta log^{-4} n )$ with $c_\delta$ be some constant depending on $\delta$, whenever $m\geq C \delta^{-2} n \log^4 n$, with probability at least $1 - 2 m^{-c \log^2 n}$, for any set $\mc S\in [m]$ with $\abs{\mc S}<\rho m$ and for all $\mb w \in \bb C^n$, we have
	\begin{align*}
		\norm{(\mb A \mb w) \odot \mb 1_{\mc S} }{} \;\leq\; \delta \sqrt{m} \norm{\mb w}{}.
	\end{align*}
\end{lemma}
\begin{proof}
    Without loss of generality, let us assume that $\norm{\mb w}{}=1$. First, notice that
	\begin{align*}
		\norm{\mb A \mb w \odot \mb 1_{\mc S}}{} \;=\; \sup_{ \mb v \in \bb {CS}^{m-1} ,\; \supp(\mb v)\subseteq \mc S}\innerprod{ \mb v}{\mb A\mb w} \; \leq\; \sup_{\mb v \in \bb {CS}^{m-1},\; \supp(\mb v)\subseteq \mc S} \norm{\mb A^*\mb v}{}.
	\end{align*}
	By Lemma \ref{lem:rip-large-perturb}, for any positive scalar $\delta>0$ and any $\rho \in (0,c\delta^2 \log^{-4} n )$, whenever $m \geq C\delta^{-2} n\log^4 n$, with probability at least $ 1- m^{-  c'\log^2 n }$, we have
	\begin{align*}
		\sup_{\mb v \in \bb {CS}^{n-1} ,\; \supp(\mb v)\subseteq \mc S} \norm{\mb A^*\mb v}{} \leq \delta \sqrt{m}.
	\end{align*}
	Combining the result above, we complete the proof.
\end{proof}

\begin{lemma}\label{lem:lem-1-aux}
For a variable $u\in \bb C$ and a fixed positive scalar $v\in \bb R$, the function $g_v(u)$ introduced in \eqref{eqn:g-v-u} is $1/ v$-Lipschitz. Moreover, the following bound 
\begin{align*}
g_v(u) \geq \indicator{ \abs{u} \leq v }	
\end{align*}
holds uniformly for $u$ over the whole space.
\end{lemma}
\begin{proof}
The proof of Lipschitz continuity of $g_v(u)$ is straight forward, and the inequality directly follows from the definition of $g_v(u)$.
\end{proof}

\subsection{Proof of Lemma \ref{lem:bound-2} }\label{app:lem-bound-2}
 \edited{In this section, we prove Lemma \ref{lem:bound-2} in \Cref{app:contraction}, which can be restated as follows.}
\begin{lemma}
For any scalar $\delta\in (0,1)$, whenever $m \geq C \norm{\mb C_{\mb x}}{}^2 \delta^{-2} n \log^4 n$, with probability at least $1 - cm^{-c'\log^3n}$ for all $\mb w \in \bb C^n$ with $\mb w \perp \mb x$, we have
\begin{align*}
  \norm{  \sqrt{ \frac{2\sigma^2+1 }{m } } \diag\paren{ \zeta_{\sigma^2}^{1/2}(\mb y) } \Im\paren{ \mb A \mb w \odot \exp\paren{ - \im ·\phi(\mb A\mb x) } } }{}^2 \leq \frac{ 1+(2+\eps) \delta +(1+\delta) \Delta_\infty(\eps) }{2}\norm{\mb w}{}^2
\end{align*}
holds. In particular, when $\sigma^2 = 0.51$ and $\eps = 0.2$, we have $\Delta(\eps)\leq 0.404$. With the same probability for all $\mb w \in \bb C^n$ with $\mb w \perp \mb x$, we have
\begin{align*}
  \norm{  \sqrt{ \frac{2\sigma^2+1 }{m } } \diag\paren{ \zeta_{\sigma^2}^{1/2}(\mb y) } \Im\paren{ \mb A \mb w \odot \exp\paren{ - \im ·\phi(\mb A\mb x) } } }{}^2 \leq \frac{ 1+2.2 \delta + 0.404 (1+\delta)  }{2}\norm{\mb w}{}^2.
\end{align*}
\end{lemma}

\begin{proof}
Without loss of generality, let us assume $\mb w \in \bb {CS}^{n-1}$. For any $\mb w\in \bb {CS}^{n-1} $ with $\mb w \perp \mb x$, we observe
\begin{align*}
   &\norm{ \sqrt{\frac{2\sigma^2+1}{m} } \diag\paren{ \zeta_{\sigma^2}^{1/2}(\mb y) } \Im\paren{ \mb A\mb w \odot \exp\paren{ - \im \phi(\mb A\mb x) } }  }{}^2	 \\
   =\;& \norm{ \frac{1}{2}  \sqrt{\frac{2\sigma^2+1}{m} }  \diag\paren{ \zeta_{\sigma^2}^{1/2}(\mb y) }   \brac{  (\mb A \mb w) \odot \exp\paren{ - \im \phi(\mb A\mb x) }  - (\ol{\mb A\mb w}) \odot \exp\paren{ \im \phi (\mb A\mb x) }  }  }{}^2 \\
   \leq \;& \frac{1}{2} \abs{ \mb w^* \mb P_{\mb x^\perp} \mb M \mb P_{\mb x^\perp}\mb w } + \frac{1}{2} \abs{ \frac{2\sigma^2+1}{m} \mb w^\top \mb P_{\mb x^\perp}^\top \mb A^\top \diag\paren{ \zeta_{\sigma^2}(\mb A\mb x) \psi(\mb A\mb x) } \mb A \mb P_{\mb x^\perp} \mb w } \\
   \leq \;& \frac{1}{2} \norm{ \bb E\brac{\mb H} }{} + \frac{1}{2} \norm{ \mb H - \bb E\brac{\mb H}  }{}  
    + \frac{1}{2} \abs{ \frac{2\sigma^2+1}{m} \mb w^\top \mb P_{\mb x^\perp}^\top \mb A^\top \diag\paren{ \zeta_{\sigma^2}(\mb A\mb x) \psi(\mb A\mb x) } \mb A \mb P_{\mb x^\perp} \mb w }
\end{align*}
holding for all $\mb w\in\bb {CS}^{n-1}$ with $\mb w \perp \mb x$,
where $\mb M$ and $\mb H$ are defined in \eqref{eqn:M-def-main} and \eqref{eqn:H-def-main}, and $\psi(t) = \paren{\ol{t}/\abs{t} }^2 $. By Lemma \ref{lem:M-H-expectation}, we know that 
\begin{align}
   \norm{\bb E\brac{\mb H}  }{} \;=\;  \norm{ \mb P_{\mb x^\perp} }{} \;\leq\; 1.
\end{align}
By \Cref{thm:M-H-concentration}, we know that for any $\delta>0$, whenever $m \geq C_1\delta^{-2} \norm{\mb C_{\mb x}}{}^2 n \log^4n$, we have
\begin{align*}
   \norm{ \mb H - \bb E\brac{\mb H} }{} \;\leq\; \delta,	
\end{align*}
with probability at least $1 - c_1m^{-c_2\log^3n}$. In addition, Lemma \ref{lem:cross-term-concentration} implies that for any $\delta>0$, when $m \geq C_2\delta^{-2} n \log^4 n$ for some constant $C_2>0$, we have
\begin{align*}
   	\abs{ \frac{2\sigma^2+1}{m} \mb w^\top \mb P_{\mb x^\perp}^\top \mb A^\top \diag\paren{ \zeta_{\sigma^2}(\mb A\mb x) \psi(\mb A\mb x) } \mb A \mb P_{\mb x^\perp} \mb w } \;\leq\; (1+\delta)\Delta_\infty(\eps)  + (1+\eps)\delta,
\end{align*}
holds with probability at least $1- 2m^{-c_3\log^3 n}$ for some constant $c_3>0$. Combining the results above, we obtain
\begin{align*}
	\norm{ \sqrt{\frac{2\sigma^2+1}{m} } \diag\paren{ \zeta_{\sigma^2}^{1/2}(\mb y) } \Im\paren{ \mb A\mb w \odot \exp\paren{ - \im \phi(\mb A\mb x) } }  }{}^2 \;	\leq\;\frac{1 +(2+\eps)\delta + (1+\delta)\Delta_\infty(\eps) }{2}.
\end{align*}
Finally, by using Lemma \ref{lem:Delta-infty-bound}, when $\sigma^2 = 0.51$ and $\eps = 0.2$, we have
\begin{align*}
   	\norm{ \sqrt{\frac{2\sigma^2+1}{m} } \diag\paren{ \zeta_{\sigma^2}^{1/2}(\mb y) } \Im\paren{ \mb A\mb w \odot \exp\paren{ - \im \phi(\mb A\mb x) } }  }{}^2 \;	\leq\;\frac{1 + 2.2 \delta + 0.404(1+\delta) }{2},
\end{align*}
as desired.
\end{proof}

\begin{lemma}\label{lem:cross-term-concentration}
For a fixed scalar $\eps>0$, let $\Delta_\infty(\eps) $ be defined as \eqref{eqn:Delta-infty}. For any $\delta>0$, whenever $m \geq C\delta^{-2} n \log^4 n$, with probability at least $1 - m^{- c\log^3 n}$ for all $\mb w \in \bb {CS}^{n-1}$ with $\mb w \perp \mb x$, we have
\begin{align*}
   \abs{ \frac{2\sigma^2+1}{m} \mb w^\top \mb A^\top \diag\paren{ \zeta_{\sigma^2}(\mb A\mb x) \psi(\mb A\mb x) } \mb A \mb w }	\;\leq\; \paren{1+\delta} \Delta_{\infty}(\eps) + (1+\eps) \delta.
\end{align*}
\end{lemma}

\begin{proof}
First, let $\mb g = \mb a$ and let $\mb \delta \sim \mc {CN}(\mb 0,\mb I)$ independent of $\mb g$, given a small scalar $\eps>0$, we have
\begin{align*}
   	&\abs{ \frac{2\sigma^2+1}{m} \mb w^\top  \mb R_{[1:n]} \mb C_{\mb g}^\top \diag\paren{ \zeta_{\sigma^2}(\mb g \conv \mb x) \psi(\mb g \conv \mb x) } \mb C_{\mb g} \mb R_{[1:n]}^\top \mb w } \\
   \leq \;& \abs{ \frac{2\sigma^2+1}{m} \mb w^\top  \mb R_{[1:n]} \mb C_{\mb g}^\top \diag\paren{ \zeta_{\sigma^2}(\mb g \conv \mb x) \psi(\mb g \conv \mb x) - (1+ \eps) \bb E_{\mb \delta}\brac{ \psi( \paren{\mb g - \mb \delta}\conv \mb x )  } } \mb C_{\mb g} \mb R_{[1:n]}^\top \mb w } \\
   & + (1+\eps) \abs{ \frac{2\sigma^2+1}{m} \mb w^\top  \mb R_{[1:n]} \mb C_{\mb g}^\top \diag\paren{ \bb E_{\mb \delta}\brac{ \psi( \paren{\mb g - \mb \delta}\conv \mb x )  } } \mb C_{\mb g} \mb R_{[1:n]}^\top \mb w } \\
   \leq \;& \frac{\Delta_\infty(\eps) }{m} \norm{ \mb R_{[1:n]}\mb C_{\mb g}^* }{}^2 + (1 +\eps )\abs{ \underbrace{ \frac{2\sigma^2+1}{m} \mb w^\top  \mb R_{[1:n]} \mb C_{\mb g}^\top \diag\paren{ \bb E_{\mb \delta}\brac{ \psi( \paren{\mb g - \mb \delta}\conv \mb x )  } } \mb C_{\mb g} \mb R_{[1:n]}^\top  \mb w  }_{ \mc D(\mb g,\mb w) } }.
\end{align*}
By Corollary \ref{cor:circ-spectral-norm}, for any $\delta>0$, whenever $ m \geq C_0 \delta^{-2} n \log^4n$ for some constant $C_0>0$, we have
\begin{align*}
  \norm{ \mb R_{[1:n]}\mb C_{\mb g}^* }{}^2 \;\leq\; (1+\delta)m	
\end{align*}
with probability at least $1- m^{-c_0 \log^3m }$ for some constant $c_0>0$. Next, let us define a decoupled version of $\mc D(\mb g,\mb w)$,
\begin{align}
   \mc Q_{dec}^{\mc D}(\mb g^1,\mb g^2,\mb w) \;=\; \frac{2\sigma^2+1}{m} \mb w^\top  \mb R_{[1:n]} \mb C_{\mb g^1}^\top  \diag\paren{  \psi(\mb g^2 \conv \mb x) } \mb C_{\mb g^1} \mb R_{[1:n]}^\top  \mb w.
\end{align}
where $\mb g^1 =\mb g + \mb \delta $ and $\mb g^2 = \mb g - \mb \delta$. Then by using the fact that $\mb w \perp \mb x$, we have
\begin{align*}
  \abs{\bb E_{\mb \delta}\brac{ \mc Q_{dec}^{\mb D}(\mb g^1,\mb g^2,\mb w) } } \;& =\; \abs{ \frac{2\sigma^2+1}{m} \mb w^\top \mb R_{[1:n]} \mb C_{\mb g}^\top \diag\paren{ \bb E_{\mb \delta}\brac{ \psi\paren{(\mb g - \mb \delta)\conv \mb x}  } } \mb C_{\mb g} \mb R_{[1:n]}^\top \mb w  }.
\end{align*}
Then for any positive integer $p \geq 1$, by Jensen's inequality and \Cref{thm:moments-T-2}, we have
\begin{align*}
   \norm{ \sup_{\norm{\mb w}{}=1,\; \mb w \perp \mb x } \abs{ \mc D(\mb g,\mb w) }  }{L^p} 
   \;&\leq\; \norm{ \sup_{\norm{\mb w}{}=1} \abs{ \mc Q_{dec}^{\mc D}(\mb g^1,\mb g^2,\mb w)  }  }{L^p}\\
   \;&\leq\;  C_{\sigma^2} \paren{  \sqrt{\frac{n}{m}}  \log^{3/2} n\log^{1/2}m +  \sqrt{\frac{n}{m}} \sqrt{p} +  \frac{n}{m} p  },
\end{align*}
where $C_{\sigma^2}>0$ is some positive constant depending on $\sigma^2$, and we used the fact that $\norm{ \psi(\mb g^2 \conv \mb x)}{\infty}\leq 1$ holds uniformly for all $\mb g^2$. Thus, by Lemma \ref{lem:moments-tail-bound}, then for any $\delta>0$, whenever $m \geq C_1 \delta^{-2} n \log^3 n \log m $, we have
\begin{align*}
   \sup_{\norm{\mb w}{}=1,\; \mb w \perp \mb x } \abs{ \mc D(\mb g,\mb w) } \;\leq\; \delta
\end{align*}
\edited{holding with probability at least $1 - m^{-c_1\log^3m }$. Combining the results above completes the proof.}
\end{proof}

\subsection{Bounding $\Delta_\infty(\eps)$}

\edited{Given $h(t)$ and $\Delta_\infty (\eps)$ introduced in \eqref{eqn:h-def-main} and \eqref{eqn:Delta-infty}, we prove the following results.}

\begin{figure*}[!htbp]
\centering
%\captionsetup{font=normalsize,labelfont={bf,sf}}
\captionsetup[sub]{font=small,labelfont={bf,sf}}
\centering
\begin{minipage}[c]{0.48\textwidth}
\subcaption{Plot of functions with $\sigma^2 = 0.51$, $\eps = 0.2$}
\centering
	\includegraphics[width = \linewidth]{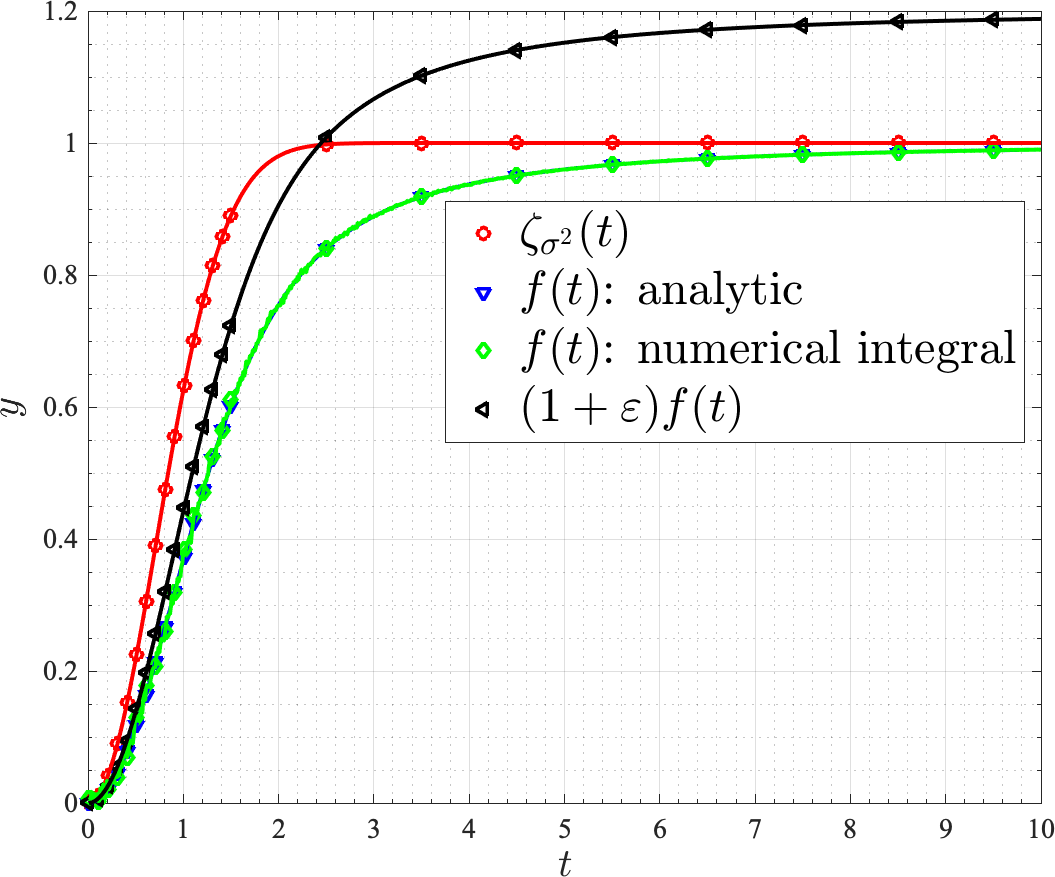}
\end{minipage}
\hfill
\begin{minipage}[c]{0.48\textwidth}
\subcaption{Plot of function difference $\zeta_{\sigma^2}(t)- (1+\eps)h(t)$}
\label{fig:g-func-diff}
\centering
	\includegraphics[width = \linewidth]{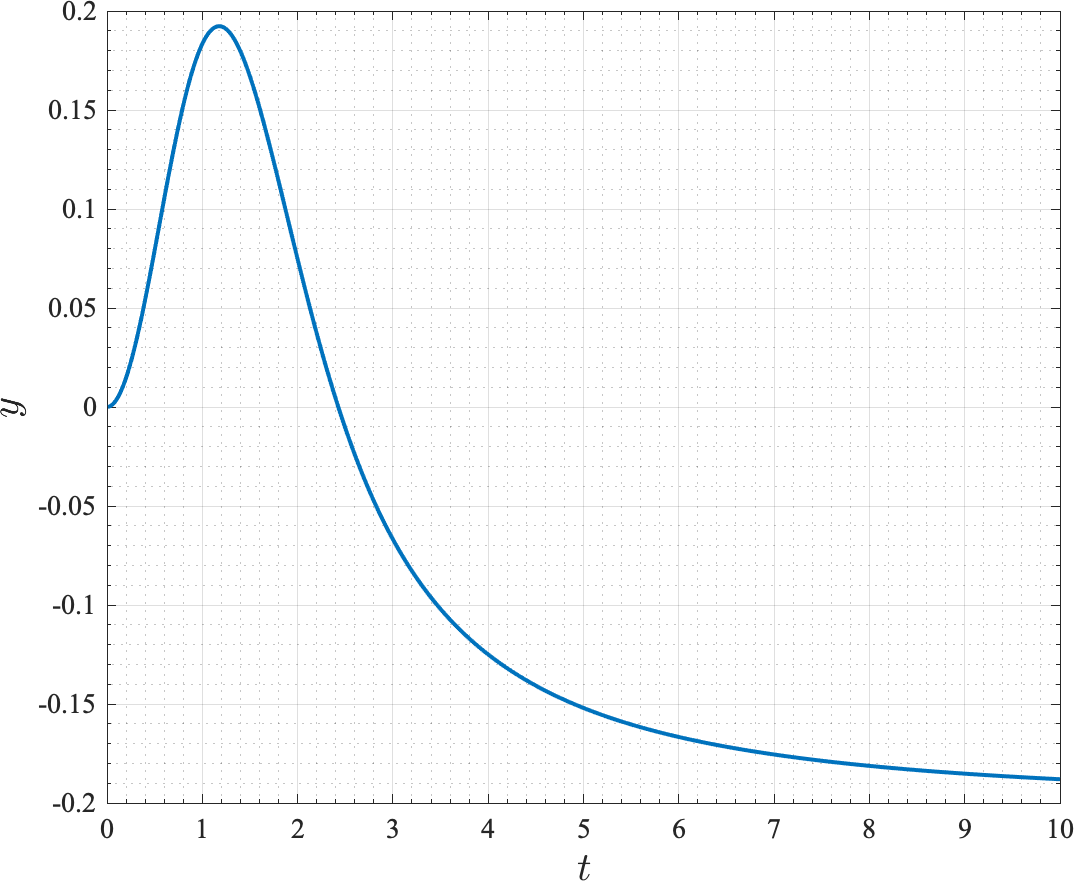}
\end{minipage}
\caption{\textbf{Computer simulation of the functions $\zeta_{\sigma^2}(t)$ and $h(t)$.} (a) displays the functions $\zeta_{\sigma^2}(t)$ and $h(t)$ with $\sigma^2 = 0.51$; (b) shows differences two function $\zeta_{\sigma^2}(t) - (1+\eps) h(t)$ with $\eps = 0.2$. }
\label{fig:func}
\end{figure*}

\begin{lemma}\label{lem:Delta-infty-bound}
Given $\sigma^2 = 0.51$ and $\eps = 0.2$, we have
\begin{align*}
   \Delta_\infty(\eps) \leq 0.404.
\end{align*}
\end{lemma}

\begin{proof}
First, by Lemma \ref{lem:f-expecation}, notice that the function $h(t)$ can be decomposed as
\begin{align*}
   h(t) = g(t)\psi(t)	
\end{align*}
where $g(t): \bb C\mapsto [0,1)$ is rotational invariant with respect to $t$. Since $\zeta_{\sigma^2}(t)$ is also rotational invariant with respect to $t$, it is enough to consider the case when $t\in [0,+\infty)$ , and bounding the following quantity
\begin{align*}
   \sup_{t\in [0,+\infty)} \abs{ (1+\eps)h(t) - \zeta_{\sigma^2}(t) }.	
\end{align*}
Lemma \ref{lem:f-integral} implies that
\begin{align*}
	h(t) \;= \; \bb E_{s\sim \mc {CN}(0,1)} \brac{ \psi(t+s) } \; =\; \begin{cases}
 	1-t^{-2} + t^{-2} e^{-t^2} & t >0, \\
 	0 & t=0.
 \end{cases}
\end{align*}
\edited{When $t =0$, we obtain} that $\abs{(1+\eps)h(t) - \zeta_{\sigma^2}(t)} =0$. For $t>0$, when $\eps = 0.2$ and $\sigma^2 = 0.51$, we have
\begin{align*}
   \zeta_{\sigma^2}(t) - (1+\eps) h(t)	 = -0.2 - e^{-\frac{t^2}{1.02} } +1.2 t^{-2} - 1.2 t^{-2} e^{-t^2}.
\end{align*}
\edited{From Lemma \ref{lem:func-approx}, we can prove that $\norm{\zeta_{\sigma^2}(t) - (1+\eps) h(t)}{L^\infty} \leq 0.2$ by a tight approximation of the function $\zeta_{\sigma^2}(t) - (1+\eps) h(t)$.} Therefore, we have
\begin{align*}
   \Delta_\infty(\eps) \;=\; (1+2\sigma^2) \norm{\zeta_{\sigma^2}(t) - (1+\eps) h(t) }{ L^\infty } \;\leq\; 0.2\times (1+2\times 0.51) \;=\; 0.404,
\end{align*}
when $\sigma^2 = 0.51$ and $\eps = 0.2$.
\end{proof}

\edited{
\begin{lemma}\label{lem:func-approx}
For $t\geq 0$, when $\eps = 0.2$ and $\sigma^2 = 0.51$, we have
\begin{align*}
   \abs{ \zeta_{\sigma^2}(t) - (1+\eps) h(t) } \;\leq\; 0.2 .
\end{align*}
\end{lemma}
\begin{proof}
Given $\eps = 0.2$ and $\sigma^2 = 0.51$, we have
\begin{align*}
g(t) \;=\; \zeta_{\sigma^2}(t) - (1+\eps) h(t)\;=\; -0.2 - e^{-\frac{t^2}{1.02} } +1.2 t^{-2} - 1.2 t^{-2} e^{-t^2}	.
\end{align*}
When $t =0$, we have $\abs{g(t)} =0$. When $t>10$, we have
\begin{align*}
   g'(t) \;=\; \frac{t}{0.51} e^{-\frac{t^2}{1.02} } - 2.4t^{-3}\paren{1 - e^{-t^2}} +2.4t^{-1}e^{-t^2}\;\leq\; 0. 
\end{align*}
So the function $g(t)$ is monotonically decreasing for $t \geq 10$. As $\lim_{t\rightarrow +\infty  }g(t) = -0.2$ and $\abs{g(10)}<0.2$, we have 
\begin{align*}
   \abs{g(t)} \leq 0.2, \quad \forall \; t >10.	
\end{align*}
For $0<t \leq 10$, since $h(t)$ is continuously differentiable, \Cref{fig:g-func-diff} implies that $\abs{g(t)} \leq 0.2$ for all $t \in(0,10)$ (we omit the tedious proof here). 
\end{proof}
}

\begin{lemma}\label{lem:f-expecation}
Let $\psi(t) = \paren{\ol{t}/\abs{t} }^2 $	, then we have
\begin{align*}
h(t) \;=\; \bb E_{s\sim \mc {CN}(0,1) }\brac{ \psi(t+s) } \;=\; g(t) \psi(t).
\end{align*}
where $g(t) :\bb C \mapsto [0,1)$, such that
\begin{align*}
g(t) \;=\; \bb E_{v_1,v_2 \sim \mc N(0,1/2) }\brac{ \frac{ (\abs{t} + v_1)^2  - v_2^2 }{ (\abs{t} + v_1)^2  + v_2^2 } },
\end{align*}
where $ v_1 \sim \mc N(0,1/2)$, and $v_2 \sim \mc N(0,1/2)$.
\end{lemma}
\begin{proof}
By definition, we know that
\begin{align*}
	\bb E_{s\sim \mc {CN}(0,1) }\brac{ \psi(t+s) } \;=\; \bb E_s \brac{ \paren{ \frac{\ol{t+s}}{\abs{t+s}} }^2 } 
	\;=\; \underbrace{ \bb E_s \brac{ \paren{ \frac{t}{\abs{t}}  \frac{\ol{t+s}}{\abs{t+s}} }^2 } }_{g(t)} \psi(t) .
\end{align*}
Next, we estimate $g(t)$ and show that it is indeed real. We decompose the random variable $s$ as
\begin{align*}
 s \;=\; \Re\paren{ \frac{\ol{t}s}{\abs{t}}  }\frac{t}{\abs{t}} + \im \Im \paren{ \frac{\ol{t}s}{\abs{t}} }	 \frac{t}{\abs{t}} \;=\; v_1 \frac{t}{\abs{t}} + \im v_2 \frac{t}{\abs{t}},
\end{align*}
	where $v_1 = \Re\paren{ \frac{\ol{t}s}{\abs{t}}  }$ and $v_2 = \Im \paren{ \frac{\ol{t}s}{\abs{t}} } $ are the real and imaginary parts of a complex Gaussian variable $\ol{t}s/\abs{t} \sim \mc {CN}(0,1)$. By rotation invariant property, we have $v_1 \sim \mc N(0,1/2)$ and $v_2 \sim \mc N(0,1/2)$, and $v_1$ and $v_2$ are independent. Thus, we have
    \begin{align*}
    	h(t) \;=\; \bb E_s \brac{ \paren{\frac{ \abs{t} + v_1 - \im v_2  }{ \abs{t+s} } }^2  } 
    	 \;&=\; \bb E_s \brac{ \frac{ \paren{ \abs{t} + v_1   }^2- v_2^2  }{ \abs{t+s}^2 }  } - 2\im \bb E_s\brac{ \frac{  \paren{ \abs{t} +v_1  }  v_2 }{ \abs{t+s}^2 } } \\
    	\;& =\; \bb E_{v_1,v_2} \brac{ \frac{ \paren{ \abs{t} + v_1 }^2- v_2^2  }{ \paren{\abs{t}+v_1}^2 + v_2^2 }  } - 2\im \bb E_{v_1, v_2}\brac{ \frac{  \paren{ \abs{t} + v_1  } v_2 }{ \paren{ \abs{t}+v_1 }^2 + v_2^2 } }.
    	\end{align*}
    	We can see that $\frac{  \paren{ \abs{t} + v_1  } v_2 }{ \paren{ \abs{t}+v_1 }^2 + v_2^2 } $
        is an odd function of $v_2$. Therefore, the expectation of $\frac{  \paren{ \abs{t} + v_1  } v_2 }{ \paren{ \abs{t}+v_1 }^2 + v_2^2 } $ with respect to $v_2$ is zero. Thus, we have
    	\begin{align*}
    		g(t) \;=\; \bb E_s \brac{ \paren{ \frac{t}{\abs{t}}  \frac{\ol{t+s}}{\abs{t+s}} }^2 } = \bb E_{v_1 ,v_2 \sim \mc N(0,1/2) }\brac{ \frac{ (\abs{t} + v_1)^2  - v_2^2 }{ (\abs{t} + v_1)^2  + v_2^2 } },
    	\end{align*}   
    	which is real.
\end{proof}

\begin{lemma}\label{lem:f-integral}
For $t \in[0,+\infty)$, we have
\begin{align}
	f(t) \;= \; \bb E_{s\sim \mc {CN}(0,1)} \brac{ \psi(t+s) } \; =\; \begin{cases}
 	1-t^{-2} + t^{-2} e^{-t^2} & t >0, \\
 	0 & t=0.
 \end{cases}
\end{align}
\end{lemma}

\begin{proof}
Let $s_r = \Re\paren{s}$ and $s_i = \Im\paren{s}$, and let $s = r \exp\paren{ \im \theta }$  with $r = \abs{s}$ and $\exp\paren{ \im \theta } = s/\abs{s}$. We observe
\begin{align*}
\bb E_{s\sim \mc {CN}(0,1)} \brac{ \psi(t+s) } \;&=\; \frac{1}{\pi}  \int_{s_r =-\infty}^{+\infty} \int_{s_i = -\infty}^{+\infty} \frac{ (\ol{s_r+\im s_i})^2 }{ \abs{s_r+\im s_i }^2 } e^{- \abs{ s_r+ \im s_i -t}^2 }ds_r d s_i \\
\;&=\; \frac{1}{\pi} \int_{r = 0}^{+\infty} \int_{\theta = 0}^{2\pi} e^{ - \im 2\theta } e^{ -r^2 - t^2 } e^{ 2r t \cos \theta } r d\theta d r \\
\;&=\; \frac{1}{\pi}e^{-t^2} \int_{r=0}^{+\infty} \int_{\theta=0}^{2\pi} \cos (2\theta) r e^{-r^2} e^{ 2rt \cos \theta } d\theta dr \\
\;&=\; \frac{2}{\pi}e^{-t^2} \int_{r=0}^{+\infty} \int_{\theta=0}^{\pi} \cos (2\theta) r e^{-r^2} \cosh \paren{ 2rt \cos \theta } d\theta dr\\
\;&=\; \frac{2}{\pi}e^{-t^2} \int_{r=0}^{+\infty} \int_{\theta=0}^{\pi/2} \cos (2\theta) r e^{-r^2} \brac{\cosh \paren{ 2rt \cos \theta } - \cosh(2rt \sin \theta) }d\theta dr
\end{align*}
where the third equality uses the fact that the integral of \edited{an odd function} is zero. By using Taylor expansion of $\cosh(x)$, and by using the \emph{dominated convergence theorem} to exchange the summation and integration, we observe
\begin{align*}
	&\bb E_{s\sim \mc {CN}(0,1)} \brac{ \psi(t+s) } \\
	=\;& \frac{2}{\pi}e^{-t^2} \int_{r=0}^{+\infty} \int_{\theta=0}^{\pi} \cos (2\theta) r e^{-r^2} \sum_{k=0}^{+\infty} \brac{ \frac{ \paren{2rt \cos \theta }^{2k} }{(2k)!} - \frac{ \paren{2rt \sin \theta }^{2k} }{(2k)!} } d\theta dr \\
	=\;& \frac{2}{\pi}e^{-t^2} \int_{r=0}^{+\infty} \int_{\theta=0}^{\pi} \cos (2\theta)   \sum_{k=0}^{+\infty} \brac{ \frac{ \paren{2t \cos \theta }^{2k} r^{2k+1} e^{-r^2}}{(2k)!} - \frac{ \paren{2t \sin \theta }^{2k} r^{2k+1} e^{-r^2} }{(2k)!} } d\theta dr \\
	=\;& \frac{2}{\pi}  e^{-t^2} \sum_{k=0}^{+\infty} \frac{ (2t)^{2k} }{(2k)! } \int_{r=0}^{+\infty} r^{2k+1} e^{-r^2} dr \brac{ \int_{\theta = 0}^{\pi} \cos (2\theta) \cos^{2k}\theta  d\theta - \int_{\theta = 0}^{\pi} \cos (2\theta) \sin^{2k}\theta  d\theta }.
\end{align*}
We have the integrals
\begin{align*}
\int_{r=0}^{+\infty} r^{2k+1} e^{-r^2} dr \;&=\; \frac{\Gamma(k+1)}{2}, \\
 \int_{\theta=0}^\pi \cos(2\theta) \cos^{2k}\theta d\theta \;&=\; \frac{\sqrt{\pi} }{2}  \frac{k \Gamma(k+1/2)}{\Gamma(k+2)}, \\
 \int_{\theta=0}^\pi \cos(2\theta) \sin^{2k}\theta d\theta \;&=\; -\frac{\sqrt{\pi} }{2}  \frac{k \Gamma(k+1/2)}{\Gamma(k+2)}
\end{align*}
holds for any integer $k \geq 0$, where $\Gamma(k)$ is the Gamma function such that
\begin{align*}
\Gamma(k+1) = k!,\quad \Gamma(k+1/2) = \frac{(2k)!}{ 4^k k!} \sqrt{\pi}.	
\end{align*}
Thus, for $t>0$, we have
\begin{align*}
	\bb E_{s\sim \mc {CN}(0,1)} \brac{ \psi(t+s) } \;&=\; \frac{2}{\pi} e^{-t^2} \sum_{k=0}^{+\infty} \frac{ (2t)^{2k} }{ (2k)! } \times \frac{ \Gamma(k+1) }{2} \times \sqrt{\pi}  \frac{k \Gamma(k+1/2)}{\Gamma(k+2)} \\
	\;&=\; e^{-t^2} \sum_{k=0}^{+\infty}  \frac{k t^{2k} }{ (k+1)! } = e^{-t^2} \paren{ \sum_{k=0}^{+\infty}  \frac{  t^{2k} }{ k!} -  \sum_{k=0}^{+\infty} \frac{t^{2k}}{(k+1)!} } \\
	\;&=\; e^{-t^2} \brac{ e^{t^2} - t^{-2}\paren{ \sum_{k=0}^{+\infty} \frac{t^{2k} }{ k! } -1  } } = 1 -t^{-2} +t^{-2} e^{-t^2}.
\end{align*}
When $t=0$, by using L'Hopital's rule, we have
\begin{align*}
\eta(0) = \lim_{t \rightarrow 0}	\bb E_{s\sim \mc {CN}(0,1)} \brac{ \psi(t+s) } \;=\; \lim_{t \rightarrow 0} \brac{ 1  + \frac{1 - e^{t^2}}{  t^2 e^{t^2} } } \;=\; 1 + \lim_{t\rightarrow 0} \frac{-1}{1+t} \;=\; 0.
\end{align*}
We complete the proof.
\end{proof}

\section{Acknowledgement}
This work was partially supported by the grants NSF CCF 1527809 and NSF IIS 1546411, the grants from the European Union's Horizon 2020 research and innovation program under grant agreement No.\ 646804-ERCCOGBNYQ, and the grant from the Israel Science Foundation under Grant No.\ 335/14. QQ thanks the generous support of the Microsoft graduate research fellowship, and Moore-Sloan fellowship. We would like to thank Shan Zhong for the helpful discussion for real applications and providing the antenna data for experiments, and we thank Ju Sun and Han-Wen Kuo for helpful discussion and input regarding the analysis of this work.

\newpage 

{\small
\medskip
\bibliographystyle{ieeetr}
\bibliography{pr,ncvx}
}

\newpage 

\appendices

In the appendix, we provide details of proofs for some supporting results. Appendix \ref{app:tools} summarizes basic tools used throughout the analysis. In Appendix \ref{app:moments-circulant-matrix}, we provide results of bounding the suprema of chaos processes for random circulant matrices. In Appendix \ref{app:decoupling}, we present concentration results for suprema of some dependent random processes via decoupling.

\section{Elementary Tools and Results}\label{app:tools}
\begin{lemma}\label{lem:phase-diff}
	\edited{Given a fixed number $\rho \in (0,1)$}, for any $z,\; z' \in \bb C$, we have
	\begin{align}
		\abs{ \exp\paren{ \im \phi(z'+z)  } - \exp\paren{ \im \phi(z')  } } \leq 2 \indicator{ \abs{z} \geq \rho \abs{z'}  } + \frac{1}{1-\rho} \abs{ \Im\paren{z/z' } }.
	\end{align}
\end{lemma}

\begin{proof}
	See the proof of Lemma 3.2 of \cite{waldspurger2016phase}.
\end{proof}

\begin{lemma}\label{lem:phase-approx}
Let $\rho \in (0,1)$, for any $\mb z \in \bb C$ with $\abs{z} \leq \rho $, we have
\begin{align}
   \abs{ 1 - \exp\paren{ \im \phi(1+z) } + \im \Im(z) } \;\leq\; \frac{2- \rho}{ (1-\rho)^2 } \abs{z}^2.		
\end{align}
\end{lemma}

\begin{proof}
For any $t \in \bb R^{+}$, let $g(t) = \sqrt{ \paren{1+\Re(z) }^2 + t^2 } $, then 
\begin{align*}
	g'(t) \;=\;  \frac{ t}{\sqrt{ \paren{1+\Re(z) }^2 +t^2 }} \;\leq\; \frac{t}{\abs{1+\Re(z)}}. 
\end{align*}
Hence, for any $z \in \bb C$ with $\abs{z}\leq \rho$, we have
\begin{align*}
  \abs{  \abs{ 1 + z} - \paren{ 1 + \Re(z) } } \;&=\; \abs{  \sqrt{ \paren{ 1 +  \Re(z) }^2 + \Im^2(z)  }  - \paren{  1+ \Re(z) } }  \\
  \;&=\;  \abs{ g\paren{ \Im(z) } - g(0) }  \leq \frac{ \Im^2(z) }{  \abs{1+ \Re(z)} } \leq \frac{1}{1-\rho} \Im^2(z) .
\end{align*}
Let $f(z) = 1 - \exp\paren{ \im \phi(1+z) }$. By using the estimates above, we observe
\begin{align*}
   \abs{ f(z) + \im \Im(z) } \;&=\; \abs{ \frac{ \abs{1+z} -(1+z) }{ \abs{1+z} } + \im \Im\paren{z}   } \\
   \;&=\; \frac{1}{\abs{1+z}} \abs{  \abs{1+z} - (1+z) + \im \Im(z) \abs{1+z} } \\
   \;&\leq\;  \frac{1}{\abs{1+z}} \paren{ \abs{\Im(z)} \abs{  1- \abs{1+z} } + \abs{  \abs{1+z} - \paren{1 + \Re(z) } }  } \\
   \;&\leq\; \frac{1}{\abs{1+z}} \paren{ \abs{z} \abs{\Im(z)} + \frac{1}{1-\rho} \Im^2(z)  } \leq \frac{ 2-\rho }{ \paren{ 1-\rho }^2 } \abs{z}^2.
\end{align*}
\end{proof}

\begin{lemma}[Gaussian Concentration Inequality]\label{lem:gauss-concentration}
Let $\mb w \in \bb R^n $ be a standard Gaussian random variable $\mb w \sim \mc N(\mb 0,\mb I)$, and let $g: \bb R^n \mapsto \bb R $ denote an $L$-Lipschitz function. Then for all $t>0$, 
	\begin{align*}
		\bb P\paren{ \abs{ g(\mb w) - \bb E \brac{g(\mb w)} }\geq t } \;\leq\; 2\exp \paren{ -t^2/(2L^2)}.
	\end{align*}
Moreover, if $\mb w \in \bb C^n $ with $\mb w \sim \mc {CN}(\mb 0,\mb I)$, and $g: \bb C^n \mapsto \bb R$ is $L$-Lipschitz, then the inequality above still holds.
\end{lemma}

\begin{proof}
	The result for real-valued Gaussian random variables is standard, see \cite[Chapter 5]{boucheron2013concentration} for a detailed proof. For the complex case, let 
	\begin{align*}
	   \mb v  \;=\; \underbrace{ \frac{1}{\sqrt{2}} \begin{bmatrix} \mb I & \im \mb I \end{bmatrix} }_{ h } \begin{bmatrix} \mb v_r \\
	    \mb v_i \end{bmatrix}, \qquad \mb v_r,\;\mb v_i \;\sim_{i.i.d.}\; \mc N\paren{\mb 0,\mb I}.
	\end{align*}
    By composition theorem, we know that $g' \circ h:\bb R^{2n} \mapsto \bb R$ is $L$-Lipschitz. Therefore, by applying the Gaussian concentration inequality for $g'\circ h $ and $\begin{bmatrix} \mb v_r \\
	    \mb v_i \end{bmatrix}$, we get the desired result.
\end{proof}

\begin{theorem}[Gaussian tail comparison for vector-valued functions, Theorem 3, \cite{ledoux2007measure}] \label{thm:gaussian-vector-concentration}
Let $\mb w\in \bb R^n$ be standard Gaussian variable $ \mb w \sim \mc N(\mb 0,\mb I)$, and let $f: \bb R^n \mapsto \bb R^\ell$ be an $L$-Lipschitz function. Then for any $t>0$, we have
\begin{align*}
  \bb P\paren{ \norm{ f(\mb w) - \bb E\brac{ f(\mb w) } }{} \geq t }	 \;\leq\; e \bb P \paren{ \norm{\mb v }{} \geq \frac{t}{L} },
\end{align*}
where $\mb v \in \bb R^\ell$ such that $\mb v \sim \mc {N}(\mb 0,\mb I)$. Moreover, if $\mb w\in \bb C^n$ with $\mb w \sim \mc {CN}(\mb 0,\mb I)$ and $f:\; \bb C^n \mapsto \bb R^\ell$ is $L$-Lipschitz, then the inequality above still holds.
\end{theorem}
The proof is similar to that of Lemma \ref{lem:gauss-concentration}.

\begin{lemma}[Tail of sub-Gaussian Random Variables]\label{lem:sub-Gaussian}
	Let $X$ be a centered $\sigma^2$ sub-Gaussian random variable, such that
	\begin{align*}
	   \bb P\paren{ \abs{X} \geq t } \leq 2\exp\paren{ -\frac{t^2}{2\sigma^2} },
	\end{align*}
    then for any integer $p\geq 1$, we have
    \begin{align*}
       \bb E\brac{\abs{X}^p } \leq 	\paren{2\sigma^2}^{p/2} p \Gamma(p/2).
    \end{align*}
    In particular, we have
    \begin{align*}
       \norm{X}{L^p} = \paren{ \bb E\brac{\abs{X}^p } }^{1/p} \leq \sigma e^{1/e}	\sqrt{p}, \quad p \geq 2,
    \end{align*}
    and $\bb E\brac{\abs{X} }  \leq \sigma \sqrt{2\pi} $.
\end{lemma}

\begin{lemma}[Sub-exponential tail bound via moment control] \label{lem:moments-tail-bound}
Suppose	$X$ is a centered random variable satisfying 
\begin{align*}
   \paren{\bb E\brac{ \abs{X}^p } }^{1/p} \leq \alpha_0 + \alpha_1 \sqrt{p} + \alpha_2 p,\quad \text{for all } p \geq p_0 
\end{align*}
for some $\alpha_0,\alpha_1,\alpha_2 ,p_0>0$. Then, for any $u \geq p_0$, we have
\begin{align*}
   \bb P \paren{  \abs{X} \geq e(\alpha_0 + \alpha_1 \sqrt{u} + \alpha_2 u)  }	\leq 2\exp\paren{ -u }.
\end{align*}
This further implies that for any $t> \alpha_1 \sqrt{p_0} + \alpha_2 p_0$, we have
\begin{align*}
 \bb P \paren{  \abs{X} \geq c_1 \alpha_0 + t } \leq 2\exp\paren{ -c_2 \min\Brac{ \frac{t^2}{\alpha_1^2}, \frac{t}{\alpha_2 }  } },		
\end{align*}
for some positive constants $c_1,\;c_2>0$.
\end{lemma}

\begin{proof}
The first inequality directly comes from Proposition 2.6 of \cite{krahmer2014suprema} via Markov inequality, also see Proposition 7.11 and Proposition 7.15 of \cite{foucart2013mathematical}. For the second, let $t = \alpha_1 \sqrt{u} +\alpha_2 u$, if $ \alpha_1 \sqrt{u} \leq \alpha_2 u$, then 
\begin{align*}
   t = \alpha_1 \sqrt{u} + \alpha_2 u  \leq 2 	\alpha_2 u  \; \Rightarrow \; u \geq \frac{t}{2\alpha_2 }.
\end{align*}
Otherwise, similarly, we have $u \geq t^2/ (4\alpha_1^2)$. Combining the two cases above, we get the desired result.
\end{proof}

\edited{In the following, we describe a tail bound for a class of heavy-tailed random variables, whose moments are growing much faster than sub-Gaussian and sub-exponential.}

\begin{lemma}[Tail bound for heavy-tailed distribution via moment control]\label{lem:moments-tail-bound-1}
Suppose $X$ is a centered random variable satisfying
\begin{align*}
   \paren{ \bb E\brac{ \abs{X}^p } }^{1/p} \leq p \paren{ \alpha_0 + \alpha_1 \sqrt{p} + \alpha_2 p }, \quad \text{for all } p \geq p_0,
\end{align*}
for some $\alpha_0,\alpha_1,\alpha_2,p_0\geq 0$. Then, for any $u \geq p_0$, we have
\begin{align*}
   \bb P\paren{ \abs{X} \geq e u \paren{ \alpha_0 + \alpha_1 \sqrt{u} + \alpha_2 u }   } \leq 2 \exp\paren{- u}.	
\end{align*}
This further implies that for any $t> p_0 \paren{ \alpha_0 + \alpha_1 \sqrt{p_0} + \alpha_2 p_0 } $, we have
\begin{align*}
   \bb P\paren{ \abs{X} \geq c_1t } \leq 2 \exp\paren{ - c_2\min\Brac{ \sqrt{ \frac{t}{2(\alpha_1+\alpha_2) }}, \frac{t}{2\alpha_0}  } },	
\end{align*}
for some positive constant $c_1,c_2>0$.
\end{lemma}

\begin{proof}
The proof of the first tail bound is similar to that of 	Lemma \ref{lem:moments-tail-bound} by using Markov inequality. Notice that
\begin{align*}
   \bb P\paren{ \abs{X} \geq e u \paren{ \alpha_0 + (\alpha_1+ \alpha_2) u } } \;\leq\; 	\bb P\paren{ \abs{X} \geq e u \paren{ \alpha_0 + \alpha_1 \sqrt{u} + \alpha_2 u }   } \;\leq\; 2 \exp\paren{- u}.
\end{align*}
Let $ t = \alpha_0 u + \paren{ \alpha_1+\alpha_2 } u^2$, if $\alpha_0 u \leq \paren{ \alpha_1+\alpha_2 } u^2$, then
\begin{align*}
   t = 	\alpha_0 u+ \paren{\alpha_1+ \alpha_2 }u^2 \leq 2\paren{ \alpha_1+\alpha_2 }u^2 \;\Rightarrow \; u \geq \sqrt{ \frac{t}{2(\alpha_1+\alpha_2) }}.
\end{align*}
Otherwise, we have $u \geq t/(2\alpha_0)$. Combining the two cases above, we get the desired result.
\end{proof}

\begin{definition}[$d_2( \cdot), d_F(\cdot)$ and $\gamma_\beta$ functional] \label{def:gamma-beta-functional}
For a given set of matrices $\mc B$, we define
\begin{align*}
	d_F(\mc B) \;\doteq\; \sup_{\mb B \in \mc B} \norm{\mb B}{F}, \quad d_2(\mc B) \doteq \sup_{\mb B \in \mc B} \norm{\mb B}{},
\end{align*}
For a metric space $(T,d)$, an admissible sequence of $T$ is a collection of subsets of $T$, $\Brac{T_r: r>0 }$, such that for every $s>1$, $\abs{T_r} \leq 2^{2^r}$ and $\abs{T_0} = 1$. For $\beta \geq 1$, define the $\gamma_\beta$ functional by
	\begin{align*}
		\gamma_\beta (T,d) \;\doteq\; \inf \;\sup_{t\in T}\; \sum_{r= 0}^\infty 2^{r/\beta} d(t, T_r),
	\end{align*}
	where the infimum is taken with respect to all admissible sequences of $T$. In particular, for $\gamma_2$ functional of the set $\mc B$ equipped with distance $\norm{\cdot }{}$, \cite{talagrand2014} shows that
	\begin{align}
		\gamma_2(\mc B, \norm{ \cdot }{}) \;\leq\; c \int_0^{d_2(\mc B)} \log^{1/2} \mc N(\mc B, \norm{ \cdot }{}, \epsilon) d\epsilon,
	\end{align}
	where $\mc N(\mc B, \norm{\cdot }{}, \epsilon)$ is the covering number of the set $\mc B$ with diameter $\epsilon \in (0,1)$.
\end{definition}

\begin{theorem}[Theorem 3.5, \cite{krahmer2014suprema}]\label{thm:moments-bound}
	Let $\sigma_\xi^2\geq 1$ and $\mb \xi= \paren{ \xi_j }_{j=1}^n$, where $\Brac{\xi_j}_{j=1}^n$ are independent zero-mean, variance one, $\sigma_\xi^2$-subgaussian random variables, and let $\mc B$ be a class of matrices. Let us define a quantity
\begin{align}\label{eqn:quantity-C-B}
	\mc C_{\mc B}\paren{\mb \xi} \doteq \sup_{\mb B \in \mc B} \abs{ \norm{\mb B\mb \xi}{}^2 - \bb E\brac{ \norm{\mb B \mb \xi}{}^2 } }. 
\end{align}
For every $p \geq 1$, we have
	\begin{align*}
	   \norm{ \sup_{\mb B \in \mc B} \norm{\mb B\mb \xi }{} }{L^p} \;& \leq\; C_{\sigma_\xi^2} \brac{\gamma_2\paren{ \mc B, \norm{\cdot }{} } + d_F(\mc B) + \sqrt{p} d_2(\mc B) } \\
	   \norm{ \sup_{\mb B \in \mc B} \abs{ \norm{\mb B\mb \xi}{}^2 - \bb E\brac{ \norm{\mb B \mb \xi}{}^2 } }  }{L^p} \;&\leq\; C_{\sigma_\xi^2}\left\lbrace \gamma_2(\mc B, \norm{\cdot }{})\brac{ \gamma_2\paren{\mc B, \norm{ \cdot }{} } + d_F(\mc B)  } \right. \\
	   &\qquad \qquad \qquad \left. + \sqrt{p} d_2(\mc B)\brac{  \gamma_2\paren{ \mc B,\norm{\cdot }{} } + d_F(\mc B) } + p d_2^2(\mc B)\right\rbrace,
	\end{align*}
	where $C_{\sigma_\xi^2}$ is some positive numerical constant only depending on $\sigma_\xi^2$, and $d_2(\cdot), d_F(\cdot)$ and $\gamma_2(\mc B, \norm{\cdot}{} )$ are given in Definition \ref{def:gamma-beta-functional}.
\end{theorem}

\begin{comment}
\begin{theorem}[Theorem 3.1 of \cite{krahmer2014suprema}]\label{thm:toepliz-sup-tail}
      Let $\mc A$ be a set of matrices, and let $\mb \xi$ be a random vector whose entries $\xi_j$ are independent, mean-zero, variance one, and $\sigma_\xi^2$-subgaussian random variables. Set
      \begin{align*}
      	E &= \gamma_2(\mc A, \norm{\cdot }{}) \paren{ \gamma_2(\mc A, \norm{\cdot }{}) + d_F(\mc A) } + d_F(\mc A) d_2(\mc A),\\
      	V &= d_2(\mc A)\paren{ \gamma_2(\mc A, \norm{\cdot }{}) + d_F(\mc A) }, \quad U = d_2^2(\mc A).
      \end{align*}
      where $d_2(\cdot), d_F(\cdot)$ and $\gamma_2(\mc A, \norm{\cdot}{} )$ are defined in Def. \ref{def:gamma-beta-functional}. Then for $t>0$,
      \begin{align*}
          \bb P\paren{ \sup_{\mb A\in \mc A} \abs{ \norm{\mb A \mb \xi }{}^2 - \bb E \norm{\mb A \mb \xi }{}^2  } \geq c_1 E +t  } \leq 2 \exp \paren{ - c_2\min\Brac{ \frac{t^2}{V^2}, \frac{t}{U} } }
      \end{align*}
	  The constants $c_1,c_2$ depend only on $\sigma_\xi^2$.
\end{theorem}

\end{comment}

The following theorem establishes the \emph{restricted isometry property} (RIP) of the Gaussian random convolution matrix.

\begin{theorem}[Theorem 4.1, \cite{krahmer2014suprema} ]\label{thm:circ-rip}
	Let $\mb \xi \in \bb C^m$ be a random vector with $\xi_i \sim_{i.i.d.} \mc {CN}(0,1)$, and let $\Omega$ be a fixed subset of $[m]$ with $\abs{\Omega} = n$. Define a set $\mc E_s = \Brac{ \mb v \in \bb C^m \mid \norm{ \mb v}{0} \leq s  }$, and define a matrix
	\begin{align*}
		\mb \Phi = \mb R_\Omega \mb C_{\mb \xi}^* \in \bb C^{n \times m },
	\end{align*}
where $\mb R_{\Omega}: \bb C^m \mapsto \bb C^n$ is an operator that restrict a vector to its entries in $\Omega$. Then for any $s \leq m$, and $\eta,\delta_s \in (0,1)$ such that
	\begin{align*}
		n \geq C \delta_s^{-2} s \log^2 s \log^2 m,
	\end{align*}
	the partial random circulant matrix $\mb \Phi \in \bb R^{n \times m}$ satisfies the restricted isometry property
	\begin{align}
		\paren{1- \delta_s} \sqrt{n}\norm{\mb v}{} \leq \norm{ \mb \Phi \mb v}{} \leq (1+\delta_s)\sqrt{n} \norm{ \mb v }{}
	\end{align}
	for all $\mb v \in \mc E_s$, with probability at least $1-m^{ - \log^2s \log m} $.
\end{theorem}

\begin{lemma}\label{lem:rip-large-perturb}
	Let the random vector $\mb \xi \in \bb C^m$ and the random matrix $\mb \Phi \in \bb C^{n\times m}$ be defined the same as \Cref{thm:circ-rip}, and let $\mc E_s= \Brac{ \mb v \in \bb C^m \mid \norm{ \mb v}{0} \leq s  }$ for some positive integer $s\leq n$. For any positive scalar $\delta>0$ and any positive integer $s \leq n$, whenever $m \geq C \delta^{-2} n \log^4 n$, we have
	\begin{align*}
	   \norm{\mb \Phi \mb v}{} \leq  \delta \sqrt{m} \norm{\mb v}{},
	\end{align*}
    for all $\mb v \in \mc E_s$, with probability at least $1 - m^{- c\log^2 s} $. 
\end{lemma}

\begin{proof}
	The proof follows from the results in \cite{krahmer2014suprema}. Without loss of generality, we assume $\norm{\mb v}{} = 1$. Let us define sets
	\begin{align*}
	   \mc D_{s,m}\;&\doteq\; \Brac{ \mb v \in \bb C^m: \; \norm{\mb v}{}= 1,\; \norm{\mb v}{0} \leq s }, \\
	   \mc V \;&\doteq\; \Brac{ \frac{1}{\sqrt{n}} \mb R_{[1:n]} \mb F_m^{-1} \diag\paren{\mb F_m \mb v } \mb F_m \mid \mb v \in \mc D_{s,m} },
	\end{align*}
	\edited{where $\mb R_{[1:n]}: \bb R^m \mapsto \bb R^n $ denotes an operator that restricts a vector to its first $n$ coordinates.} Section 4 of \cite{krahmer2014suprema} shows that 
    \begin{align*}
       \sup_{\mb v \in \mc D_{s,m}} \abs{ \frac{1}{n} \norm{\mb \Phi\mb v }{}^2 - 1  } = \sup_{\mb V_{\mb v}\in \mc V} \abs{  \norm{ \mb V_{\mb v} \mb \xi }{}^2 -  \bb E_{\mb \xi }\brac{ \norm{ \mb V_{\mb v} \mb \xi }{}^2 } }  = \mc C_{\mc V}(\mb \xi ),
    \end{align*}
    where $\mc C_{\mc V}\paren{ \mb \xi }$ is defined in \eqref{eqn:quantity-C-B}. Theorem 4.1 and Lemma 4.2 of \cite{krahmer2014suprema} implies that
    \begin{align*}
       d_F(\mc V) = 1,\quad d_2(\mc V) \leq \sqrt{ \frac{s}{n} },\quad \gamma_2(\mc V, \norm{ \cdot }{}) \leq c \sqrt{ \frac{s}{n} }\log s \log m,	
    \end{align*}
    for some constant $c>0$. By using the estimates above, Theorem 3.1 of \cite{krahmer2014suprema} further implies for any $t>0$
    \begin{align*}
       \bb P\paren{  \mc C_{\mc V}(\mb \xi ) \geq c_1 \sqrt{\frac{s}{n}} \log^2 s \log^2 m +t } \leq 2 \exp\paren{  - c_2 \min\Brac{  \frac{n t^2}{s\log^2 s \log^2 m }, \frac{n t}{s} }  }.
    \end{align*}
    For any positive constant $\delta>0$, choosing $t = \delta^2 m/n$, whenever $m\geq C \delta^{-2} n \log^2s\log^2 m$ for some constant $C >0$ large enough, we have
    \begin{align*}
       	\sup_{\mb v \in \mc D_{s,m}} \abs{ \frac{1}{n} \norm{\mb \Phi\mb v }{}^2 - 1  } \leq c_1 \sqrt{\frac{s}{n}} \log^2 s \log^2 m + \delta \frac{m}{n} \leq 2 \delta^2 \frac{m}{n},
    \end{align*}
    with probability at least $ 1- m^{ -c_3\log^2 s }$. Therefore, we have
    \begin{align*}
       	 \norm{\mb \Phi\mb v }{} \leq \sqrt{ n + 2\delta^2 m } \leq C'\delta \sqrt{ m},
    \end{align*}
    holds for any $\mb v \in \mc D_{s,m}$ with high probability.
\end{proof}

\section{Moments and Spectral Norm of Partial Random Circulant Matrix}\label{app:moments-circulant-matrix}

Let $\mb g \in \bb C^m$ be a random complex Gaussian vector with $\mb g \sim \mc {CN}(\mb 0, \sigma_g^2 \mb I)$. Given a partial random circulant matrix $\mb C_{\mb g}\mb R_{[1:n]}^\top \in \bb C^{m \times n} \;(m\geq n)$, we control the moments and the tail bound of the terms in the following form
\begin{align*}
\mb T_1(\mb g) \;&=\; \frac{1}{m}  \mb R_{[1:n]} \mb C_{\mb g}^* \diag\paren{\mb b } \mb C_{\mb g} \mb R_{[1:n]}^\top,  \\
\mb T_2(\mb g) \;&=\; \frac{1}{m}  \mb R_{[1:n]} \mb C_{\mb g}^\top \diag\paren{\wt{\mb b} } \mb C_{\mb g} \mb R_{[1:n]}^\top, 
\end{align*}
where $\mb b  \in \bb R^m $, and $\wt{\mb b}  \in \bb C^m $. The concentration of these quantities plays an important role in our arguments, and the proof mimics the arguments in \cite{rauhut2010compressive,krahmer2014suprema}. Prior to that, let us define the sets 
\begin{align}
	\mc D \;&\doteq\; \Brac{ \mb v \in \bb {CS}^{m-1}:\; \supp (\mb v) \in [n]  }, \label{eqn:set-D}  \\
	\mc V(\mb d) \;&\doteq\; \Brac{\mb V_{\mb v}:\; \mb V_{\mb v} = \frac{1}{\sqrt{m}} \diag\paren{\mb d}^{1/2} \mb F_m^{-1} \diag\paren{\mb F_m\mb v} \mb F_m,\; \mb v \in \mc D  }, \label{eqn:set-V}
\end{align}
\edited{for some $\mb d \in \bb C^m$.}

\subsection{Controlling the Moments and Tail of \boldmath{$T_1(g)$} }
\begin{theorem}\label{thm:moments-norm-circ}
	Let $\mb g \in \bb C^m$ be a random complex Gaussian vector with $\mb g \sim \mc {CN}(\mb 0, \sigma_g^2 \mb I)$ and any fixed vector $\mb b = \brac{b_1,\cdots,b_m }^\top \in \bb R^m $. Given a partial random circulant matrix $\mb C_{\mb g}\mb R_{[1:n]}^\top \in \bb C^{m \times n} \;(m\geq n)$, let us define
	\begin{align*}
		\mc L(\mb g) &\doteq \norm{ \frac{1}{m}  \mb R_{[1:n]} \mb C_{\mb g}^* \diag\paren{\mb b } \mb C_{\mb g} \mb R_{[1:n]}^\top - \frac{1}{m}\paren{ \sum_{k=1}^m b_k} \mb I}{}.
	\end{align*}
Then for any integer $p\geq 1$, we have
	\begin{align*}
	   \norm{ \mc L(\mb g) }{L^p} \leq C_{\sigma_g^2} \norm{\mb b }{ \infty }  \paren{  \sqrt{\frac{n}{m}} \log^{3/2} n\log^{1/2}m + \sqrt{p} \sqrt{\frac{n}{m}}  + p \frac{n}{m}   }.
	\end{align*}
    In addition, \edited{for} any $\delta>0$, whenever $m \geq C_{\sigma_g^2}' \delta^{-2} \norm{\mb b}{ \infty }^2 n\log^4n$, we have
	\begin{align}\label{eqn:quadratic-spectral-bound}
		\mc L(\mb g) &\leq \delta
	\end{align}
	holds with probability at least $1 - 2 m^{- c_{\sigma_g^2} \log^3n}$. Here, $c_{\sigma_g^2},\; C_{\sigma_g^2}$, and $C_{\sigma_g^2}'$ are some numerical constants only depending on $\sigma_g^2$.
\end{theorem}

\begin{proof}
	Without loss of generality, let us assume that $\sigma_g^2 =1 $. Let us first consider the case $\mb b \geq \mb 0$, and let $\mb \Lambda = \diag \paren{\mb b}$, then
	\begin{align*}
	   \mc L(\mb g) &= \sup_{ \mb w \in \bb {CS}^{n-1} }	\abs{  \frac{1}{m} \mb w^* \mb R_{[1:n]} \mb C_{\mb g}^* \mb \Lambda \mb C_{\mb g} \mb R_{[1:n]}^\top \mb w - \frac{1}{m} \sum_{k=1}^m b_k  } \\
	   & = \sup_{ \mb v \in \bb {CS}^{m-1}, \supp(\mb v) \in [n]  } \abs{  \frac{1}{m} \mb v^* \mb C_{\mb g}^* \mb \Lambda \mb C_{\mb g}\mb v - \frac{1}{m} \sum_{k=1}^m b_k }.
	\end{align*}
	By the convolution theorem, we know that
	\begin{align*}
	 \frac{1}{\sqrt{m} } \mb \Lambda^{1/2} \mb C_{\mb g}\mb v = \frac{1}{\sqrt{m}} \mb \Lambda^{1/2} \paren{\mb g \conv \mb v } = \frac{1}{\sqrt{m}} \mb \Lambda^{1/2} \mb F_m^{-1} \diag\paren{\mb F_m\mb v} \mb F_m \mb g = \mb V_{\mb v} \mb g.
	\end{align*}
	Since $\bb E\brac{ \mb R_{[1:n]} \mb C_{\mb g}^* \mb \Lambda \mb C_{\mb g} \mb R_{[1:n]}^\top } =  \paren{\sum_{k=1}^m b_k} \mb I $, we observe
	\begin{align*}
		\mc L(\mb g) & = \sup_{\mb V_{\mb v} \in \mc V(\mb b) } \abs{ \norm{\mb V_{\mb v}\mb g }{}^2 - \bb E\brac{ \norm{\mb V_{\mb v}\mb g }{}^2 }  },
	\end{align*}
	where the set $\mc V(\mb b)$ is defined in \eqref{eqn:set-V}. Next, we invoke \Cref{thm:moments-bound} to control all the moments of $\mc L(\mb a)$, where we need to control the quantities $d_2(\cdot )$, $d_F(\cdot )$ and $\gamma_2(\cdot, \norm{ \cdot}{})$ defined in Definition \ref{def:gamma-beta-functional} for the set $\mc V(\mb b)$. By Lemma \ref{lem:d_2-d_F} and Lemma \ref{lem:gamma-2}, we know that
    \begin{align}
        d_F(\mc V(\mb b)) &\leq \norm{\mb b }{\infty}^{1/2}, \quad d_2(\mc V(\mb b)) \leq \sqrt{\frac{n}{m}} \norm{ \mb b }{ \infty }^{1/2}, \label{eqn:d_F-d_2} \\
    	\gamma_2(\mc V(\mb b),\norm{ \cdot }{}) &\leq C_0 \sqrt{ \frac{n}{m}  } \norm{\mb b }{\infty }^{1/2}  \log^{3/2} n \log^{1/2} m, \label{eqn:gamma-2}
    \end{align}
    for some constant $C_0>0$. Thus, combining the results in \eqref{eqn:d_F-d_2} and \eqref{eqn:gamma-2}, whenever $m \geq C_1n \log^3 n \log m $ for some constant $C_1>0$, \Cref{thm:moments-bound} implies that
    \begin{align*}
       &\norm{ \mc L(\mb g) }{L^p} \\
       \;\leq\;& C_2\Brac{\gamma_2(\mc V(\mb b), \norm{\cdot}{})\brac{ \gamma_2\paren{\mc V(\mb b), \norm{ \cdot }{} } + d_F(\mc V(\mb b))  }  
       + \sqrt{p} d_2(\mc V(\mb b))\brac{  \gamma_2\paren{ \mc V,\norm{\cdot}{} } + d_F(\mc V(\mb b)) } + p d_2^2(\mc V(\mb b))} \\
       \;\leq\;& C_3 \norm{\mb b}{\infty} \paren{  \sqrt{\frac{n}{m}} \log^{3/2} n\log^{1/2}m +  \sqrt{\frac{n}{m}} \sqrt{p}  +  \frac{n}{m}  p  }
    \end{align*}
    holds for some constants $C_2,\;C_3>0$. Based on the moments estimate of $\mc L(\mb g)$, Lemma \ref{lem:moments-tail-bound} further implies that
    \begin{align*}
    	\bb P\paren{ \mc L(\mb g) \geq C_4\sqrt{\frac{n}{m}} \norm{\mb b }{ \infty} \log^{3/2} n\log^{1/2}m  +t } \;\leq\; 2\exp\paren{ - C_5\frac{m}{ n  }\norm{\mb b}{ \infty}^{-1} \min\Brac{ \frac{t^2}{ \norm{\mb b}{\infty} },t} },
    \end{align*}
    for some constants $C_4,\;C_5>0$. Thus, for any $\delta>0$, whenever $ m \geq C_6 \delta^{-2} \norm{\mb b}{\infty}^2 n \log^3 n \log m  $ for some constant $C_6>0$, we have
    \begin{align*}
       \mc L(\mb g)\; \leq\; \delta
    \end{align*}
    holds with probability at least $ 1 - 2 m^{ -C_7 \log^3 n } $. 

Now when $\mb b$ is not nonnegative, let $\mb b = \mb b_+ - \mb b_-$, where $\mb b_+ = \begin{bmatrix} b_1^+,\cdots,b_m^+	\end{bmatrix}^\top , \; \mb b_- = \begin{bmatrix}b_1^-,\cdots,b_m^-\end{bmatrix}^\top\in \bb R_+^m$ are the nonnegative and nonpositive part of $\mb b$, respectively. Let $\mb \Lambda = \diag\paren{\mb b}$, $\mb \Lambda_+ = \diag\paren{\mb b_+}$ and $\mb \Lambda_- = \diag\paren{\mb b_-}$, we have
	\begin{align*}
		\mc L(\mb g) & = \sup_{ \mb w \in \bb {CS}^{n-1}  } \abs{\frac{1}{m} \mb w^* \mb R_{[1:n]} \mb C_{\mb g}^*\mb \Lambda \mb C_{\mb g} \mb R_{[1:n]}^\top \mb w - \frac{1}{m}\sum_{k=1}^m b_k  } \\
		&  \leq \underbrace{ \sup_{ \mb w \in \bb {CS}^{n-1}  } \abs{\frac{1}{m} \mb w^* \mb R_{[1:n]} \mb C_{\mb g}^*\mb \Lambda_+ \mb C_{\mb g} \mb R_{[1:n]}^\top \mb w - \frac{1}{m}\sum_{k=1}^m b_{k}^+  } }_{\mc L_+(\mb g)} 
		+ \underbrace{ \sup_{ \mb w \in \bb {CS}^{n-1}  } \abs{\frac{1}{m} \mb w^* \mb R_{[1:n]} \mb C_{\mb g}^*\mb \Lambda_- \mb C_{\mb g} \mb R_{[1:n]}^\top \mb w - \frac{1}{m}\sum_{k=1}^m b_{k}^-  } }_{\mc L_-(\mb g)}.
	\end{align*}
	Now since $\mb b_+,\mb b_- \in \bb R_+^m$, we can apply the results above for $\mc L_+(\mb g)$ and $\mc L_-(\mb g)$, respectively. Then by Minkowski's inequality, we have
	\begin{align*}
	  \norm{\mc L(\mb g)}{L^p} \leq  \norm{ \mc L_+(\mb g) }{L^p} +  \norm{ \mc L_-(\mb g_) }{L^p} \leq C_6 \norm{\mb b}{\infty} \paren{  \sqrt{\frac{n}{m}} \log^{3/2} n\log^{1/2}m +  \sqrt{\frac{n}{m}} \sqrt{p}  +  \frac{n}{m}  p  }
	\end{align*}
    for some constant $C_6>0$. The tail bound can be similarly derived from the moments bound. This completes the proof.
\end{proof}

The result above also implies the following result.
\begin{corollary}\label{cor:circ-spectral-norm}
   Let $\mb g \in \bb C^m$ be a random complex Gaussian vector with $\mb g \sim \mc {CN}(\mb 0, \sigma_g^2 \mb I)$, and let $\mb G = \mb R_{[1:n]} \mb C_{\mb g}^* \in \bb C^{n \times m}\;(n\leq m)$. Then for any integer $p\geq 1$, we have
   \begin{align*}
     \paren{ \bb E\brac{ \norm{\mb G}{}^p  }	 }^{1/p} \;\leq\; C_{\sigma_g^2} \sqrt{m} \paren{ 1+ \sqrt{ \frac{n}{m}  } \log^{3/2} n \log^{1/2} m  + \sqrt{ \frac{n}{m} } \sqrt{p}   }
   \end{align*}
   Moreover, for any $\epsilon\in (0,1)$, whenever $m \geq C\delta^{-2} n \log^4 n$ for some constant $C>0$, we have
   \begin{align*}
      (1-\delta) m \norm{\mb w}{}^2 \;\leq\; \norm{ \mb G^*\mb w}{}^2	\;\leq\; (1+\delta) m \norm{\mb w}{}^2
   \end{align*}
   holds for $\mb w \in \bb C^n$ with probability at least $1 - 2 m^{- c_{\sigma_g^2} \log^3 n }$. Here $c_{\sigma_g^2}, C_{\sigma_g^2}>0$ are some constants depending only on $\sigma_g^2$.

\end{corollary}

\begin{proof}
Firstly, notice that
\begin{align*}
  \norm{\mb G}{} = \sup_{ \mb w \in \bb {CS}^{n-1}, \mb r \in \bb {CS}^{m-1} } \abs{ \innerprod{ \mb w }{ \mb G \mb r } } \leq \sup_{ \mb w \in \bb {CS}^{n-1} } \norm{\mb G^* \mb w}{}& = \sup_{ \mb w \in \bb {CS}^{n-1} } \norm{ \mb C_{\mb g} \mb R_{[1:n]}^\top  \mb w}{} \\
  &= \sup_{ \mb v \in \bb {CS}^{m-1}, \supp(\mb v) \in [n]  } \norm{ \mb C_{\mb g} \mb v }{}.
\end{align*}
Thus, similar to the argument of \Cref{thm:moments-norm-circ}, let the set $\mc D$ and $\mc V(\mb 1)$ define as \eqref{eqn:set-D} and \eqref{eqn:set-V}, we have
\begin{align*}
      \frac{1}{\sqrt{m}} \norm{\mb G}{} \leq \sup_{\mb V_{\mb v} \in \mc V(\mb 1)} \norm{ \mb V_{\mb v} \mb g }{}.
\end{align*}
By Lemma \ref{lem:d_2-d_F} and Lemma \ref{lem:gamma-2}, we know that
    \begin{align*}
        d_F(\mc V(\mb 1)) \;\leq\; 1, \quad d_2(\mc V(\mb 1)) \;\leq\; \sqrt{\frac{n}{m}},\quad 
    	\gamma_2(\mc V(\mb 1),\norm{ \cdot }{}) \;\leq\; C_0 \sqrt{ \frac{n}{m}  } \log^{3/2} n \log^{1/2} m.
    \end{align*}
Thus, using \Cref{thm:moments-bound}, we obtain
    \begin{align*}
       \bb E\brac{ \abs{ \sup_{\mb V_{\mb v} \in \mc V(\mb 1) } \norm{ \mb V_{\mb v} \mb g }{} }^p	 }^{1/p} \;\leq\; C_{\sigma_g^2} \paren{  \sqrt{ \frac{n}{m}  } \log^{3/2} n \log^{1/2} m  +1 + \sqrt{ \frac{n}{m} } \sqrt{p}   },
    \end{align*}
where $C_{\sigma_g^2}>0$ is constant depending only on $\sigma_g^2$. The concentration inequality can be directly derived from Theorem \ref{thm:moments-norm-circ}, noticing that for any $\delta>0$, whenever $m \geq C_1\delta^{-2} n \log^4 n$ for some positive constant $C_1>0$, we have
\begin{align*}
   \sup_{ \mb w \in \bb {CS}^{n-1} }\; \abs{ \frac{1}{m} \paren{\mb w^* \mb G \mb G^*\mb w -1} } \;\leq\; \delta \quad \Longrightarrow \quad (1 - \delta) m \;\leq\; \sup_{ \mb w \in \bb {CS}^{n-1} } \norm{\mb G^*\mb w}{}^2 \;\leq\; (1+\delta) m 
\end{align*}
holds with probability at least $1 - 2 m^{- c_{\sigma_g^2} \log^3 n }$, where $c_{\sigma_g^2}>0$ is some constant depending only on $\sigma_g^2$.
\end{proof}

\subsection{Controlling the Moments of \boldmath{$T_2(g)$} }
\begin{theorem}\label{thm:moments-T-2}
Let $\mb g\in \bb C^m$ are a complex random Gaussian variable with $\mb g \sim \mc {CN}(\mb 0, \sigma_g^2 \mb I)$, and let
\begin{align*}
   \mc N(\mb g) \doteq \sup_{ \mb w \in \bb {CS}^{n-1} } \abs{ \frac{1}{m} \mb w^\top   \mb R_{[1:n]} \mb C_{\mb g}^\top  \diag\paren{\wt{ \mb b} } \mb C_{\mb g} \mb R_{[1:n]}^\top \mb w},  
\end{align*}
where $\wt{ \mb b} \in \bb C^m$.
Then whenever $m \geq Cn \log^4 n$ for some positive constant $C>0$, for any positive integer $p\geq 1$, we have
\begin{align*}
  	\norm{\mc N(\mb g) }{L^p}  \leq  C_{\sigma_g^2} \norm{ \wt{ \mb b} }{ \infty }\paren{  \sqrt{\frac{n}{m}}  \log^{3/2} n\log^{1/2}m +  \sqrt{\frac{n}{m}} \sqrt{p} +  \frac{n}{m} p  },
\end{align*}
where $C_{\sigma_g^2}$ is positive constant only depending on $\sigma_g^2$.
\end{theorem}

\begin{proof}
Let $\wt{\mb \Lambda} =\diag\paren{\wt{\mb b}} $, similar to the arguments of \Cref{thm:moments-norm-circ}, we have
\begin{align}\label{eqn:N-g}
   \mc N(\mb g) = \sup_{ \mb V_{\mb v} \in \mc V\paren{\wt{\mb b}} } 	\abs{ \innerprod{ \mb V_{\mb v}\mb g }{ \ol{\mb V_{\mb v}\mb g} } },
\end{align}
where $\mc V\paren{\wt{\mb b}}$ is defined as \eqref{eqn:set-V}. Let $\mb g'$ be an independent copy of $\mb g$, by Lemma \ref{lem:decouple-V}, for any integer $p\geq 1$ we have
\begin{align*}
   	\norm{ \mc N(\mb g)}{L^p} \leq 4 \norm{ \sup_{\mb V_{\mb v} \in \mc V\paren{\wt{\mb b}}} \abs{ \innerprod{ \mb V_{\mb v} \mb g }{ \ol{\mb V_{\mb v} \mb g'} } }  }{L^p}. 
\end{align*}
\edited{For convenience, let $\mc V =\mc V\paren{\wt{\mb b}} $.} By Lemma \ref{lem:V-1} and Lemma \ref{lem:V-2}, we know that
\begin{align*}
&\norm{ \sup_{  \mb V_{\mb v} \in \mc V } \abs{\innerprod{\mb V_{\mb v}\mb g }{ \ol{\mb V_{\mb v}\mb g'} } } }{L^p} \\
\leq\;& C_{\sigma_g^2} \brac{ \gamma_2\paren{ \mc V, \norm{\cdot }{} } \norm{ \sup_{ \mb V_{\mb v} \in \mc V } \norm{\mb V_{\mb v} \mb g' }{} }{L^p} + \sup_{\mb V_{\mb v} \in \mc V} \norm{ \innerprod{\mb V_{\mb v} \mb g }{ \ol{\mb V_{\mb v} \mb g'  } } }{L^p}} \\
\leq\;& C_{\sigma_g^2}'\brac{ \gamma_2\paren{\mc V, \norm{ \cdot }{} }\paren{ \gamma_2(\mc V, \norm{\cdot}{} ) + d_F(\mc V) } + \sqrt{p} d_2(\mc V) \paren{d_F(\mc V) + \gamma_2(\mc V, \norm{\cdot}{} ) } + p d_2^2(\mc V) }.
\end{align*}
By Lemma \ref{lem:d_2-d_F} and Lemma \ref{lem:gamma-2}, we know that
\begin{align*}
   d_F(\mc V) \leq \norm{ \wt{\mb b} }{ \infty }^{1/2},\; d_2(\mc V) \leq \sqrt{ \frac{n}{m} } \norm{ \wt{ \mb b} }{ \infty }^{1/2},\; \gamma_2(\mc V, \norm{ \cdot}{}) \leq C \sqrt{\frac{n}{m}} \norm{ \wt{ \mb b} }{\infty}^{1/2} \log^{3/2} n \log^{1/2} m,
\end{align*}
where $C>0$ is constant. Thus, combining the results above, we have
\begin{align*}
  	\norm{ \mc N(\mb g)}{L^p} \leq   C_{\sigma_g^2}'' \paren{  \sqrt{\frac{n}{m}} \norm{\wt{\mb b} }{ \infty} \log^{3/2} n\log^{1/2}m + \sqrt{p} \sqrt{\frac{n}{m}} \norm{ \wt{\mb b} }{\infty } + p \frac{n}{m} \norm{ \wt{ \mb b} }{ \infty }   },
\end{align*}
where $C_{\sigma_g^2}''>0$ is some constant depending on $\sigma_g^2$.
\end{proof}

\begin{lemma}\label{lem:decouple-V}
Let $\mc N(\mb g)$ be defined as \eqref{eqn:N-g}, and let $\mb g'$ be an independent copy of $\mb g$, then we have
\begin{align*}
   \norm{ \mc N(\mb g) }{L^p} \leq 4 \norm{ \sup_{\mb V_{\mb v} \in \mc V(\mb d)} \abs{ \innerprod{ \mb V_{\mb v} \mb g }{ \ol{\mb V_{\mb v} \mb g'} } }  }{L^p	},
\end{align*}
\edited{for some vector $\mb d \in \bb C^m $.}
\end{lemma}

\begin{proof}
Let $\mb \delta \sim \mc {CN}(\mb 0, \sigma_g^2 \mb I)$ which is independent of $\mb g$, and let 
\begin{align*}
   \mb g^1 = \mb g + \mb \delta,\qquad \mb g^2 = \mb g - \mb \delta,
\end{align*}
so that $\mb g^1$ and $\mb g^2$ are also independent with $\mb g^1, \mb g^2 \sim \mc {CN}(\mb 0, 2 \sigma_g^2 \mb I)$. Let $\mc Q_{dec}^{\mc N}(\mb g^1,\mb g^2) = \innerprod{\mb V_{\mb v} \mb g^1}{ \ol{\mb V_{\mb v} \mb g^2}  } $, then we have
\begin{align*}
   \bb E_{\mb \delta} \brac{ \mc Q_{dec}^{\mc N}(\mb g^1,\mb g^2)} = \innerprod{ \mb V_{\mb v} \mb g }{ \ol{  \mb V_{\mb v} \mb g }  }.
\end{align*}
Therefore, by Jensen's inequality, we have
\begin{align*}
   \norm{ \mc N(\mb g) }{L^p} = \paren{ \bb E_{\mb g}\brac{ \paren{\sup_{\mb V_{\mb v}\in \mc V(\mb d) } \abs{ \bb E_{\mb \delta} \brac{ \mc Q_{dec}^{\mc N}(\mb g^1,\mb g^2)} } }^p } }^{1/p} &\leq \paren{ \bb E_{\mb g^1,\mb g^2}\brac{ \paren{\sup_{\mb V_{\mb v}\in \mc V(\mb d)} \abs{  \mc Q_{dec}^{\mc N}(\mb g^1,\mb g^2)}  }^p } }^{1/p} \\
   &= 4 \norm{ \sup_{\mb V_{\mb v} \in \mc V(\mb d)} \abs{ \innerprod{ \mb V_{\mb v} \mb g }{ \ol{\mb V_{\mb v} \mb g'} } }  }{L^p},
\end{align*}
as desired.
\end{proof}

\begin{lemma}\label{lem:V-1}
Let $\mb g'$ be an independent copy of $\mb g$, for every integer $p\geq 1$, we have
\begin{align*}
 	\norm{ \sup_{ } \innerprod{\mb V_{\mb v}\mb g }{ \ol{ \mb V_{\mb v}\mb g' } } }{L^p} \leq C_{\sigma_g^2} \brac{ \gamma_2\paren{ \mc V(\mb d), \norm{\cdot}{} } \norm{ \sup_{ \mb V_{\mb v} \in \mc V(\mb d) } \norm{\mb V_{\mb v} \mb g' }{} }{L^p} + \sup_{\mb V_{\mb v} \in \mc V(\mb d)} \norm{ \innerprod{\mb V_{\mb v} \mb g }{ \ol{\mb V_{\mb v} \mb g'  } } }{L^p}},
\end{align*}
for some vector $\mb d \in \bb C^m $, where $C_{\sigma_g^2}>0$ is a constant depending only on $\sigma_g^2$.
\end{lemma}
\begin{proof}
The proof is similar to the proof of Lemma 3.2 of \cite{krahmer2014suprema}, and \edited{it is} omitted here.
\end{proof}

\begin{lemma}\label{lem:V-2}
Let $\mb g'$ be an independent copy of $\mb g$, for every integer $p\geq 1$, we have
\begin{align*}
   \norm{ \sup_{ \mb V_{\mb v} \in \mc V(\mb d) } \norm{\mb V_{\mb v} \mb g' }{} }{L^p} &\leq  C_{\sigma_g^2}\brac{ \gamma_2(\mc V(\mb d), \norm{\cdot}{}) + d_F(\mc V) + \sqrt{p} d_2(\mc V(\mb d)) } \\
\sup_{\mb V_{\mb v} \in \mc V(\mb d)} \norm{ \innerprod{ \mb V_{\mb v}\mb g}{ \ol{ \mb V_{\mb v} \mb g' } } }{L^p} &\leq C_{\sigma_g^2}\brac{  \sqrt{p} d_F(\mc V(\mb d)) \cdot d_2(\mc V(\mb d)) + p d_2^2(\mc V(\mb d)) },
\end{align*}
for some vector $\mb d \in \bb C^m $, where $C_{\sigma_g^2}>0$ is a constant depending only on $\sigma_g^2$.
\end{lemma}

\begin{proof}
The proof is similar to the proofs of Theorem 3.5 and Lemma 3.6 of \cite{krahmer2014suprema}, and \edited{it is} omitted here.
\end{proof}

\subsection{Auxiliary Results}
The following are the auxiliary results required in the main proof.
\begin{lemma}\label{lem:d_2-d_F}
Let the sets $\mc D,\; \mc V(\mb d)$ be defined as \eqref{eqn:set-D} and \eqref{eqn:set-V} for some $\mb d \in \bb C^m$, we have
	\begin{align*}
	   d_F(\mc V(\mb d)) \;\leq\; \norm{\mb d }{\infty}^{1/2},\qquad d_2(\mc V(\mb d)) \;\leq\; \sqrt{\frac{n}{m}} \norm{ \mb d }{ \infty }^{1/2}.
	\end{align*}
\end{lemma}

\begin{proof}
Since each row of $\mb V_{\mb v} \in \mc V(\mb d)$ consists of weighted shifted copies of $\mb v$, the $\ell_2$-norm of each nonzero row of $\mb V_{\mb v}$ is $m^{-1/2} \abs{d_k}^{1/2} \norm{\mb v}{}$. Thus, we have 
	\begin{align*}
		d_F(\mc V(\mb d)) \;=\; \sup_{\mb V_{\mb v} \in \mc V(\mb d) } \norm{\mb V_{\mb v}}{F} \leq  \norm{ \mb d }{\infty }^{1/2} \sup_{\mb v \in \mc D} \norm{\mb v}{} = \norm{\mb d}{ \infty }^{1/2}.
	\end{align*}
    Also, for every $\mb v \in \mc D$, we observe
	\begin{align*}
		\norm{\mb V_{\mb v} }{}  \leq  \frac{1}{ \sqrt{m} } \norm{ \mb d }{\infty }^{1/2} \norm{\diag\paren{ \mb F_m \mb v}  }{} = \frac{1}{\sqrt{m} }  \norm{ \mb d }{\infty }^{1/2} \norm{\mb F_m\mb v }{\infty } .
	\end{align*}
	It is obvious for any $\mb v \in \mc D $ that $\norm{\mb F_m\mb v}{\infty} \leq \norm{\mb v}{1} \leq \sqrt{n} \norm{\mb v}{2} = \sqrt{n} $, so that 
	\begin{align*}
			d_2(\mc V(\mb d)) \;=\;  \sup_{\mb v \in \mc D} \norm{\mb V_{\mb v} }{} \leq  \sqrt{ \frac{n}{m}} \norm{\mb d}{ \infty }^{1/2}.
	\end{align*}	
\end{proof}

\begin{lemma}\label{lem:gamma-2}
	Let the sets $\mc D,\; \mc V$ be defined as \eqref{eqn:set-D} and \eqref{eqn:set-V} for some $\mb d\in \bb C^m$, we have
	\begin{align*}
		\gamma_2(\mc V(\mb d), \norm{ \cdot }{}) \;\leq\; C \sqrt{ \frac{n}{m} } \norm{\mb d}{ \infty }^{1/2}  \log^{3/2}n \log^{1/2}m,
	\end{align*}
	\edited{where $\gamma_2(\cdot)$ is defined in Definition \ref{def:gamma-beta-functional}}.
\end{lemma}

\begin{proof}
	By Definition \ref{def:gamma-beta-functional}, we know that
	\begin{align*}
		\gamma_2\paren{\mc V(\mb d), \norm{\cdot }{} } \leq C \int_0^{d_2(\mc V) } \log^{1/2}\mc N\paren{ \mc V(\mb d), \norm{\cdot }{}, \epsilon } d \epsilon,
	\end{align*}
	for some constant $C>0$, where the right hand side is known as the ``Dudley integral''. To estimate the covering number $\mc N\paren{ \mc V(\mb d), \norm{\cdot }{}, \epsilon }$, we know that for any $\mb v,\;\mb v' \in \mc D $,
	\begin{align}
		\norm{\mb V_{\mb v} - \mb V_{\mb v'} }{} = \norm{ \mb V_{\mb v - \mb v'} }{} \leq \frac{1}{ \sqrt{m} } \norm{ \diag\paren{\mb d}^{1/2}}{} \norm{ \mb F_m(\mb v - \mb v') }{\infty } \leq \frac{1}{ \sqrt{m} } \norm{\mb d }{\infty }^{1/2} \norm{ \mb F_m(\mb v - \mb v') }{\infty }.
	\end{align}
    Let $ \norm{\mb v}{\wh{\infty} } \doteq \norm{\mb F_m \mb v}{\infty}$ that $ \norm{\mb v}{\wh{\infty} } \leq \norm{ \mb v }{1} $, we have $\mc N\paren{ \mc V(\mb d), \norm{\cdot }{}, \epsilon} \leq \mc N\paren{ \mc D , m^{-1/2} \norm{\mb d }{\infty }^{1/2} \norm{\cdot }{\wh{\infty} },\epsilon  } $. Next, we bound the covering number $\mc N\paren{ \mc D , m^{-1/2} \norm{\mb d }{\infty }^{1/2} \norm{\cdot }{\wh{\infty} },\epsilon  }$ when $\epsilon$ is small and large, respectively.

    When $\epsilon$ is small (i.e., $\epsilon \leq \mc O(1/\sqrt{m})$), let $\mc B_1^{[n]} = \Brac{ \mb v \in \bb C^m:\; \norm{\mb v}{1} \leq 1,\supp \mb v \in \brac{n} } $, then it is obvious that $\mc D \subseteq  \sqrt{n} \mc B_1^{[n]}  $. By Proposition 10.1 of \cite{rauhut2010compressive}, we have
    \begin{align*}
    	\mc N\paren{ \mc D , m^{-1/2}\norm{\mb d }{\infty }^{1/2} \norm{\cdot }{\wh{\infty} },\epsilon  } &\leq \mc N\paren{ \sqrt{n} \mc B_1^{[n]}, m^{-1/2} \norm{\mb d }{\infty }^{1/2}\norm{\cdot }{1}, \epsilon } \\
    	 &\leq \mc N(  \mc B_1^{[n]}, \norm{\cdot }{1},  \norm{\mb d }{\infty }^{-1/2}\sqrt{ \frac{m}{n} } \epsilon ) \leq \paren{ 1+ \frac{2\sqrt{n} \norm{\mb d }{\infty }^{1/2} }{ \sqrt{m} \epsilon} }^n.
    \end{align*}
    Thus, we have
    \begin{align*}
    	\log \mc N\paren{ \mc D , m^{-1/2} \norm{\mb d}{\infty }^{1/2} \norm{\cdot }{\wh{\infty} },\epsilon  } \leq n \log \paren{  1+ \frac{2\sqrt{n}  \norm{\mb d }{\infty }^{1/2} }{ \sqrt{m} \epsilon} }.
    \end{align*}
   If the scalar $\epsilon$ is large, let us introduce a norm
    \begin{align}\label{eqn:1-norm-star}
    	\norm{\mb v}{1}^* = \sum_{k=1}^m \abs{ \Re(v_k)} + \abs{ \Im(v_k)},\quad \forall\; \mb v \in \bb C^m,
    \end{align}
    which is the usual $\ell_1$-norm after identification of $\bb C^m$ with $\bb R^{2m}$. Let $\mc B_{ \norm{\cdot }{1}^* }^{[n]} = \Brac{ \mb v \in \bb C^m:\; \norm{\mb v}{1}^* \leq 1,\supp (\mb v) \in \brac{n} } $, then we have $\mc D \subseteq  \sqrt{2n} \mc B_{ \norm{\cdot }{1}^* }^{[n]}$. By Lemma \ref{lem:gamma-2-aux}, we obtain
    \begin{align*}
    	\log \mc N\paren{ \mc D , m^{-1/2} \norm{\mb d}{\infty }^{1/2} \norm{\cdot }{\wh{\infty} },\epsilon } \leq \log \mc N\paren{ \mc B_{ \norm{\cdot }{1}^* }^{[n]},  \norm{\cdot }{\wh{\infty} }, \frac{\sqrt{m}}{ \sqrt{2n} } \norm{\mb d }{\infty }^{-1/2} \epsilon } \leq \frac{Cn}{m \epsilon^2 } \norm{\mb d }{\infty } \log m \log n.
    \end{align*}
    Finally, we combine the results above to estimate the ``Dudley integral'',
    \begin{align*}
    	\mc I & \doteq  \int_0^{d_2(\mc V(\mb d)) } \log^{1/2}\mc N\paren{ \mc V, \norm{\cdot }{}, \epsilon } d \epsilon \\
    	&\leq \sqrt{n} \int_0^{ \kappa }  \log^{1/2}\paren{ 1+2\sqrt{ \frac{n}{m} } \frac{ \norm{\mb d }{\infty }^{1/2} }{\epsilon} } d\epsilon + C\sqrt{\frac{n}{m} \norm{\mb d }{\infty } \log m \log n}  \int_\kappa^{ \sqrt{\frac{n}{m}}  \norm{\mb v }{\infty }^{1/2}  }  \epsilon^{-1} d \epsilon \\
    	& \leq  \frac{2n}{\sqrt{m} } \norm{\mb d}{\infty}^{1/2}  \int_0^{\frac{ \kappa}{2  } \norm{\mb d}{\infty}^{-1/2}\sqrt{ \frac{m}{n} }  } \log^{1/2}\paren{ 1+t^{-1} } dt  + C\sqrt{\frac{n}{m} \norm{\mb d}{\infty} \log m \log n}  \log \paren{ \sqrt{ \frac{n}{m}} \norm{\mb d}{\infty}^{1/2}  / \kappa } \\
    	& \leq  \kappa \sqrt{n} \sqrt{ \log\paren{e\paren{1+ \frac{2}{\kappa} \norm{\mb d}{\infty }^{1/2} \sqrt{ \frac{n}{m} } } } } + C\sqrt{\frac{n}{m} \norm{\mb d }{\infty } \log m \log n}  \log \paren{ \sqrt{ \frac{n}{m}} \norm{\mb d}{\infty}^{1/2}  / \kappa },
    \end{align*}
    where the last inequality we used Lemma 10.3 of \cite{rauhut2010compressive}. Choose $\kappa = \frac{ \norm{\mb d}{ \infty }^{1/2} }{ \sqrt{m} }$, we obtain the desired result.
\end{proof}

\begin{lemma}\label{lem:gamma-2-aux}
	Let $\mc B_{ \norm{\cdot }{1}^* }^{[n]} = \Brac{ \mb v \in \bb C^m:\; \norm{\mb v}{1}^* \leq 1,\;\supp(\mb v)\in \brac{n} } $, and $\norm{\cdot }{1}^*$ is defined in \eqref{eqn:1-norm-star}, we have
	\begin{align}
		\log \mc N \paren{\mc B_{ \norm{\cdot }{1}^* }^{[n]}, \norm{\cdot }{\wh{\infty} }, \epsilon } \leq \frac{C}{\epsilon^2} \log m \log n
	\end{align}
	for some constant $C>0$, where the norm $\norm{\mb v}{\wh{\infty} } = \norm{\mb F_m \mb v}{\infty} $.
\end{lemma}

\begin{proof}
    Let $\mc U = \Brac{ \pm  \mb e_1, \cdots, \pm  \mb e_n,\pm \im \mb e_1, \cdots, \pm \im \mb e_n }$, it is obvious that $\mc B_{ \norm{\cdot }{1}^* }^{[n]} \subseteq \mathrm{conv}( \mc U)$, where $\mathrm{conv}(\mc U)$ denotes the convex hull of $\mc U$.
	\edited{Fix any $\mb v\in \mc U$, the idea is to approximate $\mb v$ by a finite set of very sparse vectors.} We define a random vector
	\begin{align*}
	    \mb z = \begin{cases}
	    	\sign\paren{ \Re (v_j) } \mb e_j, & \text{with prob.}\;\abs{\Re\paren{w_j} },\; 1\leq j\leq n\\
	    	\sign\paren{ \Im (v_j) } \mb e_j, & \text{with prob.}\;\abs{\Im\paren{w_j} },\; 1\leq j\leq n\\
	    	\mb 0, & \text{with prob.}\;1 - \norm{\mb w}{1}^*.
	    \end{cases}
	\end{align*}
    Since $\norm{\mb v}{1}^* \leq 1$, this is a valid probability distribution with $\bb E\brac{\mb z} = \mb v $. Let $\mb z_1, \cdots,\mb z_L$ be independent copies of $\mb z$, where $L$ is a number to be determined later. We attempt to approximate $\mb v$ with a $L$-sparse vector
    \begin{align*}
    	\mb z_S = \frac{1}{L} \sum_{k=1}^L \mb z_k.
    \end{align*}
    \edited{By using a classical symmetrization argument (e.g., see Lemma 6.7 of \cite{rauhut2010compressive})}, we obtain
    \begin{align*}
    	\bb E \brac{\norm{\mb z_S - \mb v }{ \wh{\infty} } }= \bb E \brac{ \norm{ \frac{1}{L}\sum_{k=1}^L \paren{\mb z_k - \bb E \brac{ \mb z_k}  } }{ \wh{\infty} }  } &\leq \frac{2}{L} \bb E\brac{ \norm{\sum_{k=1}^L \varepsilon_k \mb z_k }{ \wh{\infty} } } \\
    	= \frac{2}{L} \bb E \brac{\max_{\ell \in [m]} \abs{ \sum_{k=1}^L \varepsilon_k \innerprod{\mb f_\ell }{\mb z_k} } }
    \end{align*}
    where $\mb \varepsilon =\brac{\varepsilon_1, \cdots, \varepsilon_L }^* $ is a Rademacher vector, independent of $\Brac{\mb z_k }_{k=1}^L$. Fix a realization of $\Brac{\mb z_k }_{k=1}^L$, by applying the Hoeffding's inequality to $\mb \epsilon$, we obtain
    \begin{align*}
    	\bb P_{\mb \varepsilon} \paren{  \abs{ \sum_{k=1}^L \varepsilon_k \innerprod{\mb f_\ell }{\mb z_k} } \geq \sqrt{L} t } \leq \bb P_{\mb \varepsilon} \paren{  \abs{ \sum_{k=1}^L \varepsilon_k \innerprod{\mb f_\ell }{\mb z_k} } \geq \norm{ \sum_{k=1}^L \innerprod{\mb f_\ell}{\mb z_k}  }{}  t } \leq 2\exp\paren{- t^2/2 }
    \end{align*}
    for all $t>0$ and $\ell \in [m]$. Thus, by combining the result above with Lemma 6.6 of \cite{rauhut2010compressive}, it implies that
    \begin{align*}
        	\bb E \brac{ \max_{\ell \in [m]} \ \abs{ \sum_{k=1}^L \varepsilon_k \innerprod{\mb f_\ell }{\mb z_k} } } \leq C \sqrt{L \log(8m) },
    \end{align*}
    with $C = \sqrt{2} + \paren{4\sqrt{2} \log 8}^{-1}<1.5 $. By Fubini's theorem, we obtain
    \begin{align}\label{eqn:z-Lip}
    	\bb E \brac{\norm{\mb z_S - \mb v }{ \wh{\infty} } } \leq \frac{2}{L} \bb E_{\mb z} \bb E_{\mb \varepsilon} \brac{ \max_{\ell \in [m]} \abs{ \sum_{k=1}^L \varepsilon_k \innerprod{\mb f_\ell }{\mb z_k} } } \leq \frac{3}{ \sqrt{L} } \sqrt{ \log(8m)  }.
    \end{align}
    This implies that there exists a vector $\mb z_S = \frac{1}{L} \sum_{k=1}^L \mb z_k$ where each $\mb z_k \in \mc U$ such that $\norm{\mb z_S - \mb v }{ \wh{\infty} } \leq \frac{3}{ \sqrt{L} } \sqrt{ \log(8m)  } $. Since each $\mb z_k$ can take $4n+1$ values, so that $\mb z_S$ can take at most $(4n+1)^L$ values. And for each $\mb v \in \mathrm{conv}(\mc U)$, according to \eqref{eqn:z-Lip}, we can therefore find a vector $\mb z_S$ such that $\norm{\mb v - \mb z_S }{\wh{\infty} } \leq \epsilon $ with the choice $  L \leq \lfloor \frac{9 }{\epsilon^2} \log (10m) \rfloor $. Thus, we have
    \begin{align*}
    			\log \mc N \paren{\mc B_{ \norm{\cdot }{1}^* }^{[n]}, \norm{\cdot }{\wh{\infty} }, \epsilon } \leq \log \mc N \paren{\mathrm{conv}(\mc U), \norm{\cdot }{\wh{\infty} }, \epsilon } \leq L \log (4n+1) \leq \frac{9}{\epsilon^2} \log (10m) \log(4n+1)
    \end{align*}
    as desired.
\end{proof}

\section{Concentration via Decoupling}\label{app:decoupling}

In this section, we assume that $\norm{\mb x}{}=1$, and we develop concentration inequalities for the following quantities
\begin{align}
	\mb Y(\mb g) \;&=\; \frac{1}{m} \mb R_{[1:n]} \mb C_{\mb g}^* \diag\paren{\abs{ \mb g \conv \mb x }^2 }  \mb C_{\mb g} \mb R_{[1:n]}^\top, \label{eqn:Y-def} \\
	\mb M(\mb g) \;&=\; \frac{2\sigma^2 +1 }{m} \mb R_{[1:n]} \mb C_{\mb g}^* \diag\paren{ \zeta_{\sigma^2} (  \mb g \conv \mb x  )  } \mb C_{\mb g} \mb R_{[1:n]}^\top, \label{eqn:M-def}
\end{align}
via the decoupling technique and moments control, where $\zeta_{\sigma^2}(\cdot)$ is defined in \eqref{eqn:xi-zeta} and $\sigma^2>1/2$. Suppose $\mb g\in \bb C^m $ is complex Gaussian random variable $\mb g \sim \mc {CN}(\mb 0,\mb I)$. Once all the moments are bounded, it is easy to turn the moment bounds into a tail bound via Lemma \ref{lem:moments-tail-bound} and Lemma \ref{lem:moments-tail-bound-1}. To bound the moments, we use the decoupling technique developed in \cite{arcones1993decoupling,de1999decoupling,krahmer2014suprema}. The basic idea is to decouple the terms above into terms like
\begin{align}
\mc Q_{dec}^{\mb Y}(\mb g^1, \mb g^2) \;&=\; \frac{1}{m} \mb R_{[1:n]}\mb C_{\mb g^1}^* \diag \paren{ \abs{\mb g^2 \conv \mb x}^2  } \mb C_{\mb g^1} \mb R_{[1:n]}^\top  , \label{eqn:decoupled-Y} 	\\
\mc Q_{dec}^{\mb M}(\mb g^1, \mb g^2) \;&=\; \frac{1+ 2\sigma^2 }{m} \mb R_{[1:n]} \mb C_{\mb g^1}^* \diag \paren{  \eta_{\sigma^2 }\paren{  \mb g^2 \conv \mb x  } }  \mb C_{\mb g^1} \mb R_{[1:n]}^\top,\label{eqn:decoupled-M}
\end{align}
where $\eta_{\sigma^2} (t) = 1-  2\pi \sigma^2 \xi_{\sigma^2 -\frac{1}{2} }(t)$, and $\mb g^1$ and $\mb g^2$ are two independent random variables with
\begin{align}\label{eqn:g1-g2}
\mb g^1 \;=\; \mb g+ \mb \delta, \qquad \mb g^2 \;=\; \mb g - \mb \delta,	
\end{align}
where $\mb \delta \sim \mc {CN}(\mb 0,\mb I)$ is an independent copy of $\mb g$. As we discussed in \Cref{sec:analysis}, it turns out that controlling the moments of the decoupled terms $\mc Q_{dec}^{\mb Y}(\mb g^1, \mb g^2)$ and $\mc Q_{dec}^{\mb M}(\mb g^1, \mb g^2) $ for convolutional random matrices is easier and sufficient for providing the tail bound of $\mb Y$ and $\mb M$. The detailed results and proofs are described in the following subsections.

\subsection{Concentration of \boldmath{$ Y(g)$} }
In this subsection, we show that
\begin{theorem}\label{thm:concentration-Y}
   	Let $\mb g \sim \mc {CN}(\mb 0,\mb I)$, and let $\mb Y(\mb g)$ be defined as \eqref{eqn:Y-def}. For any $\delta>0$, when $m \geq C\delta^{-2} n \log^7 n$, we have
   	\begin{align*}
   	   \norm{\mb Y(\mb g) - \mb x \mb x^* - \mb I }{} \leq \delta,	
   	\end{align*}
    holds with probability at least $1 -2m^{-c}$. 
\end{theorem}

\begin{proof}
Suppose $\mb g^1,\mb g^2$ are defined as \eqref{eqn:g1-g2}, and $\mc Q_{dec}^{\mb Y}(\mb g^1, \mb g^2)$ is defined as \eqref{eqn:decoupled-Y}. Let $ [\mb g^2 \conv \mb x]_k  = \paren{\mb g^2_k}^* \mb x$ and $\mb C_{\mb g^1} \mb R_{[1:n]}^\top = \begin{bmatrix}
	\paren{\mb g_1^1}^* \\ \cdots \\  \paren{\mb g_m^1}^*
\end{bmatrix}
$, then by Lemma \ref{lem:Y-expected}, we have
\begin{align*}
   \bb E_{\mb \delta}\brac{ \mc Q_{dec}^{\mb Y}(\mb g^1,\mb g^2) } \;=\;& \frac{1}{m} \sum_{k=1}^m \bb E_{\mb \delta}\brac{\abs{\paren{\mb g_k - \mb \delta_k }^* \mb x }^2 \paren{ \mb g_k + \mb \delta_k } \paren{ \mb g_k + \mb \delta_k }^* } \\
    \;=\;& \frac{1}{m} \sum_{k=1}^m \paren{ \abs{\mb g_k^*\mb x}^2 \mb g_k \mb g_k^* + \mb g_k \mb g_k^* + \abs{\mb g_k^*\mb x }^2 \mb I + \mb x\mb x^* +\mb I - \mb x \mb x^* \mb g_k \mb g_k^* - \mb g_k \mb g_k^* \mb x \mb x^* } \\
   \;=\;& 4\mb I + \mb Y(\mb g) - \bb E_{\mb g}\brac{\mb Y(\mb g)} - \frac{1}{m}\sum_{k=1}^m \paren{ \mb x\mb x^*\mb g_k\mb g_k^* + \mb g_k\mb g_k^*\mb x\mb x^* - 2\mb x \mb x^* } \\
   &+ \frac{1}{m} \sum_{k=1}^m\paren{  \mb g_k\mb g_k^* - \mb I } + \frac{1}{m} \sum_{k=1}^m\paren{ \abs{\mb g_k^*\mb x}^2-1 }\mb I.
\end{align*}
Thus, by Minkowski inequality and Jensen's inequality, for any positive integer $p \geq 1$, we have
\begin{align*}
  &\paren{ \bb E_{\mb g}\brac{ \norm{ \mb Y(\mb g) - \bb E\brac{\mb Y(\mb g)} }{}^p  }	 }^{1/p}\\
  \leq\;& 4\paren{\bb E_{\mb g} \brac{ \norm{ \frac{1}{m} \mb R_{[1:n]}\mb C_{\mb g}^* \mb C_{\mb g}\mb R_{[1:n]}^\top - \mb I }{}^p  } }^{1/p}  + \paren{ \bb E_{\mb g}\brac{ \norm{ \bb E_{\mb \delta}\brac{ \mc Q_{dec}^{\mb Y}(\mb g^1,\mb g^2) } - 4 \mb I }{}^p  }	 }^{1/p} \\
  \leq\;& 4 \underbrace{ \paren{\bb E_{\mb g} \brac{ \norm{ \frac{1}{m} \mb R_{[1:n]}\mb C_{\mb g}^* \mb C_{\mb g}\mb R_{[1:n]}^\top - \mb I }{}^p  } }^{1/p} }_{\mc T_1}  + \underbrace{ \paren{ \bb E_{\mb g^1,\mb g^2}\brac{ \norm{ \mc Q_{dec}^{\mb Y}(\mb g^1,\mb g^2)  - 4 \mb I }{}^p  }	 }^{1/p} }_{\mc T_2}.
\end{align*}
\edited{By using \Cref{thm:moments-norm-circ} with  $\mb b = \mb 1$}, we have
\begin{align*}
	\mc T_1 \;\leq \; C_1 \paren{ \sqrt{\frac{n}{m}} \log^{3/2} \log^{1/2}m +  \sqrt{\frac{n}{m} } \sqrt{p} + \frac{n}{m} p  } ,
\end{align*}
where $C_1>0$ is some numerical constant. For $\mc T_2$, we have
\begin{align*}
  \mc T_2\;=\; &\paren{\bb E_{\mb g^1,\mb g^2}\brac{ \norm{ \mc Q_{dec}^{\mb Y}(\mb g^1,\mb g^2)  - 4 \mb I }{}^p  } }^{1/p} \\
   \;\leq\;&  \paren{\bb E_{\mb g^1,\mb g^2}\brac{ \norm{ \mc Q_{dec}^{\mb Y}(\mb g^1,\mb g^2)  - 2  \frac{1}{m}\norm{\mb g^2 \conv \mb x}{}^2 \mb I  }{}^p  } }^{1/p} \\
   & + 2 \paren{\bb E_{\mb g^1} \brac{ \norm{ \frac{1}{m} \mb R_{[1:n]}\mb C_{\mb g^1}^* \mb C_{\mb g^1}\mb R_{[1:n]}^\top - 2\mb I }{}^p  } }^{1/p}    \\
   \;\leq\;& C_2  \paren{\bb E_{\mb g^2}\brac{ \norm{  \mb g^2 \conv \mb x  }{\infty }^{2p} }^{1/p} +2}\paren{ \sqrt{\frac{n}{m}} \log^{3/2}n \log^{1/2}m +  \sqrt{\frac{n}{m} } \sqrt{p} + \frac{n}{m} p  } \\
    \;\leq\;& C_3 \paren{  \sqrt{\frac{n}{m}} \paren{\log^{3/2} n \log^{3/2}m } p +  \sqrt{\frac{n}{m}  } \paren{\log m } p^{3/2} + \frac{n}{m} \paren{\log m} p^2  },
\end{align*}
\edited{where the first inequality follows from the triangle inequality, the second inequality follows from \Cref{thm:moments-norm-circ}, and the last inequality follows from Lemma \ref{lem:infty-moments}.} Thus, combining the estimates for $\mc T_1$ and $\mc T_2$ above, we have
\begin{align*}
 \paren{ \bb E_{\mb g}\brac{ \norm{ \mb Y(\mb g) - \bb E\brac{\mb Y(\mb g)} }{}^p  }	 }^{1/p} \leq C_4 \paren{  \sqrt{\frac{n}{m}} \paren{ \log^{3/2} n \log^{3/2}m } p +  \sqrt{\frac{n}{m}  } \paren{ \log m  } p^{3/2} + \frac{n}{m} \paren{ \log m } p^2  }.
\end{align*}
Therefore, by using Lemma \ref{lem:moments-tail-bound-1}, for any $\delta>0$, whenever $ m \geq C_5\delta^{-2}n \log^4 m \log^3 n  $
\begin{align*}
	\norm{\mb Y - \bb E\brac{\mb Y} }{} \;\leq\; \delta,
\end{align*}
with probability at least $1-2m^{-c}$, where $c>0$ is some numerical constant . Finally, using Lemma \ref{lem:Y-expected}, we get the desired result.
\end{proof}

\begin{lemma}\label{lem:Y-expected}
\edited{Suppose $\mb x\in \bb C^n$ with $\norm{\mb x}{}=1$.} Let $\mb g\sim \mc {CN}(\mb 0,\mb I)$, and let $\mb Y(\mb g)$ be defined as \eqref{eqn:Y-def}, then we have
\begin{align*}
   \bb E\brac{\mb Y(\mb g)} \;=\; \mb x\mb x^* + \mb I.	
\end{align*}
\end{lemma}

\begin{proof}
\edited{Let $\mb C_{\mb g} \mb R_{[1:n]}^\top =\begin{bmatrix}
   \mb g_1^*  \\ \vdots \\ \mb g_m^* 	
\end{bmatrix}$. Since we have
\begin{align*}
   \bb E\brac{\mb Y(\mb g)} \;=\; \frac{1}{m} \bb E\brac{\sum_{k=1}^m \paren{\mb g_k^*\mb x}^2 \mb g_k \mb g_k^*} \;=\;  \bb E\brac{ \paren{\mb g^*\mb x}^2 \mb g \mb g^*},
\end{align*}
we obtain the desired result by applying Lemma 20 of \cite{sun2016geometric}.	}
\end{proof}

\begin{lemma}\label{lem:infty-moments}
   Suppose $\wt{\mb g} \sim \mc {CN}(\mb 0,2\mb I)$, for any positive integer $p \geq 1$, we have
   \begin{align*}
     \paren{ \bb E_{\wt{\mb g}}\brac{ \norm{ \wt{\mb g} \conv \mb x }{\infty }^p } }^{1/p} \;\leq\; 6\sqrt{\log m} \sqrt{p}.
   \end{align*}
\end{lemma}
\begin{proof}
By Minkowski inequality, we have
\begin{align*}
   \bb E_{\wt{\mb g}} \brac{ \norm{ \wt{\mb g} \conv \mb x }{ \infty }^p  }^{1/p}\;\leq\; \bb E\brac{ \norm{\wt{\mb g} \conv \mb x }{\infty} } + \paren{ \bb E_{\wt{\mb g}} \brac{ \paren{ \norm{\wt{\mb g} \conv \mb x}{\infty} - \bb E\brac{ \norm{\wt{\mb g} \conv \mb x }{\infty} }  }^p  } }^{1/p}.
\end{align*}
We know that $\norm{\wt{\mb g}\conv \mb x}{\infty}$ is $1$-Lipschitz w.r.t. $ \wt{\mb g}$. Thus, by Gaussian concentration inequality in Lemma \ref{lem:gauss-concentration}, we have
\begin{align*}
   \bb P\paren{ \abs{ \norm{\wt{\mb g} \conv \mb x}{\infty} - \bb E_{\wt{\mb g}}\brac{ \norm{ \wt{\mb g} \conv \mb x }{\infty} }  } \geq t }	 \;\leq\; 2\exp\paren{ - t^2/2 }.
\end{align*}
By Lemma \ref{lem:sub-Gaussian}, we know that $\norm{\wt{\mb g} \conv \mb x}{\infty}$ is sub-Gaussian, and satisfies
\begin{align*}
   \paren{\bb E_{\wt{\mb g}}\brac{ \abs{ \norm{\wt{\mb g} \conv \mb x}{\infty} - \bb E_{\wt{\mb g}}\brac{ \norm{ \wt{\mb g} \conv \mb x }{\infty} }  }^p } }^{1/p} \;\leq\; 4\sqrt{p}	.
\end{align*}
Besides, let $\wt{\mb g} \conv \mb x = \begin{bmatrix}
	\wt{\mb g}_1^*\mb x \\ \vdots \\ \wt{\mb g}_m^*\mb x
\end{bmatrix}$, then by Jensen's inequality, for all $\lambda>0$, we have
\begin{align*}
   	\exp\paren{ \lambda \bb E\brac{ \norm{\wt{\mb g} \conv \mb x}{\infty}  } } \;\leq\; \bb E\brac{ \exp\paren{ \lambda \norm{\wt{\mb g} \conv \mb x}{\infty} } }\;&=\; \bb E\brac{ \max_{1\leq k\leq m} \exp\paren{\lambda \wt{\mb g}_k^*\mb x } } \\
   	\;&\leq\; \sum_{k=1}^m \bb E\brac{ \exp\paren{ \lambda \wt{\mb g}_k^*\mb x } } \;\leq\; m \exp\paren{ \lambda^2 },
\end{align*}
where we used the fact that the moment generating function of $\wt{\mb g}_k^*\mb x$ satisfies $\bb E\brac{ \exp\paren{ \lambda \wt{\mb g}_k^*\mb x } } \leq \exp\paren{\lambda^2}$. Taking the logarithms on both sides, we have
\begin{align*}
   \bb E\brac{ \norm{ \wt{\mb g} \conv \mb x }{ \infty } } \;\leq\; \log m /\lambda + \lambda	.
\end{align*}
Taking $\lambda = \sqrt{\log m}$, \edited{so that the right hand side of the inequality above} achieves the minimum, which is
\begin{align*}
   	\bb E\brac{ \norm{ \wt{\mb g} \conv \mb x }{ \infty } } \;\leq\; 2 \sqrt{\log m}.
\end{align*}
Combining the results above, we obtain the desired result.
\end{proof}

\subsection{Concentration of \boldmath{$M(g)$}}

Given $\mb M(\mb g)$ as in \eqref{eqn:M-def}, let us define 
\begin{align}
	\mb H(\mb g) \;&=\; \mb P_{\mb x^\perp} \mb M \mb P_{\mb x^\perp}  \label{eqn:H}
\end{align}
and correspondingly its decoupled term
\begin{align}
\mc Q_{dec}^{\mb H}(\mb g^1, \mb g^2) \;&=\; \mb P_{\mb x^\perp}  \mc Q_{dec}^{\mb M} (\mb g^1, \mb g^2)  \mb P_{\mb x^\perp}, \label{eqn:decoupled-H}	
\end{align}
and let
\begin{align}
   \eta_{\sigma^2} (t) \;=\; 1-  2\pi \sigma^2 \xi_{\sigma^2 -\frac{1}{2} }(t), \qquad \nu_{\sigma^2}(t) \;=\; 1 - \frac{4 \pi \sigma^4 }{ 2 \sigma^2 - 1} \xi_{\sigma^2 - \frac{1}{2}}(t), \label{eqn:eta-nu}	
\end{align}
where $\sigma^2>1/2$. In this subsection, we show the following result.

\begin{theorem}\label{thm:M-H-concentration}
For any $\delta >0$, when $m \geq C\delta^{-2} \norm{\mb C_{\mb x}}{}^2  n \log^4 n$, we have 
\begin{align} 
   \norm{ \mb H(\mb g) - \mb P_{\mb x^\perp} }{} \; &\leq \; \delta \label{eqn:H-bound} \\
   \norm{\mb M(\mb g) -  \mb I - \frac{2\sigma^2 }{ 1+2\sigma^2 }\mb x \mb x^* }{}	  \;&\leq \; 3\delta \\
   \norm{ \mb P_{\mb x^\perp}\mb M(\mb g) - \mb P_{\mb x^\perp} }{} \;&\leq \; 2 \delta 
\end{align}
holds with probability at least $1 - c m^{-c'\log^3 n}$, where $c,c'$ and $C$ are some positive numerical constants depending only on $\sigma^2$.
\end{theorem}

\begin{proof}
Let $\mc Q_{dec}^{\mb H}(\mb g^1, \mb g^2)$ be defined as \eqref{eqn:decoupled-H}. We calculate its expectation with respect to $\mb \delta$,
\begin{align*}
&\bb E_{ \mb \delta}\brac{ \mc Q_{dec}^{\mb H}(\mb g + \mb \delta, \mb g - \mb \delta)  } \\
=\;& \frac{1+ 2\sigma^2 }{m} \mb P_{\mb x^\perp} \mb R_{[1:n]} \mb C_{\mb g}^* \diag \paren{ \bb E_{\mb \delta} \brac{  \eta_{\sigma^2  }\paren{ \abs{ \paren{\mb g- \mb \delta} \conv \mb x } } }} \mb C_{\mb g} \mb R_{[1:n]}^\top \mb P_{\mb x^\perp} \\
& \quad + \frac{1+2\sigma^2}{m} \innerprod{ \mb 1 }{ \bb E_{\mb \delta}\brac{ \eta_{ \sigma^2 }\paren{ (\mb g - \mb \delta) \conv \mb x }  } } \mb P_{\mb x^\perp}  \\
=\;& \frac{1+2\sigma^2}{m}  \brac{ \mb P_{\mb x^\perp} \mb R_{[1:n]} \mb C_{\mb g}^* \diag\paren{ \zeta_{\sigma^2}( \abs{\mb g \conv \mb x} ) } \mb C_{\mb g} \mb R_{[1:n]}^\top \mb P_{\mb x^\perp} +  \innerprod{\mb 1}{ \zeta_{\sigma^2}\paren{\abs{\mb g \conv \mb x} }  }  \mb P_{\mb x^\perp} },
\end{align*}
\edited{where the last equality follows from the second equality in Lemma \ref{lem:xi-eta-zeta-expectation}. } Using the results above and Lemma \ref{lem:M-H-expectation}, for all integer $p \geq 1$, we observe
\begin{align*}
    &\paren{\bb E\brac{ \norm{ \mb H - \bb E \brac{\mb H} }{}^p } }^{1/p}  \\
    =\;& \paren{ \bb E_{\mb g} \brac{ \norm{ \frac{1+2\sigma^2}{m} \mb P_{\mb x^\perp} \mb R_{[1:n]} \mb C_{\mb g}^* \diag\paren{ \zeta_{\sigma^2} ( \abs{\mb g \conv \mb x} ) } \mb C_{\mb g} \mb R_{[1:n]}^\top \mb P_{\mb x^\perp} - \mb P_{\mb x^\perp} }{}^p  }  }^{1/p} \\
   =\;& \paren{ \bb E_{\mb g} \brac{ \norm{ \bb E_{\mb \delta}\brac{ \mc Q_{dec}^{\mb H}(\mb g+ \mb \delta, \mb g - \mb \delta)  } - \mb P_{\mb x^\perp} -  \frac{1+2\sigma^2}{m} \innerprod{\mb 1}{ \zeta_{\sigma^2}\paren{\abs{\mb g \conv \mb x} }  } \mb P_{\mb x^\perp}  }{}^p  } }^{1/p} \\
   \leq\;& \paren{ \bb E_{\mb g} \brac{ \norm{ \bb E_{\mb \delta}\brac{ \mc Q_{dec}^{\mb H}(\mb g+ \mb \delta, \mb g - \mb \delta)  } - 2\mb P_{\mb x^\perp} }{}^p  }  }^{1/p} + \paren{  \bb E_{\mb g} \brac{ \abs{ 1- \frac{1+2\sigma^2}{ m } \innerprod{\mb 1}{ \zeta_{\sigma^2}\paren{\abs{\mb g \conv \mb x} }  } }^p } }^{1/p} \\
   \leq\;& \paren{ \bb E_{\mb g^1,\mb g^2} \brac{ \norm{  \mc Q_{dec}^{\mb M}(\mb g^1, \mb g^2)   - 2\mb I }{}^p  }  }^{1/p} + \paren{  \bb E_{\mb g} \brac{ \abs{ 1- \frac{1+2\sigma^2}{ m } \innerprod{\mb 1}{ \zeta_{\sigma^2}\paren{\abs{\mb g \conv \mb x} }  } }^p } }^{1/p},
\end{align*}
where $\mc Q_{dec}^{\mb M}(\mb g^1, \mb g^2)$ is defined as \eqref{eqn:M-def}, and we have used the Minkowski's inequality and the Jensen's inequality, respectively. By Lemma \ref{lem:moments-bound-2} and Lemma \ref{lem:moments-bound-1}, we obtain
\begin{align*}
   \paren{\bb E\brac{ \norm{ \mb H - \bb E \brac{\mb H} }{}^p } }^{1/p} \leq C_{\sigma^2} \paren{   \sqrt{\frac{n}{m}} \log^{3/2} n\log^{1/2}m + \sqrt{p} \sqrt{\frac{n}{m}}  + p \frac{n}{m}   },
\end{align*}
where $C_{\sigma^2}$ is some numerical constant depending only on $\sigma^2$. Thus, by using the tail bound in Lemma \ref{lem:moments-tail-bound}, for any $t> 0$, we obtain
\begin{align*}
   \bb P\paren{ \norm{ \mb H - \bb E\brac{\mb H} }{}  \geq  C_1 \sqrt{\frac{n}{m}} \log^{3/2} n\log^{1/2}m +  t  } \;\leq\; 2\exp\paren{  - C_2 \frac{mt^2}{n}   }
\end{align*}
for some constants $C_1,C_2>0$. This further implies that for any $\delta >0$, if $m \geq C_3 \delta^{-2} n \log^3 n \log m$ for some positive numerical constant $C_3$, we have
\begin{align*}
  \norm{ \mb H - \bb E\brac{\mb H} }{} \;\leq \; \delta,	
\end{align*}
holds with probability at least $1 - 2 m^{-C_4\log^3 n}$, where $C_4>0$ is numerical constant. Next, we use this result to bound the term $\norm{\mb M - \bb E\brac{\mb M} }{}$, by Lemma \ref{lem:M-H-expectation}, notice that
\begin{align*}
	\norm{ \mb M - \bb E\brac{\mb M} }{} \;&\leq\; \norm{ \mb P_{\mb x^\perp} \paren{ \mb M - \bb E\brac{\mb M}  } \mb P_{\mb x^\perp} }{} + 2 \norm{ \mb P_{\mb x^\perp}\paren{\mb M - \bb E\brac{\mb M} \mb P_{\mb x} } }{} + \norm{ \mb P_{\mb x}\paren{ \mb M - \bb E\brac{\mb M}  } \mb P_{\mb x}  }{} \\
	\;&\leq \;  \norm{ \mb H - \bb E\brac{ \mb H } }{} + 2\norm{ \mb P_{\mb x^\perp} \mb M \mb x }{} + \abs{ \mb x^*\paren{ \mb M - \bb E\brac{\mb M}  } \mb x }.
\end{align*}
Hence, by using the results in Lemma \ref{lem:M-x-term} and Lemma \ref{lem:M-x-cross}, whenever $m \geq C \norm{\mb C_{\mb x}}{}^2 \delta^{-2} n\log^4 n $ we obtain
\begin{align*}
   \norm{ \mb M - \bb E\brac{\mb M} }{} \;\leq\; 3 \delta,	
\end{align*}
holds with probability at least $1 - c m^{-c'\log^3 n}$. Here $c,c'>0$ are some numerical constants. Similarly, we have
\begin{align*}
   \norm{ \mb P_{\mb x^\perp}\paren{ \mb M - \bb E\brac{\mb M} } }{} \;&\leq\; \norm{ \mb P_{\mb x^\perp}\paren{ \mb M - \bb E\brac{\mb M} } \mb P_{\mb x^\perp} }{} + \norm{ \mb P_{\mb x^\perp} \paren{ \mb M - \bb E\brac{\mb M}  }\mb P_{\mb x} }{} \\
   \;&=\; \norm{\mb H - \bb E\brac{\mb H}}{} + \norm{ \mb P_{\mb x^\perp} \mb M \mb x }{}.	
\end{align*}
Again, by Lemma \ref{lem:M-x-cross}, we have
\begin{align*}
   \norm{ \mb P_{\mb x^\perp}\paren{ \mb M - \bb E\brac{\mb M} } }{} \;\leq\; 2 \delta,	
\end{align*}
holds with probability at least $1 - c m^{-c'\log^3 n}$. By using Lemma \ref{lem:M-H-expectation}, we obtain the desired results.
\end{proof}

\begin{lemma}\label{lem:moments-bound-2}
Suppose $\mb g^1, \mb g^2$ are independent with $\mb g^1, \mb g^2 \sim \mc {CN}(\mb 0, 2\mb I)$, and let $\mc Q_{dec}^{\mb M}(\mb g^1, \mb g^2)$ be defined as \eqref{eqn:decoupled-M}, then for any integer $p \geq 1$, we have
\begin{align}
\paren{\bb E_{\mb g^1,\mb g^2}\brac{ \norm{ \mc Q_{dec}^{\mb M}(\mb g^1, \mb g^2) - 2\mb I}{}^p } }^{1/p} \leq C_{\sigma^2} \paren{   \sqrt{\frac{n}{m}} \log^{3/2} n\log^{1/2}m + \sqrt{p} \sqrt{\frac{n}{m}}  + p \frac{n}{m}   },
\end{align}
where $C_{\sigma^2}>0$ is some numerical constant only depending on $\sigma^2$.
\end{lemma}

\begin{proof}
Let $\mb b = \paren{2\sigma^2 +1 } \eta_{\sigma^2}\paren{ \mb g^2 \conv \mb x}$, set $\mb b =\begin{bmatrix} b_1 \\ \vdots \\ b_m	\end{bmatrix}$, and write $\mc Q_{dec}^{\mb M}(\mb g^1, \mb g^2) = \frac{1}{m} \mb R_{[1:n]}\mb C_{\mb g}^* \diag\paren{\mb b} \mb C_{\mb g} \mb R_{[1:n]}^\top$. By Minkowski's inequality, we observe
\begin{align*}
	 \paren{\bb E_{\mb g^1,\mb g^2}\brac{ \norm{ \mc Q_{dec}^{\mb M}(\mb g^1, \mb g^2) - 2\mb I}{}^p } }^{1/p}  \leq \underbrace{ \paren{ \bb E_{\mb g^1,\mb g^2} \brac{ \norm{  \mc Q_{dec}^{\mb M}(\mb g^1, \mb g^2) -  \frac{2 }{m} \sum_{k=1}^m b_k \mb I}{}^p  }  }^{1/p} }_{\mc T_1} \\
	+ 2 \underbrace{ \norm{ \frac{1+2\sigma^2}{ m } \innerprod{\mb 1}{ \eta_{\sigma^2}\paren{\abs{\mb g^2 \conv \mb x} }  } -1 }{L^p} }_{\mc T_2} .
\end{align*}
For the term $\mc T_1$, conditioned $\mb g^2$ so that $\mb b$ is fixed, Theorem \ref{thm:moments-norm-circ} implies that for any integer $p\geq 1$,
\begin{align*}
   &\paren{ \bb E_{\mb g^1} \brac{ \norm{  \mc Q_{dec}^{\mb M}(\mb g^1, \mb g^2) -  \frac{2 }{m} \sum_{k=1}^m b_k \mb I}{}^p \mid \mb g^2  }  }^{1/p}   	\\
    \leq\;& C_{\sigma^2} \norm{\mb b}{\infty}\paren{  \sqrt{\frac{n}{m}}  \log^{3/2} n\log^{1/2}m + \sqrt{p} \sqrt{\frac{n}{m}}  + p \frac{n}{m}   },
\end{align*}
where $C_{\sigma^2}>0$ is some numerical constant depending only on $\sigma^2$.
Given the fact that $\norm{\mb b}{\infty} \leq c_{\sigma^2} $ for some constant $c_{\sigma^2}>0$, and for any choice of $\mb g^2$, we have
\begin{align*}
\mc T_1 \; \leq \;  C_{\sigma^2} \paren{  \sqrt{\frac{n}{m}}  \log^{3/2} n\log^{1/2}m + \sqrt{p} \sqrt{\frac{n}{m}}  + p \frac{n}{m}   }.
\end{align*}
For the term $\mc T_2$, Lemma \ref{lem:moments-bound-1} implies that 
\begin{align*}
	\mc T_2 = \norm{ \frac{1+2\sigma^2}{ m } \innerprod{\mb 1}{ \eta_{\sigma^2}\paren{\abs{\mb g^2 \conv \mb x} }  } -1 }{L^p} \;\leq\; \frac{C_{\sigma^2}' }{ \sqrt{m} } \norm{\mb C_{\mb x}}{}   \sqrt{p},
\end{align*}
for some constant $C_{\sigma^2}'>0$. Combining the results above and use the fact that $\norm{\mb C_{\mb x}}{} \leq \sqrt{n}$, we obtain
    \begin{align*}
       	\paren{ \bb E_{\mb g^1,\mb g^2} \brac{ \norm{  \mc Q_{dec}^{\mb M}(\mb g^1, \mb g^2) - 2\mb I}{}^p  }  }^{1/p} \;\leq\; C_{\sigma^2}'' \paren{   \sqrt{\frac{n}{m}} \log^{3/2} n\log^{1/2}m + \sqrt{p} \sqrt{\frac{n}{m}}  + p \frac{n}{m}   },
    \end{align*}
    where $C_{\sigma^2}''>0$ is some numerical constant only depending on $\sigma^2$.
\end{proof}

\begin{lemma}\label{lem:M-x-term}
Let $\mb g\in \bb C^m$ be a complex Gaussian random variable $\mb g \sim \mc {CN}(\mb 0,\mb I)$. Let $\mb M(\mb g)$ be defined as \eqref{eqn:M-def}. For any $\delta \geq 0$, whenever $m \geq C_{\sigma^2} \delta^{-1} \norm{\mb C_{\mb x}}{}^2 n \log m $, we have
\begin{align*}
   \abs{ \mb x^* \paren{ \mb M - \bb E\brac{\mb M}  } \mb x } \;\leq\; \delta
\end{align*}
holds with $1 - m^{ - C_{\sigma^2}' \norm{\mb C_{\mb x}}{}^2  n } $. Here, $C_{\sigma^2},\;C_{\sigma^2}'$ are some numerical constants depending on $\sigma^2$.
\end{lemma}

\begin{proof}
Let $h(\mb g) = \abs{ \mb x^* \mb M(\mb g) \mb x }^{1/2} =  \sqrt{\frac{2\sigma^2+1 }{m}} \norm{ \diag\paren{ \zeta_{\sigma^2}^{1/2}(\mb C_{\mb x}\mb g) } \mb C_{\mb x} \mb g }{} $.
Then we have its Wirtinger gradient
\begin{align*}
	\frac{ \partial }{ \partial \mb z } h(\mb g) =  \frac{1}{2} \sqrt{ \frac{ 2\sigma^2+1 }{ m } } \norm{ \diag\paren{ \zeta_{\sigma^2}^{1/2}\paren{ \mb C_{\mb x} \mb g }   } \mb C_{\mb x} \mb g  }{}^{-1} \brac{  \mb C_{\mb x}^* \diag\paren{ \zeta_{\sigma^2}(\mb C_x\mb g) } \mb C_{\mb x}\mb g +  \mb C_{\mb x}^* \diag\paren{ f ( \mb C_{\mb x} \mb g ) } \mb C_{\mb x} \mb g  },
\end{align*}
where $g_1(t) = \frac{\abs{t}^2}{2\sigma^2} \exp\paren{ -\frac{ \abs{t}^2 }{ 2\sigma^2 }  }  $, so that
\begin{align*}
 	\norm{\nabla_{\mb g} h(\mb g)}{} \;=\;&   \sqrt{\frac{2\sigma^2+1}{m}} \norm{ \diag\paren{ \zeta_{\sigma^2}^{1/2}(\mb C_{\mb x}\mb g) } \mb C_{\mb x} \mb g }{}^{-1} \times \\
 	& \norm{ \mb C_{\mb x}^* \diag \paren{  \zeta_{\sigma^2}(\mb C_{\mb x}\mb g) } \mb C_{\mb x} \mb g + \mb C_{\mb x}^* \diag\paren{ g_1 ( \mb C_{\mb x} \mb g ) } \mb C_{\mb x} \mb g   }{}.
\end{align*}
Thus, we have
\begin{align*}
	\norm{\nabla_{\mb g} h(\mb g)}{} \;\leq\; \sqrt{\frac{2\sigma^2+1}{m}} \norm{ \mb C_{\mb x} }{} \paren{\norm{  \diag \paren{ \zeta_{\sigma^2}^{1/2}\paren{ \mb C_{\mb x} \mb g }  } }{} +  \norm{ \diag\paren{ g_2(\mb C_{\mb x} \mb g  }   }{}  },
\end{align*}
where $ g_2(t) = g_1(t)  \zeta_{\sigma^2}^{-1/2} (t)$. By using the fact that $ \norm{ \zeta_{\sigma^2}^{1/2} }{ \ell^\infty } \leq 1$ and $ \norm{ g_2 }{ \ell^\infty } \leq C_1 $ for some constant $C_1>0$, we have
\begin{align*}
    \norm{\nabla_{\mb g} h(\mb g)}{} \;\leq\; C_2\sqrt{\frac{2\sigma^2+1}{m}}  \norm{ \mb C_{\mb x} }{},
\end{align*}
for some constant $C_2>0$. Therefore, we can see that the Lipschitz constant $L$ of $h(\mb g)$ is  bounded by $C_2\sqrt{\frac{2\sigma^2+1}{m}}  \norm{ \mb C_{\mb x} }{}$. Thus, by the Gaussian concentration inequality, we observe
\begin{align}\label{eqn:sub-Gaussian-tail-1}
   \bb P \paren{ \abs{ h(\mb g) - \bb E\brac{ h(\mb g) } } \geq t } \;\leq\; 2 \exp \paren{ - \frac{C_{\sigma^2} m t^2}{ \norm{\mb C_{\mb x}}{}^2  }  }
\end{align}
holds with some constant $C_{\sigma^2}>0$ depending only on $\sigma^2$. Thus, we have
\begin{align}\label{eqn:sub-Gaussian-tail-2}
  -t \;\leq\; h(\mb g) - \bb E\brac{h (\mb g)} \;\leq\; t	
\end{align}
holds with probability at least $1- 2 \exp \paren{ - \frac{C_{\sigma^2} m t^2}{ \norm{\mb C_{\mb x}}{}^2  }  }$. By Lemma \ref{lem:M-H-expectation}, we know that
\begin{align*}
 \bb E\brac{h^2(\mb g) } \;=\; \mb x^* \bb E\brac{\mb M(\mb g) } \mb x	\;=\; \frac{4\sigma^2+1}{ 2\sigma^2+1 }.
\end{align*}
This implies that 
\begin{align}\label{eqn:square-bound-1}
   h^2(\mb g) \;\leq \; \paren{\bb E\brac{ h(\mb g) } +t}^2 \; \Longrightarrow \; h^2(\mb g) - \bb E\brac{ h^2(\mb g) }\;\leq\; 2t \sqrt{ \bb E\brac{ h^2(\mb g) }  } +t^2 \;\leq\; 2t  \sqrt{\frac{ 1+ 4 \sigma^2 }{ 1+ 2\sigma^2 } } + t^2
\end{align}
holds with probability at least $1- 2 \exp \paren{ - \frac{C_{\sigma^2} m t^2}{ \norm{\mb C_{\mb x}}{}^2  }  }$. On the other hand, \eqref{eqn:sub-Gaussian-tail-1} also implies that $h(\mb g)$ is subgaussian, Lemma \ref{lem:sub-Gaussian} implies that
\begin{align*}
   \bb E\brac{ \paren{ h(\mb g) - \bb E\brac{ h(\mb g) } }^2  } \leq \frac{ C_{\sigma^2}' \norm{\mb C_{\mb x}}{}^2 }{ m } \; \Longrightarrow \; \bb E\brac{ h^2(\mb g) } \leq \paren{ \bb E\brac{ h(\mb g) } }^2 + \frac{C_{\sigma^2}' \norm{\mb C_{\mb x} }{}^2}{m}
\end{align*}
for some constant $C_{\sigma^2}'>0$ only depending on $\sigma^2$. Suppose $m \geq C_{\sigma^2}'' \norm{ \mb C_{\mb x}}{}^2 $ for some large constant $C_{\sigma^2}''>0$  depending on $\sigma^2>0$, from \eqref{eqn:sub-Gaussian-tail-2}, we have
\begin{align*}
   h(\mb g) \geq \bb E\brac{ h(\mb g) } - t \;\geq\; \sqrt{ \bb E\brac{ h^2(\mb g) } - \frac{ C_{\sigma^2}' \norm{ \mb C_{\mb x} }{}^2  }{ m }  }	 -t.
\end{align*}
Suppose $t \leq \sqrt{ \bb E\brac{ h^2(\mb g) } - \frac{ C_{\sigma^2}' \norm{ \mb C_{\mb x} }{}^2  }{ m }  }$, by squaring both sides, we have
\begin{align*}
   h^2(\mb g) \;\geq\;  \bb E\brac{ h^2(\mb g) } - \frac{ C_{\sigma^2}' \norm{\mb C_{\mb x}}{}^2 }{m }  +t^2 - 2t \sqrt{  \bb E\brac{ h^2(\mb g) } - \frac{ C_{\sigma^2}' \norm{ \mb C_{\mb x} }{}^2  }{ m }  } .
\end{align*}
This further implies that
\begin{align}\label{eqn:square-bound-2}
   h^2(\mb g) - \bb E\brac{ h^2(\mb g) } \;\geq \;t^2 -2t \sqrt{  \frac{4\sigma^2+1}{2\sigma^2+1 }   - \frac{C_{\sigma^2}' \norm{\mb C_{\mb x}}{}^2 }{ m } }	- \frac{ C_{\sigma^2}' \norm{\mb C_{\mb x}}{}^2 }{m },
\end{align}
holds $1- 2 \exp \paren{ - \frac{C_{\sigma^2}m t^2}{ \norm{\mb C_{\mb x}}{}^2  }  }$. Therefore, combining the results in \eqref{eqn:square-bound-1} and \eqref{eqn:square-bound-2}, for any $\delta \geq 0$, whenever $m \geq C_4 \delta^{-1} \norm{\mb C_{\mb x} }{}^2 n\log m $, choosing $t = C_5 \delta$, we have
\begin{align*}
  \abs{ h^2(\mb g) - \bb E\brac{h^2(\mb g) }	 }\; \leq\; \delta,
\end{align*}
holds with probability at least $ 1 - m^{ - C_6 \norm{\mb C_{\mb x}}{}^2  n } $.   	
\end{proof}

\begin{lemma}\label{lem:M-x-cross}
Let $\mb g\in \bb C^m$ be a complex Gaussian random variable $\mb g \sim \mc {CN}(\mb 0,\mb I)$, and let $\mb M(\mb g)$ be defined as \eqref{eqn:M-def}. For any $\delta>0$, whenever $m \geq C_{\sigma^2} \delta^{-2} \norm{\mb C_{\mb x}}{}^2  n \log^4 n$, we have
\begin{align*}
  \norm{ \mb P_{\mb x^\perp} \mb M \mb x}{} \;\leq\; \delta,
\end{align*}
holds with probability at least $1 - 2m^{-c_{\sigma^2} \log^3 n}$. Here, $c_{\sigma^2} ,C_{\sigma^2} $ are some positive constants only depending on $\sigma^2$.
\end{lemma}

\begin{proof}
First, let us define decoupled terms
\begin{align}
  \mc Q_{dec}^{\mb M \mb x^\perp} (\mb g^1, \mb g^2) \;&=\; \frac{2\sigma^2+1}{m} \mb P_{\mb x^\perp} \mb R_{[1:n]} \mb C_{\mb g^1 }^* \diag\paren{ \nu_{\sigma^2}\paren{ \mb g^2 \conv \mb x } } \mb C_{\mb g^2} \mb R_{[1:n]}^\top \mb x, \\
  \mc Q_{dec}^{\mb H \mb x^\perp} (\mb g^1, \mb g^2) \;&=\; \frac{2\sigma^2+1}{m}  \mb R_{[1:n]} \mb C_{\mb g^1 }^* \diag\paren{ \nu_{\sigma^2}\paren{ \mb g^2 \conv \mb x } } \mb C_{\mb g^2} \mb R_{[1:n]}^\top \mb x, \label{eqn:Q-dec-Hxperp}
\end{align}
where $\nu_{\sigma^2}(t)$ is defined in \eqref{eqn:eta-nu}. Let $\mb C_{\mb g} \mb R_{[1:n]}^\top = \begin{bmatrix}
 \mb g_1^* \\ \vdots \\ \mb g_m^*	
 \end{bmatrix}
 $ and $\mb C_{\mb \delta} \mb R_{[1:n]}^\top = \begin{bmatrix}
 \mb \delta_1^* \\ \vdots \\ \mb \delta_m^*	
 \end{bmatrix}
 $ , then by Lemma \ref{lem:xi-eta-zeta-expectation}, we observe
\begin{align*}
   \bb E_{\mb \delta}\brac{ \mc Q_{dec}^{\mb M \mb x^\perp} (\mb g+ \mb \delta, \mb g - \mb \delta) } 
   \;=\;& \frac{2\sigma^2+1}{m} \bb E_{\mb \delta}\brac{ \mb P_{\mb x^\perp} \mb R_{[1:n]}\mb C_{\mb g+ \mb \delta }^* \diag\paren{ \nu_{\sigma^2}\paren{ (\mb g - \mb \delta) \conv \mb x }  } \mb C_{\mb g - \mb \delta } \mb R_{[1:n]}^\top \mb x } \\
   =\;& \frac{2\sigma^2+1}{m} \sum_{k=1}^m  \bb E_{\mb \delta} \brac{ \nu_{\sigma^2} \paren{   (\mb g_k - \mb \delta_k )^* \mb x }  \mb P_{\mb x^\perp} (\mb g_k+ \mb \delta_k) \paren{ \mb g_k - \mb \delta_k }^* \mb x } \\
   =\;& \frac{2\sigma^2+1}{m} \sum_{k=1}^m \mb P_{\mb x^\perp} \mb g_k \bb E_{\mb \delta_k^* \mb x } \brac{ \nu_{\sigma^2} \paren{   (\mb g_k - \mb \delta_k )^* \mb x } \paren{ \mb g_k - \mb \delta_k }^* \mb x } \\
   =\;& \frac{2\sigma^2+1}{m} \sum_{k=1}^m \zeta_{\sigma^2} \paren{ \mb g_k^*\mb x } \mb P_{\mb x^\perp} \mb g_k \mb g_k^*\mb x  \\
   =\;&  \frac{2\sigma^2+1}{m} \mb P_{\mb x^\perp} \mb R_{[1:n]} \mb C_{\mb g}^* \diag \paren{ \zeta_{\sigma^2} \paren{ \mb g \conv \mb x }  } \mb C_{\mb g} \mb R_{[1:n]}^\top \mb x.
\end{align*}
Thus, for any integer $p \geq 1$, we have
\begin{align*}
   &\paren{\bb E_{\mb g}\brac{ \norm{ \frac{2\sigma^2+1}{m} \mb P_{\mb x^\perp} \mb R_{[1:n]} \mb C_{\mb g}^* \diag \paren{ \zeta_{\sigma^2} \paren{ \mb g \conv \mb x }  } \mb C_{\mb g} \mb R_{[1:n]}^\top \mb x }{}^p } }^{1/p} \\
   \;=\;& \paren{\bb E_{\mb g} \brac{ \norm{  \bb E_{\mb \delta}\brac{ \mc Q_{dec}^{\mb M \mb x^\perp} (\mb g+ \mb \delta, \mb g - \mb \delta) }    }{}^p } }^{1/p} \\
   \;\leq\;& \paren{ \bb E_{\mb g^1,\mb g^2} \brac{ \norm{  \mc Q_{dec}^{\mb M \mb x^\perp} (\mb g^1, \mb g^2)  }{}^p   }  }^{1/p} \leq \paren{ \bb E_{\mb g^1,\mb g^2} \brac{ \norm{ \mc Q_{dec}^{\mb H \mb x^\perp} (\mb g^1, \mb g^2) }{}^p   }  }^{1/p}.
\end{align*}
By Lemma \ref{lem:moments-cross-term}, we have
\begin{align*}
   \paren{ \bb E_{\mb g^1,\mb g^2} \brac{ \norm{  \mc Q_{dec}^{\mb H \mb x^\perp} (\mb g^1, \mb g^2)  }{}^p   }  }^{1/p} \;\leq\; C_{\sigma^2} \norm{\mb C_{\mb x}}{} \brac{  \sqrt{\frac{n}{m}}\paren{1+ \sqrt{\frac{n}{m}} \log^{3/2} n \log^{1/2}m } \sqrt{p}   + \frac{n}{m} p }.
\end{align*}
Therefore, by Lemma \ref{lem:moments-tail-bound}, finally for any $\delta>0$, whenever $m \geq C \delta^{-2} \norm{\mb C_{\mb x}}{}^2  n \log^4 n $ we obtain
\begin{align*}
   \bb P\paren{ \norm{ \mb P_{\mb x^\perp} \mb M \mb x}{} \;\geq\; \delta } 
   \;\leq\; 2 m^{- c \log^3 n},	
\end{align*}
where $c,C>0$ are some positive constants.
\end{proof}

\begin{lemma}\label{lem:moments-cross-term}
Let $\mb g^1$ and $\mb g^2$ be random variables defined as in \eqref{eqn:g1-g2}, and let $  \mc Q_{dec}^{\mb H \mb x^\perp} (\mb g^1, \mb g^2)$ be defined as \eqref{eqn:Q-dec-Hxperp}. Then for any integer $p \geq 1$, we have
\begin{align*}
   \paren{ \bb E_{\mb g^1,\mb g^2} \brac{ \norm{  \mc Q_{dec}^{\mb H \mb x^\perp} (\mb g^1, \mb g^2)  }{}^p   }  }^{1/p} \leq C_{\sigma^2} \norm{\mb C_{\mb x}}{} \brac{  \sqrt{\frac{n}{m}}\paren{1+ \sqrt{\frac{n}{m}} \log^{3/2} n \log^{1/2}m } \sqrt{p}   + \frac{n}{m} p },
   \end{align*}
where $C_{\sigma^2}$ is some positive constant only depending on $\sigma^2$.
\end{lemma}

\begin{proof}
First, we fix $\mb g^1$, and let $h(\mb g^2) = \mc Q_{dec}^{\mb H \mb x^\perp} (\mb g^1, \mb g^2) $. Let $g(t) = t \nu_{\sigma^2}(t) $, for which the Lipschitz constant $L_f \leq C_{\sigma^2}$ for some positive constant $C_{\sigma^2}$ only depending on $\sigma^2$. Then given an independent copy $\wt{\mb g^2}$ of $\mb g^2$, we observe
\begin{align*}
   \norm{ h(\mb g^2) - h(\wt{\mb g^2})  }{} \;&\leq\; \frac{2\sigma^2+1}{m} \norm{ \mb R_{[1:n]} \mb C_{\mb g^1}^* }{} \norm{g(\mb C_{\mb x}\mb g^2) - g(\mb C_{\mb x} \wt{\mb g^2}  ) }{} \\
   \;&\leq\; \underbrace{ \frac{C_{\sigma^2}'}{m} \norm{ \mb R_{[1:n]} \mb C_{\mb g^1}^* }{} \norm{\mb C_{\mb x}}{} }_{ L_h }  \norm{ \mb g^2 - \wt{ \mb g^2 } }{} 
\end{align*}
where $L_h$ is the Lipschitz constant of $h(\mb g^2)$. Given the fact that $\bb E_{\mb g^2}\brac{ h(\mb g^2) } = \mb 0$, by Lemma \ref{thm:gaussian-vector-concentration}, for any $t>\sqrt{n} L_h $ we have
\begin{align*}
   \bb P\paren{  \norm{h(\mb g^2) }{} \geq t  } \;\leq\; e \bb P\paren{ \norm{\mb v}{} \;\geq\; \frac{t}{L_h} } \leq e \exp \paren{  - \frac{1}{2} \paren{ \frac{ t}{ L_h } - \sqrt{n} }^2 },
\end{align*}
where $\mb v \in \bb R^n$ with $\mb v \sim \mc N(\mb 0, \mb I)$, and we used the Gaussian concentration inequality for the tail bound of $\norm{\mb v}{}$. By a change of variable, we obtain
\begin{align*}
   \bb P \paren{  \norm{h(\mb g^2) }{} \geq t + \sqrt{n} L_h  }	 \;\leq\; e \exp\paren{ - \frac{1}{2L_h^2} t^2 }
\end{align*}
holds for all $t>0$. By using the tail bound above, we obtain
\begin{align*}
&\bb E_{\mb g^2}\brac{ \norm{  \mc Q_{dec}^{\mb H \mb x^\perp} (\mb g^1, \mb g^2) }{}^p } \\
=\;& \int_{t=0}^\infty  \bb P\paren{ \norm{h(\mb g^2) }{}^p \geq t } dt	 \\
=\;&  \int_{t=0}^{ \paren{\sqrt{n }L_h }^p } \bb P\paren{ \norm{h(\mb g^2) }{}^p \geq t} dt + \int_{t = \paren{\sqrt{n}L_h}^p }^\infty \bb P\paren{ \norm{h(\mb g^2) }{} \geq t^{1/p} }  dt \\
\leq\;& \paren{ \sqrt{n} L_h }^p + p\int_{u = \sqrt{n}L_h}^\infty \bb P\paren{ \norm{h(\mb g^2)}{} \geq u  } u^{p-1} du \\
=\;& \paren{ \sqrt{n} L_h }^p + p\int_{u = 0 }^\infty \bb P\paren{ \norm{h(\mb g^2)}{} \geq u + \sqrt{n}L_h } \paren{u+\sqrt{n}L_h}^{p-1} du \\
\leq\;& \paren{ \sqrt{n} L_h }^p + 2^{p-2}p\paren{ \sqrt{n} L_h }^{p-1} e \int_{u=0}^\infty \exp\paren{ - \frac{u^2}{2L_h^2} }  du + 2^{p-2}p e\int_{u = 0 }^\infty  \exp\paren{ - \frac{u^2}{2L_h^2} }  u^{p-1} du  \\
=\; & \paren{ \sqrt{n} L_h }^p + \sqrt{\frac{\pi}{2}} 2^{p-2} p \sqrt{n}^{p-1} L_h^p e + 2^{3p/2-3} p L_h^p e \int_{\tau=0}^\infty e^{-\tau} \tau^{\frac{p}{2}-1} d\tau  \\
\leq \; & 3\sqrt{n}^p L_h^p \paren{1+ \sqrt{\frac{\pi}{2}}2^{p-1}p  + 2^{3p/2-3}p \Gamma(p/2) } \leq 3\paren{4\sqrt{n}L_h}^p p \max\Brac{(p/2)^{p/2}, \sqrt{2\pi} }  ,
\end{align*}
where we used the fact that $\Gamma(p/2) \leq \max\Brac{(p/2)^{p/2}, \sqrt{2\pi} }$ for any integer $p \geq 1$. By Corollary \ref{cor:circ-spectral-norm}, we know that 
\begin{align*}
   \bb E_{\mb g^1}\brac{ \norm{ \mb R_{[1:n]} \mb C_{\mb g^1}^* }{}^p  } \;\leq\; c_{\sigma^2}^p \sqrt{m}^p \paren{ 1+ \sqrt{\frac{n}{m}} \log^{3/2} n \log^{1/2}m + \sqrt{\frac{n}{m}} \sqrt{p} }^p,
\end{align*}
where $c_{\sigma^2}$ is some constant only depending only on $\sigma^2$. Therefore, using the fact that $L_h = C_{\sigma^2}'  \norm{ \mb R_{[1:n]} \mb C_{\mb g^1}^* }{} \norm{\mb C_{\mb x}}{} /m$ and $p^{1/p} \leq e^{1/e} $, we obtain
\begin{align*}
	\paren{\bb E_{\mb g^1,\mb g^2}\brac{ \norm{  \mc Q_{dec}^{\mb H \mb x^\perp} (\mb g^1, \mb g^2) }{}^p } }^{1/p} \;&\leq\; C_{\sigma^2}'' \norm{\mb C_{\mb x}}{} \sqrt{\frac{n}{m}}\paren{ 1+ \sqrt{\frac{n}{m}} \log^{3/2} n \log^{1/2}m + \sqrt{\frac{n}{m}} \sqrt{p} } \sqrt{p} \\
	\;&=\; C_{\sigma^2}'' \norm{\mb C_{\mb x}}{} \brac{  \sqrt{\frac{n}{m}}\paren{1+ \sqrt{\frac{n}{m}} \log^{3/2} n \log^{1/2}m } \sqrt{p}   + \frac{n}{m} p },
\end{align*}
where $C_{\sigma^2}''>0$ is some constant depending only on $\sigma^2$.
\end{proof}

\subsection{Auxiliary Results}
The following are some auxiliary results used in the main proof.
\begin{lemma}\label{lem:xi-eta-zeta-expectation}
Let $\xi_{\sigma^2}, \zeta_{\sigma^2}$, $\eta_{\sigma^2}$ and $\nu_{\sigma^2}$ be defined as \eqref{eqn:xi-zeta} and \eqref{eqn:eta-nu}, for $t \in \bb C$, we have
\begin{align*}
   \bb E_{s \sim \mc {CN}(0,1)} \brac{ \xi_{\sigma^2}(t+s) } \;&=\;\xi_{ \sigma^2 + \frac{1}{2} }(t) \\
   \bb E_{s\sim \mc {CN}(0,1)} \brac{ \eta_{\sigma^2}(t+s) } \;&=\;\zeta_{\sigma^2}(t) \\
   \bb E_{ s \sim \mc {CN}(0,1)} \brac{ \zeta_{\sigma^2}(s) } \;&=\; \frac{1}{2\sigma^2+1} \\
   \bb E_{ s \sim \mc {CN}(0,1) } \brac{  \abs{t}^2 \zeta_{\sigma^2}(s) } \;&=\; \frac{ 4 \sigma^2 +1 }{ \paren{2 \sigma^2 +1}^2 } \\
   \bb E_{ s \sim \mc {CN}(0,2) } \brac{ \eta_{\sigma^2}(s) } \;&=\; \frac{1}{2\sigma^2 +1 } \\
   \bb E_{ s \sim \mc {CN}(0,1) } \brac{ (t+s) \nu_{\sigma^2}(t+s) } \;&=\; t \zeta_{\sigma^2}(t).
\end{align*}
\end{lemma}

\begin{proof}
Let $s_r = \Re(s)$, $s_i = \Im(s)$ and $t_r = \Re(t)$, $t_i =\Im(t)$, by definition, we observe
\begin{align*}
&\bb E_{ s \sim \mc {CN}(0,1) } \brac{ \xi_{\sigma^2} (t+s) } \\
=\;& 	\frac{1}{2\pi \sigma^2} \frac{1}{\pi} \int_s \exp \paren{ - \frac{ \abs{t+s}^2 }{2\sigma^2} } \exp\paren{ - \abs{s}^2 } ds \\
=\;& \frac{1}{2 \pi^2 \sigma^2} \int_{s_r = - \infty}^{+\infty} \exp \paren{ - \frac{  (s_r+ t_r)^2 }{2\sigma^2 }   - s_r^2 }  ds_r \int_{s_i = - \infty}^{+\infty} \exp \paren{ - \frac{  (s_i+ t_i)^2 }{2\sigma^2 }   - s_i^2 }  ds_i \\
=\;& \frac{1}{2\pi \paren{ \sigma^2 + 1/2 } } \exp\paren{ - \frac{\abs{t}^2}{ 2\paren{\sigma^2+1/2 } } } \;=\; \xi_{ \sigma^2+ \frac{1}{2} }(t).
\end{align*}
Thus, by definition of $\eta_{\sigma^2}$ and $\zeta_{\sigma^2}$, we have
\begin{align*}
	\bb E_{s\sim \mc {CN}(0,1)} \brac{ \eta_{\sigma^2}(t+s) } = 1 - 2\pi \sigma^2 \bb E_{s \sim \mc {CN}(0,1)} \brac{ \xi_{\sigma^2 - 1/2}(t+s) } = 1 - 2\pi \sigma^2 \xi_{\sigma^2}(t) = \zeta_{\sigma^2}(t).
\end{align*}
For $\bb E_{ t \sim \mc {CN}(0,1) } \brac{ \zeta_{\sigma^2} (t) }$, we have
\begin{align*}
   \bb E_{ t \sim \mc {CN}(0,1) } \brac{ \zeta_{\sigma^2} (t) }\;&=\; 1 - 2 \pi \sigma^2 \bb E_{t \sim \mc {CN}(0,1)}  \brac { \xi_{\sigma^2}(t) } \\
   \;&=\; 1 - \bb E_{t \sim \mc {CN}(0,1)} \brac{  \exp\paren{ - \frac{\abs{t}^2}{2\sigma^2} } } \\
   \;&=\; 1 - \frac{2\sigma^2}{2\sigma^2+1} = \frac{1}{1+ 2\sigma^2}.
\end{align*}
For $\bb E_{ t \sim \mc {CN}(0,1) } \brac{  \abs{t}^2 \zeta_{\sigma^2}(t) }$, we observe
\begin{align*}
   \bb E_{ t \sim \mc {CN}(0,1) } \brac{  \abs{t}^2 \zeta_{\sigma^2}(t) } \;&=\; \frac{1}{\pi} \int_t \abs{t}^2 \brac{ 1 - \exp\paren{ - \frac{\abs{t}^2}{2\sigma^2} } }	 \exp\paren{ -\abs{t}^2} dt \\
   \;&=\; \bb E_{t \sim \mc {CN}(0,1)} \brac{\abs{t}^2}  - \frac{1}{\pi} \int_t \abs{t}^2 \exp\paren{ - \frac{ 2\sigma^2 +1  }{2\sigma^2} \abs{t}^2  }  dt \\
   \;&=\; 1 -  \frac{2 \sigma^2 }{2\sigma^2+1 } \bb E_{t \sim \mc {CN}\paren{ 0, \frac{2\sigma^2}{2\sigma^2+1} }  } \bb E\brac{ \abs{t}^2 }  \\
   \;&=\;  1 - \paren{\frac{2\sigma^2}{ 2 \sigma^2+1 } }^2 \; = \; \frac{ 4\sigma^2 +1 }{ \paren{2\sigma^2 +1 }^2 }.
\end{align*}
In addition, by using the fact that $\bb E_{s \sim \mc {CN}(0,1)} \brac{ \xi_{\sigma^2}(t+s) } = \xi_{ \sigma^2 + \frac{1}{2} }(t)$, we have
\begin{align*}
   \bb E_{t\sim \mc {CN}(0,2) }\brac{ \eta_{\sigma^2}(t) }	
      \;&=\; \bb E_{t_1, t_2 \sim_{i.i.d.} \mc {CN}(0,1) } \brac{  \eta_{\sigma^2}(t_1+t_2) } \;&=\; \bb E_{t_1 \sim \mc {CN}(0,1)}\brac{ \zeta_{\sigma^2}(t_1) } \;&=\; \frac{1}{1+2\sigma^2}.
\end{align*}
For the last equality, first notice that
\begin{align*}
   \bb E_{s \sim \mc {CN}(0,1)} 	\brac{ s  \xi_{\sigma^2}(t+s) } \;&=\; \frac{1}{\pi} \int_s s \frac{1}{2\pi \sigma^2} \exp\paren{ - \frac{ \abs{t+s}^2 }{2\sigma^2} }  \exp\paren{ - \abs{s}^2 } ds \\
   \;&=\; \frac{1}{ 2\pi^2 \sigma^2 } \exp \paren{ - \frac{ \abs{t}^2 }{ 1+ 2\sigma^2 } } \int_s s \exp\paren{ - \frac{1+ 2\sigma^2 }{2\sigma^2} \abs{ s+ \frac{ t}{ 1+ 2\sigma^2 } }^2  }  ds \\
   \;&=\; \frac{1}{ 2\pi^2 \sigma^2 } \exp \paren{ - \frac{ \abs{t}^2 }{ 1+ 2\sigma^2 } } \times  2 \pi \frac{\sigma^2 }{ 1 + 2\sigma^2}  \times \frac{-t}{1+2\sigma^2} \\
   \;&=\; \frac{-t }{ \pi \paren{1 + 2\sigma^2}^2 } \exp \paren{  - \frac{\abs{t}^2 }{ 1+ 2\sigma^2 } } \;=\; \frac{-t}{1+2\sigma^2} \xi_{\sigma^2+\frac{1}{2}}(t).
\end{align*}
Therefore, we have
\begin{align*}
   \bb E_{s \sim \mc {CN}(0,1)} \brac{ (t+s) \xi_{\sigma^2-\frac{1}{2}}(t+s) } \;&=\; t \bb E_{s \sim \mc {CN}(0,1)} \brac{ \xi_{\sigma^2-\frac{1}{2}}(t+s) }	+ \bb E_{s \sim \mc {CN}(0,1)} \brac{ s \xi_{\sigma^2-\frac{1}{2}}(t+s) } \\
   \;&=\; t \xi_{\sigma^2}(t) - \frac{t}{2\sigma^2} \xi_{\sigma^2}(t) \; = \; \frac{ 2\sigma^2-1 }{2\sigma^2} t\xi_{\sigma^2}(t).
\end{align*}
Using the result above, we observe
\begin{align*}
   \bb E_{ s \sim \mc {CN}(0,1) }\brac{ (t+s) \nu_{\sigma^2}(t+s)  } \;&=\; t - \frac{4 \pi \sigma^4 }{2\sigma^2 -1}   \bb E_{ s \sim \mc {CN}(0,1) }\brac{ (t+s) \xi_{\sigma^2 - \frac{1}{2} } (t+s)  } \\
  \;&=\; t\paren{ 1 - 2\pi \sigma^2 \xi_{\sigma^2}(t) } \;=\;t \zeta_{\sigma^2}(t).  
\end{align*}

\end{proof}

\begin{lemma}\label{lem:M-H-expectation}
Let $\mb g \sim \mc {CN}(\mb 0,\mb I)$, and $\mb M(\mb g)$, $\mb H(\mb g)$ be defined as \eqref{eqn:M-def} and \eqref{eqn:H}, we have
\begin{align*}
 \bb E_{\mb g}\brac{\mb M (\mb g) } = \mb P_{\mb x^\perp} + \frac{1+4\sigma^2 }{ 1+2\sigma^2 } \mb x \mb x^* ,\qquad 
 \bb E_{\mb g}\brac{\mb H (\mb g) } = \mb P_{\mb x^\perp}.
\end{align*}
\end{lemma}

\begin{proof}
By Lemma \ref{lem:xi-eta-zeta-expectation} and suppose $ \mb C_{\mb g} \mb R_{[1:n]}^\top  = \begin{bmatrix}
   \mb g_1^* \\ \vdots \\ \mb g_m^*	
 \end{bmatrix}
  $ , we observe
\begin{align*}
  \bb E\brac{\mb M} \;&=\; \frac{2\sigma^2+1 }{m} \sum_{k=1}^m \bb E\brac{ \zeta_{\sigma^2 }( \mb g_k^* \mb x ) \mb g_k \mb g_k^*  } \\
  \;&=\; \frac{2\sigma^2+1}{m} \sum_{k=1}^m \Brac{\bb E\brac{ \zeta_{\sigma^2 }( \mb g_k^* \mb x )   } \bb E\brac{ \mb P_{\mb x^\perp} \mb g_k \mb g_k^*  \mb P_{\mb x^\perp} } + \bb E\brac{  \zeta_{\sigma^2 }( \mb g_k^* \mb x ) \mb P_{\mb x} \mb g_k \mb g_k^* \mb P_{\mb x} } } \\
  \;&=\; \mb P_{\mb x^\perp} +  \frac{2\sigma^2+1}{m} \mb x \mb x^*  \sum_{k=1}^m \bb E \brac{ \zeta_{\sigma^2}\paren{ \mb g_k^* \mb x  } \abs{\mb g_k^* \mb x }^2 } \\
  \;&=\;\mb P_{\mb x^\perp} + \frac{4\sigma^2+1}{ 2\sigma^2 +1 } \mb x \mb x^* .
\end{align*}
Thus, we have
\begin{align*}
   \bb E\brac{\mb H } \;=\; \mb P_{\mb x^\perp} \bb E\brac{\mb M} \mb P_{\mb x^\perp} \;=\; \mb P_{\mb x^\perp} \brac{\mb P_{\mb x^\perp} + \frac{4\sigma^2+1}{ 2\sigma^2 +1 } \mb x \mb x^*} \mb P_{\mb x^\perp} \;=\; \mb P_{\mb x^\perp}
\end{align*}

\end{proof}

\begin{lemma}\label{lem:moments-bound-1}
Let $\mb g \sim \mc {CN}(\mb 0,\mb I)$ and $\wt{\mb g}\sim \mc {CN}(\mb 0, 2\mb I)$, for any positive integer $p \geq 1$, we have
\begin{align*}
    \norm{ 1- \frac{1+2\sigma^2}{ m } \innerprod{\mb 1}{ \zeta_{\sigma^2}\paren{ \mb g \conv \mb x }  } }{L^p} \;&\leq\; \frac{ 3}{ \sqrt{m}}  \frac{\paren{ 2\sigma^2+1 } }{ \sigma } \norm{\mb C_{\mb x}}{}  \sqrt{p}, \\
    \norm{ 1- \frac{1+2\sigma^2}{ m } \innerprod{\mb 1}{ \eta_{\sigma^2}\paren{\wt{\mb g} \conv \mb x }  } }{L^p} \;&\leq\;  \frac{3}{\sqrt{m}} \frac{ \sigma^2 \paren{ 2\sigma^2 +1 }  }{ \paren{\sigma^2 - \frac{1}{2} }^{3/2}  } \norm{\mb C_{\mb x}}{}   \sqrt{p}.
\end{align*}

\end{lemma}

\begin{proof}
	Let $h(\mb g) = \frac{1+2\sigma^2}{m} \innerprod{ \mb 1 }{  \zeta_{\sigma^2}\paren{ \mb g \conv \mb x   } } -1 $ and let $h'(\wt{\mb g}) = \frac{1}{m} \innerprod{ \mb 1 }{  \eta_{\sigma^2}\paren{ \abs{ \wt{\mb g} \conv \mb x }  } } $, by Lemma \ref{lem:xi-eta-zeta-expectation}, we know that
	\begin{align*}
	   \bb E_{\mb g}\brac{ h(\mb g)  } \;=\; 0,\qquad \bb E_{\wt{\mb g}}\brac{ h'(\wt{\mb g})  } \;=\; 0.
	\end{align*}
    And for an independent copy $\mb g'$  of $\mb g$, we have
    \begin{align*}
       \abs{ h(\mb g) - h(\mb g') } \;&\leq\; \frac{1+2\sigma^2}{m} \abs{ \innerprod{ \mb 1 }{ \exp\paren{  - \frac{1}{2\sigma^2} \abs{\mb g \conv \mb x }^2  } - \exp\paren{  - \frac{1}{2\sigma^2} \abs{\mb g' \conv \mb x }^2  } }  } \\
       \;&\leq\; \frac{1+2\sigma^2 }{m} \norm{ \exp\paren{  - \frac{1}{2\sigma^2} \abs{\mb g \conv \mb x }^2  } - \exp\paren{  - \frac{1}{2\sigma^2} \abs{\mb g' \conv \mb x }^2  } }{1} \\
       \;&\leq\; \frac{1+2\sigma^2}{ \sqrt{m} \sigma } \norm{ \mb C_{\mb x}(\mb g - \mb g') }{} \leq \frac{ 1+2\sigma^2 }{ \sqrt{m} \sigma } \norm{\mb C_{\mb x}}{}\norm{ \mb g - \mb g' }{},
    \end{align*}
    where we used the fact that $\exp\paren{-\frac{x^2}{2\sigma^2 }}$ is $ \frac{1}{\sigma} e^{-1/2} $-Lipschitz. By applying \edited{the} Gaussian concentration inequality in Lemma \ref{lem:gauss-concentration}, we have
    \begin{align*}
      \bb P\paren{ \abs{h(\mb g)} \geq t } \;=\; \bb P \paren{ \abs{ \frac{1+2\sigma^2}{m} \innerprod{\mb 1}{  \zeta_{\sigma^2}\paren{ \abs{\mb g \conv \mb x} }  } - 1 }   \geq t }\;\leq\; \exp \paren{ - \frac{ \sigma^2 mt^2}{2 \paren{ 2\sigma^2+1 }^2 \norm{\mb C_{\mb x}}{}^2   } },
    \end{align*}
    for any scalar $t\geq 0$. Thus, we can see that $h(\mb g)$ is a centered $\frac{ \paren{\sigma^2+1}^2 \norm{\mb C_{\mb x}}{}^2 }{\sigma^2 m }$-subgaussian random variable, by Lemma \ref{lem:sub-Gaussian}, we know that for any positive $p\geq 1$
    \begin{align*}
       \norm{  h(\mb g) }{L^p} \;\leq\;  3\frac{\paren{2 \sigma^2+1 } \norm{\mb C_{\mb x}}{} }{ \sigma \sqrt{m} }  \sqrt{p},
    \end{align*}
    as desired. For $h'(\wt{\mb g})$, we can obtain the result similarly.
\end{proof}

%\end{appendices}

\end{document}